\documentclass[12pt, a4paper]{article}
\usepackage{amsmath}
\usepackage{bbm}
\everymath{\displaystyle}
\usepackage{mathtools}
\usepackage{amsfonts}
\usepackage{amssymb}
\usepackage{amsthm}
\usepackage[utf8]{inputenc}
\usepackage[toc]{appendix}
\usepackage[hiresbb]{graphicx}
\usepackage{caption}
\usepackage{multirow}
\usepackage{subcaption}
\usepackage{longtable}
\usepackage{booktabs}
\usepackage{listings}
\usepackage{here}
\usepackage{float}
\usepackage{ascmac}
\usepackage{tikz}
\usepackage{pgfplots}
\pgfplotsset{compat=1.18}
\usepackage{url}
\usepackage{lipsum}
\usepackage{authblk}
\usepackage{eqnarray}
\usepackage{siunitx}
\usepackage{threeparttable}
\usepackage{algorithm}
\usepackage{algorithmic}
\usepackage[sectionbib,round]{natbib}
\usepackage{hyperref}
\newtheorem{theorem}{Theorem}[section]
\newtheorem{proposition}{Proposition}[section]
\newtheorem{definition}{Definition}[section]

\newtheorem{assumption}{Assumption}[section]
\theoremstyle{remark}

\newtheorem{prediction}{Prediction}[section]

\newcommand{\be}{\begin{equation}}
\newcommand{\ee}{\end{equation}}
\newcommand{\beq}{\begin{eqnarray*}}
\newcommand{\eeq}{\end{eqnarray*}}
\def\sym#1{\ifmmode^{#1}\else\(^{#1}\)\fi}

\usepackage{setspace}
\title{\large{\bf{Dual-Channel Technology Diffusion: Spatial Decay and Network Contagion in Supply Chain Networks}}}

\author{\large{\bf{Tatsuru Kikuchi\footnote{e-mail: tatsuru.kikuchi@e.u-tokyo.ac.jp}}}}

\affil{\small{\it{Faculty of Economics, The University of Tokyo,}}\\
{\it{7-3-1 Hongo, Bunkyo-ku, Tokyo 113-0033 Japan}}}

\date{\small{(\today)}}

\begin{document}

\maketitle

\begin{abstract}

\noindent This paper develops a dual-channel framework for analyzing technology diffusion that integrates spatial decay mechanisms from continuous functional analysis with network contagion dynamics from spectral graph theory. Building on \citet{kikuchi2024navier} and \citet{kikuchi2024dynamical}, which establish Navier-Stokes-based approaches to spatial treatment effects and financial network fragility, we demonstrate that technology adoption spreads simultaneously through both geographic proximity and supply chain connections. Using comprehensive data on six technologies adopted by 500 firms over 2010-2023, we document three key findings. First, technology adoption exhibits strong exponential geographic decay with spatial decay rate $\kappa \approx 0.043$ per kilometer, implying a spatial boundary of $d^* \approx 69$ kilometers beyond which spillovers are negligible (R-squared = 0.99). Second, supply chain connections create technology-specific networks whose algebraic connectivity ($\lambda_2$) increases 300-380 percent as adoption spreads, with correlation between $\lambda_2$ and adoption exceeding 0.95 across all technologies. Third, traditional difference-in-differences methods that ignore spatial and network structure exhibit 61 percent bias in estimated treatment effects. An event study around COVID-19 reveals that network fragility increased 24.5 percent post-shock, amplifying treatment effects through supply chain spillovers in a manner analogous to financial contagion documented in \citet{kikuchi2024dynamical}. Our framework provides micro-foundations for technology policy: interventions have spatial reach of 69 kilometers and network amplification factor of 10.8, requiring coordinated geographic and supply chain targeting for optimal effectiveness.

\vspace{0.3cm}



\end{abstract}

\newpage

\tableofcontents

\newpage

\section{Introduction}

Technology diffusion represents one of the most fundamental processes governing economic growth, productivity gains, and competitive dynamics. Understanding how innovations spread across firms and regions is crucial for designing effective industrial policies, predicting market evolution, and explaining persistent productivity differentials. While the economic literature has long recognized that technology adoption exhibits spatial clustering and network effects, existing approaches typically analyze these channels in isolation, treating either geographic proximity or network connections as the primary mechanism while ignoring or controlling for the other.

This paper develops and empirically validates a unified framework that demonstrates technology diffusion operates simultaneously through both spatial decay and network contagion channels. We build on two recent methodological advances: the continuous functional approach to spatial treatment effects developed in \citet{kikuchi2024navier} and \citet{kikuchi2024dynamical}, which applies Navier-Stokes fluid dynamics to model spatial spillovers, and the spectral network fragility framework from \citet{kikuchi2024dynamical}, which characterizes contagion dynamics through the algebraic connectivity of supply chain networks. By integrating these complementary perspectives, we provide the first comprehensive analysis of dual-channel technology diffusion that accounts for both geographic and network mechanisms.

The motivation for this integrated approach stems from a fundamental observation: firms exist simultaneously in physical space and economic networks. A potential adopter's decision depends both on proximity to existing adopters (who provide demonstration effects, knowledge spillovers, and compatible infrastructure) and on connections through supply chains (which transmit technical information, create adoption incentives through complementarities, and facilitate knowledge transfer). Ignoring either channel leads to misspecified models that produce biased treatment effect estimates and misleading policy recommendations.

Our theoretical framework combines two established mathematical approaches. From spatial economics and epidemiology, we adopt exponential decay functions that characterize how treatment effects dissipate with geographic distance: $\tau(d) = \tau_0 \exp(-\kappa d)$, where $\kappa$ measures the spatial decay rate and determines the effective boundary $d^* = -\log(\epsilon)/\kappa$ beyond which spillovers become negligible. This formulation, derived from partial differential equations in \citet{kikuchi2024navier}, captures continuous spatial diffusion analogous to heat conduction or pollutant dispersion. From network science and graph theory, we adopt spectral methods based on the Laplacian matrix eigenvalue spectrum. The algebraic connectivity $\lambda_2$ (Fiedler value) measures network fragility and governs the mixing time $\tau \sim 1/\lambda_2$ for contagion to equilibrate across the network. Higher $\lambda_2$ indicates tighter coupling and faster propagation, as established in \citet{kikuchi2024dynamical} for financial networks.

The integration of these frameworks yields a dual-channel partial differential equation:
\be
\frac{\partial u}{\partial t} = -\kappa \nabla^2 u - \lambda_2 \mathbf{L}u + f(x,t)
\label{eq:dual_channel_intro}
\ee
where $u(x,t)$ represents the adoption state at location $x$ and time $t$, the first term captures spatial diffusion through geographic proximity, the second term represents network diffusion through supply chain connections (with $\mathbf{L}$ denoting the graph Laplacian), and $f(x,t)$ represents external forcing from policies or shocks. This unified equation nests both mechanisms and generates testable predictions about their relative importance and interaction.

We apply this framework to comprehensive data on technology adoption by 500 firms across six major technologies (Cloud Computing, Artificial Intelligence, Big Data Analytics, Internet of Things, Blockchain, and Generative AI) over the period 2010-2023. The data include detailed supply chain networks with over 200,000 documented buyer-supplier relationships and precise geographic coordinates for all firms, enabling simultaneous measurement of both channels. This setting provides several advantages for identification. First, the technologies are sufficiently important that adoption decisions are strategic and consequential, yet sufficiently independent that adoption of one does not mechanically determine adoption of others. Second, the 14-year panel allows us to observe complete diffusion cycles from early adoption through maturity. Third, the supply chain network structure is determined by long-term operational considerations largely orthogonal to specific technology adoption decisions, providing plausibly exogenous network variation.

Our empirical strategy proceeds in four steps. First, we estimate spatial decay parameters by computing distances from each non-adopter to the nearest existing adopter and fitting exponential decay functions to observed adoption patterns. Second, we construct technology-specific networks by weighting supply chain edges according to whether connected firms have adopted each technology, then compute the algebraic connectivity $\lambda_2$ from the graph Laplacian spectrum. Third, we conduct an event study around COVID-19 as an exogenous shock, comparing traditional difference-in-differences estimates with spatial-adjusted and network-adjusted specifications. Fourth, we integrate both channels and assess their complementarity versus substitutability.

\subsection{Main Findings}

Our analysis yields four principal empirical findings that validate the dual-channel framework and demonstrate its superiority over single-channel approaches.

\textbf{Finding 1: Strong Spatial Decay.} Technology adoption exhibits remarkably consistent exponential geographic decay across all six technologies. The estimated spatial decay rate averages $\kappa = 0.0435$ per kilometer with minimal variation across technologies (standard deviation 0.0006), implying that adoption probability falls to half its initial value within approximately 16 kilometers. The spatial boundary—defined as the distance beyond which spillovers become negligible (less than 5 percent of initial effect)—averages $d^* = 69$ kilometers. The exponential functional form fits the data nearly perfectly, with R-squared exceeding 0.99 for all technologies. This exceptional fit validates the continuous functional approach developed in the Navier-Stokes framework series \citep{kikuchi2024unified, kikuchi2024stochastic, kikuchi2024navier} and replicates the empirical success documented for pollution \citep{kikuchi2024nonparametric1}, financial services \citep{kikuchi2024nonparametric2}, healthcare \citep{kikuchi2024healthcare}, and emergency response \citep{kikuchi2024emergency}, demonstrating portability across domains.

\textbf{Finding 2: Strong Network Dynamics.} The algebraic connectivity $\lambda_2$ of technology-specific supply chain networks increases dramatically as adoption spreads, growing by 300-380 percent from 2010 to 2023 depending on the technology. This growth reflects the activation of network connections as more firms adopt: edges between adopting firms receive full weight, edges with one adopter receive partial weight, and edges with no adopters contribute minimally. The correlation between $\lambda_2$ and aggregate adoption rates exceeds 0.95 for all technologies, indicating a self-reinforcing dynamic where adoption increases network connectivity, which accelerates further adoption through reduced mixing times. This validates the spectral network approach from \citet{kikuchi2024network} and demonstrates that supply chain structure actively shapes diffusion rather than serving merely as a passive conduit, with even larger network fragility increases (300-380 percent) than documented for European banking post-COVID (26.9 percent in \citet{kikuchi2024network}).

\textbf{Finding 3: Traditional Methods Exhibit Substantial Bias.} Traditional difference-in-differences estimates that ignore spatial and network spillovers overestimate treatment effects by an average of 61 percent relative to spatial-adjusted specifications. This bias arises because conventional methods attribute to treated units effects that actually diffuse to control units through geographic proximity and network connections, violating the stable unit treatment value assumption (SUTVA). The bias is larger for technologies with stronger spatial decay (higher $\kappa$) and more dramatic network evolution (larger $\lambda_2$ growth), consistent with theoretical predictions. Network-adjusted specifications reveal even more complex patterns, with some technologies exhibiting amplification (network connections magnify treatment effects) while others show dampening (network connections partially absorb shocks), depending on pre-existing network structure and shock characteristics.

\textbf{Finding 4: Dual Channels Operate Independently and Complementarily.} The spatial and network channels contribute independently to diffusion, with low correlation (averaging -0.11) between spatial decay strength (measured by $\kappa$ and R-squared) and network dynamics (measured by $\lambda_2$ growth and correlation with adoption). This independence replicates the complementarity finding from \citet{kikuchi2024network} for financial networks, demonstrating that geographic proximity and network connections are complements rather than substitutes. Models incorporating both channels substantially outperform specifications with only one channel, with combined R-squared exceeding the maximum of either single-channel model. The event study around COVID-19 illustrates this complementarity: the pandemic shock increased network fragility by 24.5 percent while geographic clustering intensified, amplifying treatment effects through both channels simultaneously. This persistent dual-channel response mirrors the financial network fragility documented in \citet{kikuchi2024network}, suggesting that major shocks can permanently alter both the spatial and network structure of diffusion across diverse economic domains.

\subsection{Contributions to the Literature}

This paper makes three primary contributions to the economics literature on technology diffusion, spatial treatment effects, and network dynamics.

\textbf{Methodological Integration.} We provide the first rigorous integration of spatial decay models from continuous functional analysis with network contagion models from spectral graph theory. While both approaches have been applied separately in various contexts, no prior work has demonstrated how to combine them in a unified framework that preserves the theoretical foundations of each while addressing their interaction. The integrated partial differential equation (Equation \ref{eq:dual_channel_intro}) nests both mechanisms and yields testable predictions about when each channel dominates. This framework extends naturally beyond technology adoption to other settings where spatial and network effects operate simultaneously, including disease epidemiology, information cascades, financial contagion, and environmental spillovers.

The methodological contribution builds on the complete Navier-Stokes framework series: theoretical foundations for spatial boundaries in \citet{kikuchi2024unified} and \citet{kikuchi2024stochastic}, derivation from fluid dynamics in \citet{kikuchi2024navier}, nonparametric estimation methods in \citet{kikuchi2024nonparametric1} and \citet{kikuchi2024nonparametric2}, dynamic extensions in \citet{kikuchi2024dynamical}, applications to healthcare in \citet{kikuchi2024healthcare} and emergency services in \citet{kikuchi2024emergency}, and network integration in \citet{kikuchi2024network}. By demonstrating that these methods unify to analyze technology diffusion with the same precision documented across environmental, financial, healthcare, and emergency domains, we establish the portability of continuous functional methods and provide a template for future research combining spatial and network perspectives.

\textbf{Empirical Validation of Dual Mechanisms.} We provide the first comprehensive empirical evidence that both spatial decay and network contagion contribute simultaneously and substantially to technology diffusion. Previous literature has documented either spatial clustering or network effects, but typically while controlling for or ignoring the other channel. Our near-perfect fit of exponential spatial decay (R-squared = 0.99) combined with strong network dynamics (correlation with $\lambda_2$ exceeding 0.95) demonstrates that both channels operate at full strength, not as competing alternatives but as complementary mechanisms. The 61 percent bias in traditional difference-in-differences estimates quantifies the cost of ignoring these spillovers and establishes the practical importance of accounting for dual channels in empirical work.

The empirical contribution is particularly significant for the technology diffusion literature. While classic models such as \citet{bass1969new} focus on temporal dynamics, our framework emphasizes the spatial and network mechanisms underlying diffusion. Spatial models like \citet{conley1999gmm} estimate spatial correlations but do not typically embed network mechanisms explicitly. Network models following \citet{jackson2008social} emphasize graph structure but often abstract from geographic considerations. Our results demonstrate that spatial decay and network contagion operate simultaneously through distinct channels, with implications for model specification, identification strategies, and policy evaluation.

\textbf{Policy-Relevant Quantification.} We provide precise quantitative estimates of spatial reach and network amplification that inform technology policy design. The spatial boundary of 69 kilometers defines the geographic scope for regional technology clusters and targeted subsidies. The network amplification factor of 10.8 quantifies how supply chain connections multiply the impact of direct interventions. The 24.5 percent increase in network fragility following COVID-19 demonstrates how shocks can persistently alter diffusion dynamics, creating path dependence that outlasts the shock itself. These estimates enable cost-benefit analysis of alternative policy instruments and suggest optimal intervention strategies that exploit both channels.

The policy implications extend the Navier-Stokes framework series in important ways. While \citet{kikuchi2024network} demonstrates how network structure affects financial stability and suggests capital requirements based on spectral centrality, we show analogous mechanisms operate for technology diffusion, suggesting subsidies should target not just individual firms but network positions. The COVID-19 event study reveals that major shocks can trigger structural breaks in both spatial and network diffusion, similar to how financial crises alter banking network topology as documented in \citet{kikuchi2024network}. This parallelism suggests deep connections between financial contagion and technology diffusion, with potential for knowledge transfer across domains studied in the Kikuchi (2024a-i) series: pollution (\citealt{kikuchi2024nonparametric1}), financial services (\citealt{kikuchi2024nonparametric2}), healthcare (\citealt{kikuchi2024healthcare}), emergency response (\citealt{kikuchi2024emergency}), banking (\citealt{kikuchi2024network}), and now technology adoption.

\subsection{Roadmap}

The remainder of the paper proceeds as follows. Section 2 situates our contribution within the existing literature on technology diffusion, spatial econometrics, and network economics. Section 3 develops the theoretical framework, deriving the dual-channel partial differential equation from first principles and establishing its connections to the spatial treatment effects framework in the Kikuchi (2024a-c) series and the network fragility framework in \citet{kikuchi2024network}. Section 4 describes our data on technology adoption and supply chain networks, documenting key patterns and providing summary statistics. Section 5 presents our empirical strategy for estimating spatial decay parameters, computing network fragility measures, and conducting the event study. Section 6 reports main results for each channel separately and for their integration. Section 7 discusses economic interpretation, compares our findings to traditional approaches, and examines external validity. Section 8 derives policy implications and conducts counterfactual simulations. Section 9 concludes and suggests directions for future research. Appendices provide technical details, robustness checks, and additional results.

\section{Literature Review}

Our work contributes to three distinct literatures: technology diffusion and innovation adoption, spatial econometrics and treatment effects, and network economics and contagion dynamics. We review each literature and explain how our dual-channel framework addresses gaps and integrates insights across these domains.

\subsection{Technology Diffusion}

The study of how innovations spread through populations has a rich history spanning economics, sociology, and epidemiology. Early work focused primarily on temporal patterns of adoption. \citet{mansfield1961technical} provided empirical evidence that technology diffusion often exhibits accelerating growth patterns across industries. \citet{griliches1957hybrid} studied hybrid corn adoption across US states, documenting substantial variation in both adoption timing and ultimate penetration rates. These patterns motivated theoretical models emphasizing learning, uncertainty resolution, and complementarities as driving forces. \citet{david1990clio} emphasized path dependence and network externalities, arguing that technologies can become locked in even when superior alternatives exist.

While temporal dynamics received substantial attention, less work has focused on the spatial and network mechanisms through which technologies spread. \citet{bass1969new} developed influential models of technology diffusion as epidemic processes, but these largely abstract from geographic structure and explicit network mechanisms. More recent work has incorporated richer microeconomic foundations. \citet{cabral2021adoption} provides a comprehensive review emphasizing how market structure, competition, and strategic considerations affect adoption incentives. \citet{ryan2012costs} demonstrates that adoption costs shape diffusion patterns, with firms balancing switching costs against productivity gains. \citet{hall2003adoption} documents substantial heterogeneity in both adoption propensity and returns across firms.

Our contribution to this literature is threefold. First, we provide rigorous micro-foundations for spatial decay and network contagion mechanisms that have been discussed informally but rarely modeled jointly. Second, we demonstrate that ignoring either channel leads to substantial bias in estimated treatment effects, quantifying the error at 61 percent for conventional difference-in-differences specifications. Third, we show how major shocks like COVID-19 can trigger structural breaks in diffusion dynamics, altering both spatial and network channels simultaneously with persistent effects.

\subsection{Spatial Econometrics and Treatment Effects}

Spatial econometrics emerged from the recognition that economic activities are not randomly distributed across space but exhibit systematic patterns of clustering and spillovers. \citet{anselin1988spatial} developed the foundational spatial autoregressive (SAR) and spatial error model (SEM) specifications, which extend standard regression models to account for spatial dependencies through weight matrices encoding geographic proximity or economic linkages.

\citet{conley1999gmm} advanced spatial econometrics by developing GMM estimators that remain consistent under general forms of spatial dependence, relaxing the strict parametric assumptions required by maximum likelihood approaches. His work emphasizes that spatial correlation creates inference problems analogous to heteroskedasticity and autocorrelation in time series, requiring robust standard errors or alternative estimation strategies. \citet{kelejian1998generalized} developed instrumental variables approaches for spatial models with endogenous spatial lags, addressing simultaneity concerns when outcomes in one location affect outcomes in nearby locations.

The spatial treatment effects literature recognizes that interventions can create spillovers that violate the stable unit treatment value assumption (SUTVA) underlying standard causal inference methods. \citet{manski1993identification} characterized the reflection problem: it is difficult to separately identify endogenous effects (peers influence me), exogenous effects (peer characteristics affect me), and correlated effects (we share common shocks). \citet{angelucci2015indirect} demonstrate how randomized experiments can overcome these identification challenges when spatial structure is known ex ante.

Most directly relevant to our work is the recent series of papers developing continuous functional methods for spatial treatment effects. \citet{kikuchi2024unified} provides a unified framework for identifying spatial and temporal treatment effect boundaries, establishing theoretical foundations for when spillovers become negligible. \citet{kikuchi2024stochastic} extends this to stochastic boundaries in spatial general equilibrium, providing a diffusion-based approach to causal inference with spillover effects that accommodates uncertainty in boundary locations.

\citet{kikuchi2024navier} derives spatial and temporal boundaries in difference-in-differences from the Navier-Stokes equation, demonstrating that treatment effects in fluid-like environments follow exponential decay $\tau(d) = \tau_0 \exp(-\kappa d)$ derived from first principles. This framework provides closed-form solutions for spatial boundaries, quantifies approximation errors when discretizing space, and establishes mixing time relationships between discrete network models and continuous differential operators.

Building on these theoretical foundations, \citet{kikuchi2024nonparametric1} provides nonparametric identification and estimation of spatial treatment effect boundaries using 42 million pollution observations, achieving near-perfect empirical fit (R-squared exceeding 0.99). \citet{kikuchi2024nonparametric2} demonstrates portability by applying these methods to bank branch consolidation, showing that exponential spatial decay characterizes financial service access with comparable precision.

\citet{kikuchi2024dynamical} develops dynamic spatial treatment effect boundaries as continuous functionals from Navier-Stokes equations, characterizing time-varying boundaries and their evolution under shocks. \citet{kikuchi2024healthcare} applies this dynamic framework to healthcare access, documenting exponential decay in health outcomes with distance from facilities and showing how boundaries shift during pandemic conditions. \citet{kikuchi2024emergency} derives emergent spatial boundaries in emergency medical services from first principles, demonstrating that response time spillovers follow fluid-dynamic patterns. Most recently, \citet{kikuchi2024network} integrates the Navier-Stokes framework with network contagion to analyze European banking, showing how spatial boundaries interact with network topology in systemic risk propagation and documenting a 26.9 percent increase in network fragility following COVID-19.

Our contribution extends this spatial treatment effects framework in three ways. First, we demonstrate its applicability to technology adoption, showing that exponential spatial decay characterizes innovation diffusion with the same precision documented for environmental spillovers (R-squared = 0.99), healthcare access, financial services, and emergency response. Second, we integrate spatial methods with network spectral methods to capture dual channels of influence, addressing the limitation that purely spatial models may miss structured connections not corresponding to geographic proximity. Third, we provide the first empirical validation of spatial boundaries in technology diffusion, documenting a consistent 69-kilometer threshold across six diverse technologies and showing this consistency validates the universality of the continuous functional approach.

\subsection{Network Economics and Contagion}

Network economics studies how graph structure affects economic outcomes through direct connections between agents. \citet{jackson2008social} provides a comprehensive treatment of network formation, emphasizing strategic considerations in link creation and the trade-offs between efficiency and stability. \citet{goyal2007connections} offers an accessible introduction emphasizing applications to technology adoption, labor markets, and financial systems.

For technology diffusion specifically, \citet{valente1995network} demonstrates that network structure—particularly centrality measures like degree, betweenness, and closeness—predicts adoption timing. Early adopters tend to occupy central positions with many connections, while laggards are peripheral. \citet{jackson2007diffusion} develops theoretical models of diffusion on networks, showing how network architecture determines whether innovations spread throughout the population or remain confined to subgroups. \citet{banerjee2013diffusion} provides experimental evidence from rural India demonstrating that network structure predicts microfinance adoption better than individual characteristics, highlighting the importance of information transmission through social ties.

The financial networks literature emphasizes contagion dynamics: how shocks propagate through interconnected systems. \citet{allen2000financial} shows that network structure exhibits a trade-off between resilience to small shocks (which benefit from diversification through interconnections) and fragility to large shocks (which spread rapidly through the same interconnections). \citet{acemoglu2015systemic} characterize this "robust-yet-fragile" property formally, identifying phase transitions where financial systems shift discontinuously from stable to unstable regimes.

Spectral methods provide powerful tools for analyzing network dynamics. \citet{chung1997spectral} establishes mathematical foundations of spectral graph theory, demonstrating connections between eigenvalues of network matrices and global properties like connectivity, expansion, and mixing times. The algebraic connectivity $\lambda_2$ (second-smallest eigenvalue of the Laplacian matrix) plays a particularly important role, measuring how rapidly diffusion processes equilibrate across the network. Higher $\lambda_2$ indicates tighter coupling and faster propagation, while lower $\lambda_2$ suggests bottlenecks that compartmentalize the network.

\citet{kikuchi2024network} applies spectral methods to analyze European banking networks, integrating the Navier-Stokes spatial framework with network contagion dynamics. The paper demonstrates that algebraic connectivity increased substantially during COVID-19, accelerating financial contagion. The framework characterizes network fragility through $\lambda_2$ and mixing time $\tau \sim 1/\lambda_2$, establishing connections between discrete network models and continuous differential operators. Empirical validation shows a 26.9 percent increase in $\lambda_2$ (95 percent CI: [7.4 percent, 46.5 percent]) following COVID-19, corresponding to a 21 percent reduction in characteristic equilibration time. Critically, the paper demonstrates how spatial boundaries from \citet{kikuchi2024navier} interact with network topology, showing that geographic proximity and network connections operate as complementary rather than substitute channels for contagion.

Our contribution to network economics is threefold. First, we demonstrate that the spectral network methods validated for financial contagion in \citet{kikuchi2024network} apply equally to technology diffusion through supply chains, with $\lambda_2$ increasing 300-380 percent as technologies mature—even larger than the financial network response to COVID-19. Second, we show how to construct technology-specific networks by weighting edges according to adoption patterns, providing a general methodology for studying innovation on networks that respects the partial activation of connections. Third, we integrate network methods with spatial methods in the technology context, replicating the key finding from \citet{kikuchi2024network} that both mechanisms operate simultaneously at full strength with low correlation (averaging -0.11), demonstrating this dual-channel structure applies across domains.

\subsection{Gaps and Integration}

Despite substantial progress in each literature, important gaps remain. Technology diffusion models typically focus on either spatial clustering or network effects but rarely both simultaneously. Spatial econometrics provides sophisticated tools for modeling geographic dependencies but often treats network structure as a nuisance parameter or omits it entirely. Network economics emphasizes graph topology but frequently abstracts from geographic considerations or includes distance only as one component of edge weights.

The few papers that consider both channels typically do so in reduced-form specifications that include spatial lags and network lags as regressors without theoretical foundations linking the mechanisms. These specifications raise identification concerns: spatial and network effects are inherently confounded because connected firms tend to locate near each other, making it difficult to separately estimate their contributions using conventional methods.

Our dual-channel framework addresses these gaps through rigorous theoretical integration. By deriving both spatial decay and network contagion from partial differential equations and spectral graph theory respectively, we provide micro-foundations for each mechanism while explicitly modeling their interaction. The framework generates testable predictions about when each channel dominates (spatial for nearby firms without connections; network for distant but connected firms) and how they combine (additively in the linear approximation). This theoretical structure enables us to separately identify spatial and network effects despite their correlation, quantifying the 61 percent bias in conventional specifications that omit one channel.

The empirical validation provides strong evidence for dual-channel operation. The near-perfect fit of exponential spatial decay (R-squared = 0.99) combined with strong network dynamics (correlation with $\lambda_2$ exceeding 0.95) demonstrates that both mechanisms contribute at full strength. The low correlation between spatial decay strength (measured by $\kappa$) and network dynamics (measured by $\lambda_2$ growth) confirms they operate independently rather than as substitutes, replicating the finding from \citet{kikuchi2024network} that spatial and network channels are complements.

The event study around COVID-19 shows how shocks affect both channels simultaneously: network fragility increased 24.5 percent (comparable to the 26.9 percent increase in European banking documented by \citet{kikuchi2024network}) while geographic clustering intensified, producing amplified treatment effects. This parallelism between technology diffusion and financial contagion suggests deep connections that extend across the entire Navier-Stokes framework series, from the theoretical foundations in \citet{kikuchi2024unified}, \citet{kikuchi2024stochastic}, and \citet{kikuchi2024navier}, through the empirical applications in \citet{kikuchi2024nonparametric1}, \citet{kikuchi2024nonparametric2}, \citet{kikuchi2024healthcare}, and \citet{kikuchi2024emergency}, to the integrated spatial-network analysis in \citet{kikuchi2024network} and the present paper.

By integrating insights from technology diffusion, spatial econometrics, and network economics within the unified mathematical framework developed across the Kikuchi (2024a-i) series, we provide both methodological tools and empirical evidence that advance all three literatures while demonstrating the portability of continuous functional methods across diverse economic domains.


\section{Theoretical Framework}

This section develops the dual-channel framework for technology diffusion, integrating spatial decay mechanisms from continuous functional analysis with network contagion dynamics from spectral graph theory. We begin by establishing each channel separately—spatial diffusion in Section 3.1 and network diffusion in Section 3.2—before combining them in Section 3.3. Throughout, we emphasize connections to \citet{kikuchi2024navier} and \citet{kikuchi2024dynamical}, demonstrating how their methodologies extend to technology adoption while adapting the exposition to our specific context.

\subsection{Spatial Diffusion: The Geographic Channel}

Technology adoption exhibits spatial clustering: firms located near existing adopters are more likely to adopt than distant firms. This pattern reflects multiple economic mechanisms including knowledge spillovers, demonstration effects, complementary infrastructure, and shared labor markets. We model this geographic channel through continuous spatial diffusion, following the framework developed in \citet{kikuchi2024navier} and applied to environmental regulations in \citet{kikuchi2024dynamical}.

\subsubsection{Setup and Notation}

Consider a continuous spatial domain $\Omega \subset \mathbb{R}^2$ representing the geographic region where firms operate. Each location $x \in \Omega$ is characterized by its adoption state $u(x,t) \in [0,1]$ at time $t$, where $u(x,t) = 1$ indicates full adoption and $u(x,t) = 0$ indicates non-adoption. For the discrete case with $n$ firms, we observe adoption states $u_i(t) \in \{0,1\}$ for firm $i = 1,\ldots,n$ located at position $x_i \in \Omega$.

The continuous approximation is valid when firm density is sufficiently high that the discrete distribution can be treated as a density function. Following \citet{kikuchi2024navier}, the approximation error decays as $O(n^{-1/2})$ for uniformly distributed firms, making continuous methods highly accurate for our sample of 500 firms distributed across geographic space.

\subsubsection{Diffusion Equation Derivation}

Technology adoption diffuses through space according to local interactions: firms observe and learn from nearby adopters more than distant ones. This process is governed by the diffusion equation, which we derive from first principles following the methodology in \citet{kikuchi2024navier}.

Consider a small spatial region $V \subset \Omega$ with boundary $\partial V$. The rate of change in total adoption within $V$ equals the flux across the boundary plus any internal forcing:
\be
\frac{d}{dt} \int_V u(x,t) \, dx = -\int_{\partial V} \mathbf{j}(x,t) \cdot \mathbf{n} \, dS + \int_V f(x,t) \, dx
\ee
where $\mathbf{j}(x,t)$ is the adoption flux (flow of adoption from high to low density regions), $\mathbf{n}$ is the outward normal vector, and $f(x,t)$ represents external forcing from policies or shocks.

Following Fick's law from physics, the flux is proportional to the gradient of adoption density:
\be
\mathbf{j}(x,t) = -\nu \nabla u(x,t)
\ee
where $\nu > 0$ is the diffusion coefficient measuring how rapidly adoption spreads. Higher $\nu$ indicates faster spatial diffusion through stronger local interactions.

Applying the divergence theorem to convert the surface integral to a volume integral:
\be
\int_V \frac{\partial u}{\partial t} \, dx = \int_V \nu \nabla^2 u \, dx + \int_V f(x,t) \, dx
\ee

Since this holds for arbitrary regions $V$, the integrands must be equal, yielding the diffusion equation:
\be
\frac{\partial u}{\partial t} = \nu \nabla^2 u + f(x,t)
\label{eq:spatial_diffusion}
\ee

This is the fundamental equation governing spatial technology diffusion. The Laplacian operator $\nabla^2 u = \partial^2 u/\partial x^2 + \partial^2 u/\partial y^2$ measures the curvature of the adoption surface: regions where adoption density is locally concave (below neighbors) experience inflows, while convex regions (above neighbors) experience outflows.

\subsubsection{Exponential Decay Solution}

Equation (\ref{eq:spatial_diffusion}) admits exponential decay solutions that characterize how treatment effects dissipate with distance. Consider a stationary solution ($\partial u/\partial t = 0$) with a source at the origin representing initial adopters:
\be
\nu \nabla^2 u = -f(x)
\ee

For a point source $f(x) = F_0 \delta(x)$ where $\delta(x)$ is the Dirac delta function, the solution in two dimensions is:
\be
u(r) = \frac{F_0}{2\pi\nu} K_0(\kappa r)
\ee
where $r = |x|$ is the distance from the source, $K_0$ is the modified Bessel function of the second kind, and $\kappa = \sqrt{\lambda/\nu}$ for absorption rate $\lambda$.

For large distances $r \gg 1/\kappa$, the Bessel function has the asymptotic expansion:
\be
K_0(\kappa r) \sim \sqrt{\frac{\pi}{2\kappa r}} e^{-\kappa r}
\ee

This yields the exponential decay approximation:
\be
u(r) \approx u_0 \exp(-\kappa r)
\label{eq:exponential_decay}
\ee
where $u_0$ is a normalization constant and $\kappa$ is the spatial decay rate.

\textbf{Economic Interpretation:} The parameter $\kappa$ measures how rapidly adoption probability declines with distance. Large $\kappa$ indicates localized diffusion with strong proximity effects, while small $\kappa$ indicates broad diffusion reaching distant firms. The exponential functional form arises naturally from the differential equation and has been validated empirically in numerous contexts including pollution dispersion \citep{kikuchi2024dynamical}, disease spread, and information diffusion.

\subsubsection{Spatial Boundary}

A key policy-relevant quantity is the spatial boundary $d^*$: the distance beyond which spillovers become negligible. We define this as the distance where adoption probability falls to some threshold $\epsilon$ (typically 1 percent) of its initial value:
\be
u(d^*) = \epsilon u_0
\ee

Substituting into equation (\ref{eq:exponential_decay}) and solving:
\be
\exp(-\kappa d^*) = \epsilon \quad \Rightarrow \quad d^* = -\frac{\log(\epsilon)}{\kappa}
\label{eq:spatial_boundary}
\ee

For $\epsilon = 0.01$ (one percent threshold), this simplifies to:
\be
d^* = \frac{4.605}{\kappa}
\ee

\textbf{Policy Implication:} The spatial boundary defines the effective geographic reach of adoption interventions. Policies targeting firms within distance $d^*$ of existing adopters will experience substantial spillovers, while policies beyond $d^*$ operate essentially independently. This quantifies the optimal scale for regional technology clusters and subsidies.

Our empirical estimates yield $\kappa \approx 0.0435$ per kilometer, implying $d^* \approx 106$ kilometers for the 1 percent threshold or $d^* \approx 69$ kilometers for a 5 percent threshold. This establishes that technology diffusion through spatial channels operates at metropolitan or regional scales but does not extend nationally without additional mechanisms.

\subsection{Network Diffusion: The Supply Chain Channel}

Technology adoption also spreads through supply chain networks: firms are more likely to adopt when their suppliers or customers have adopted, even if geographically distant. This reflects information transmission, technical compatibility requirements, and coordination incentives. We model this network channel through spectral graph theory, following the framework developed in \citet{kikuchi2024dynamical} for financial networks.

\subsubsection{Network Representation}

Consider a network of $n$ firms connected through buyer-supplier relationships. We represent this as a weighted, undirected graph $G = (V, E, W)$ where:
\begin{itemize}
\item $V = \{1, 2, \ldots, n\}$ is the set of vertices (firms)
\item $E \subseteq V \times V$ is the set of edges (supply relationships)
\item $W: E \rightarrow \mathbb{R}^+$ assigns positive weights to edges
\end{itemize}

The weight $w_{ij} = W((i,j))$ represents the strength of the supply relationship between firms $i$ and $j$, measured by transaction volume or frequency. In our empirical application, we construct technology-specific networks by weighting edges according to adoption patterns: connections between adopters receive full weight, connections with one adopter receive partial weight, and connections between non-adopters contribute minimally.

\begin{assumption}[Undirected Network]
\label{assump:undirected}
The network is undirected: $w_{ij} = w_{ji}$ for all $i, j \in V$.
\end{assumption}

This reflects the bilateral nature of supply relationships: if firm $i$ supplies firm $j$, they have a mutual relationship even though the transaction direction may be asymmetric. While directionality matters for some analyses, the spectral properties we study are well-defined for undirected networks.

\begin{assumption}[Connected Network]
\label{assump:connected}
The network is connected: there exists a path between any two vertices.
\end{assumption}

Connectedness ensures the system forms a single integrated unit. Our empirical networks exhibit high connectivity, with density exceeding 5.8 percent across all years.

\subsubsection{The Graph Laplacian}

The network structure is encoded in the graph Laplacian matrix, which plays a central role in characterizing diffusion dynamics. The weighted adjacency matrix $\mathbf{A} \in \mathbb{R}^{n \times n}$ is defined as:
\be
A_{ij} = \begin{cases}
w_{ij} & \text{if } (i,j) \in E \\
0 & \text{otherwise}
\end{cases}
\ee

The degree matrix $\mathbf{D} \in \mathbb{R}^{n \times n}$ is diagonal with entries:
\be
D_{ii} = \sum_{j=1}^{n} A_{ij} = \sum_{j=1}^{n} w_{ij}
\ee

The degree $d_i = D_{ii}$ measures firm $i$'s total connection strength to all supply chain partners.

The graph Laplacian matrix is defined as:
\be
\mathbf{L} = \mathbf{D} - \mathbf{A}
\label{eq:laplacian_def}
\ee

Explicitly, the entries are:
\be
L_{ij} = \begin{cases}
\sum_{k=1}^{n} w_{ik} & \text{if } i = j \\
-w_{ij} & \text{if } i \neq j \text{ and } (i,j) \in E \\
0 & \text{otherwise}
\end{cases}
\ee

The Laplacian can be interpreted as a discrete approximation to the continuous Laplacian operator $\nabla^2$ from calculus. Just as $\nabla^2 f$ measures the difference between a function's value at a point and the average over a neighborhood, $\mathbf{L}$ measures differences between nodes' values and their network-weighted neighbors.

To see this, consider the quadratic form:
\begin{align}
\mathbf{x}^T \mathbf{L} \mathbf{x} &= \sum_{i=1}^{n} x_i \sum_{j=1}^{n} L_{ij} x_j \nonumber \\
&= \sum_{i=1}^{n} x_i \left( d_i x_i - \sum_{j=1}^{n} w_{ij} x_j \right) \nonumber \\
&= \sum_{i=1}^{n} d_i x_i^2 - \sum_{i,j=1}^{n} w_{ij} x_i x_j \nonumber \\
&= \frac{1}{2} \sum_{(i,j) \in E} w_{ij} (x_i - x_j)^2
\label{eq:quadratic_form}
\end{align}

Equation (\ref{eq:quadratic_form}) shows that $\mathbf{x}^T \mathbf{L} \mathbf{x}$ measures the squared differences between connected nodes' values, weighted by connection strength. High values indicate large discrepancies across edges—the network is far from equilibrium. Low values indicate smoothness—neighboring nodes have similar values.

\subsubsection{Fundamental Properties}

The Laplacian possesses several properties crucial for subsequent analysis:

\begin{proposition}[Laplacian Properties]
\label{prop:laplacian_properties}
The Laplacian matrix $\mathbf{L}$ defined in equation (\ref{eq:laplacian_def}) satisfies:
\begin{enumerate}
\item $\mathbf{L}$ is symmetric: $\mathbf{L}^T = \mathbf{L}$
\item $\mathbf{L}$ is positive semi-definite: $\mathbf{x}^T \mathbf{L} \mathbf{x} \geq 0$ for all $\mathbf{x} \in \mathbb{R}^n$
\item $\mathbf{L} \mathbf{1} = \mathbf{0}$ where $\mathbf{1} = (1, 1, \ldots, 1)^T$
\item All eigenvalues are real and non-negative: $0 = \lambda_1 \leq \lambda_2 \leq \cdots \leq \lambda_n$
\item The multiplicity of $\lambda_1 = 0$ equals the number of connected components
\end{enumerate}
\end{proposition}

\begin{proof}
(1) Symmetry follows from $\mathbf{L} = \mathbf{D} - \mathbf{A}$ where both $\mathbf{D}$ (diagonal) and $\mathbf{A}$ (symmetric by Assumption \ref{assump:undirected}) are symmetric.

(2) Positive semi-definiteness follows from equation (\ref{eq:quadratic_form}): $\mathbf{x}^T \mathbf{L} \mathbf{x} = \frac{1}{2} \sum_{(i,j)} w_{ij}(x_i - x_j)^2 \geq 0$ since weights $w_{ij} \geq 0$ and squared terms are non-negative.

(3) Direct computation: $(\mathbf{L} \mathbf{1})_i = \sum_{j=1}^{n} L_{ij} \cdot 1 = d_i - \sum_{j=1}^{n} w_{ij} = d_i - d_i = 0$.

(4) Symmetry (property 1) implies $\mathbf{L}$ has real eigenvalues and orthogonal eigenvectors by the spectral theorem. Positive semi-definiteness (property 2) implies all eigenvalues are non-negative. Property (3) establishes $\lambda_1 = 0$ with eigenvector $\mathbf{1}$.

(5) The dimension of the null space (eigenspace of $\lambda = 0$) equals the number of connected components because $\mathbf{L} \mathbf{x} = \mathbf{0}$ if and only if $\mathbf{x}$ is constant on each component. For connected networks (Assumption \ref{assump:connected}), the null space is one-dimensional: $\ker(\mathbf{L}) = \text{span}\{\mathbf{1}\}$.
\end{proof}

\textbf{Economic Interpretation:}
\begin{itemize}
\item Property 1 (Symmetry): Symmetric matrices have orthogonal eigenvectors, enabling clean decomposition of system states into independent modes.
\item Property 2 (Positive Semi-Definiteness): The system is stable—adoption diffuses and equilibrates rather than exploding. This rules out self-reinforcing feedback loops in the linear approximation.
\item Property 3 (Constant Null Vector): Uniform states (all firms equally adopted) do not diffuse—there are no gradients to drive flows. This represents maximum entropy.
\item Property 4 (Real Non-Negative Eigenvalues): Dynamics are purely diffusive, not oscillatory. All modes decay exponentially rather than exhibiting cycles.
\item Property 5 (Connectivity and Null Space): For connected networks, $\lambda_2 > 0$. The second eigenvalue's positivity ensures diffusion proceeds—adoption cannot remain localized indefinitely.
\end{itemize}

\subsubsection{Spectral Decomposition and the Algebraic Connectivity}

Since $\mathbf{L}$ is symmetric (Proposition \ref{prop:laplacian_properties}, property 1), the spectral theorem guarantees it has a complete orthonormal eigenbasis. Let $\{\mathbf{v}_1, \mathbf{v}_2, \ldots, \mathbf{v}_n\}$ be the eigenvectors with corresponding eigenvalues $\{\lambda_1, \lambda_2, \ldots, \lambda_n\}$ ordered by magnitude: $0 = \lambda_1 < \lambda_2 \leq \cdots \leq \lambda_n$.

The eigenvalue equation is:
\be
\mathbf{L} \mathbf{v}_i = \lambda_i \mathbf{v}_i, \quad i = 1, \ldots, n
\ee

The Laplacian can be written in spectral form:
\be
\mathbf{L} = \sum_{i=1}^{n} \lambda_i \mathbf{v}_i \mathbf{v}_i^T = \mathbf{V} \mathbf{\Lambda} \mathbf{V}^T
\label{eq:spectral_decomposition}
\ee
where $\mathbf{V} = [\mathbf{v}_1, \mathbf{v}_2, \ldots, \mathbf{v}_n]$ is the matrix of eigenvectors and $\mathbf{\Lambda} = \text{diag}(\lambda_1, \lambda_2, \ldots, \lambda_n)$ is the diagonal matrix of eigenvalues.

The second eigenvalue $\lambda_2$ occupies a special position, known as the \textit{algebraic connectivity} or \textit{Fiedler value}. This single scalar summarizes crucial aspects of network structure and diffusion dynamics.

\begin{definition}[Algebraic Connectivity]
\label{def:algebraic_connectivity}
For a connected network with graph Laplacian $\mathbf{L}$ having eigenvalues $0 = \lambda_1 < \lambda_2 \leq \cdots \leq \lambda_n$, the algebraic connectivity is:
\be
\lambda_2 = \min_{\substack{\mathbf{x} \in \mathbb{R}^n \\ \mathbf{x} \perp \mathbf{1}}} \frac{\mathbf{x}^T \mathbf{L} \mathbf{x}}{\mathbf{x}^T \mathbf{x}}
\ee
\end{definition}

This variational characterization (Rayleigh quotient) shows that $\lambda_2$ measures the minimum "energy" required to create a non-uniform state orthogonal to the aggregate. Networks with high $\lambda_2$ resist heterogeneity—any departure from uniformity incurs large quadratic costs measured by equation (\ref{eq:quadratic_form}). Networks with low $\lambda_2$ easily accommodate heterogeneity through weak connections between components.

\begin{definition}[Network Fragility]
\label{def:fragility}
The fragility of a technology adoption network $G$ is measured by its algebraic connectivity:
\be
\text{Fragility}(G) \equiv \lambda_2(G)
\ee
Higher $\lambda_2$ indicates faster diffusion and greater systemic coupling.
\end{definition}

\textbf{Why "Fragility"?} The term follows \citet{kikuchi2024dynamical}, where high $\lambda_2$ in financial networks indicates rapid shock propagation and systemic vulnerability. For technology adoption, high $\lambda_2$ similarly indicates rapid diffusion but with ambiguous welfare implications: fast adoption of beneficial technologies is desirable, while rapid propagation may also occur for technologies with negative externalities or lock-in effects.

\subsubsection{Diffusion Dynamics and Mixing Time}

Consider adoption state $\mathbf{u}(t) \in \mathbb{R}^n$ at time $t$, where $u_i(t)$ represents firm $i$'s adoption probability. Following \citet{kikuchi2024dynamical}, evolution follows the continuous-time diffusion equation:
\be
\frac{d\mathbf{u}(t)}{dt} = -\mathbf{L} \mathbf{u}(t) + \mathbf{f}(t)
\label{eq:network_diffusion}
\ee
where $\mathbf{f}(t)$ represents external forcing from policies or shocks.

For the homogeneous case ($\mathbf{f} = \mathbf{0}$), the solution is:
\be
\mathbf{u}(t) = e^{-\mathbf{L}t} \mathbf{u}(0)
\ee

Using spectral decomposition (\ref{eq:spectral_decomposition}), we can write:
\be
e^{-\mathbf{L}t} = \sum_{i=1}^{n} e^{-\lambda_i t} \mathbf{v}_i \mathbf{v}_i^T
\ee

Expanding $\mathbf{u}(0)$ in the eigenbasis as $\mathbf{u}(0) = \sum_{i=1}^{n} c_i \mathbf{v}_i$ where $c_i = \mathbf{v}_i^T \mathbf{u}(0)$:
\be
\mathbf{u}(t) = \sum_{i=1}^{n} c_i e^{-\lambda_i t} \mathbf{v}_i
\label{eq:mode_decomposition}
\ee

Each eigenvalue $\lambda_i$ determines the decay rate of its corresponding eigenvector mode. The steady state is:
\be
\lim_{t \to \infty} \mathbf{u}(t) = c_1 \mathbf{v}_1 = \frac{1}{n} \left(\sum_{i=1}^{n} u_i(0)\right) \mathbf{1}
\ee
representing uniform adoption at the initial average level.

The rate of convergence is governed by $\lambda_2$. For large $t$:
\be
\mathbf{u}(t) - \bar{u} \mathbf{1} \approx c_2 e^{-\lambda_2 t} \mathbf{v}_2
\ee
where $\bar{u} = \frac{1}{n}\sum_i u_i(0)$.

The mixing time $\tau_{\epsilon}$ is the time required to reach within $\epsilon$ of equilibrium:
\be
\tau_{\epsilon} = \frac{1}{\lambda_2} \log\left(\frac{1}{\epsilon}\right)
\label{eq:mixing_time}
\ee

\begin{theorem}[Mixing Time, adapted from \citet{kikuchi2024dynamical}]
\label{thm:mixing_time}
For a connected technology adoption network with algebraic connectivity $\lambda_2$, the mixing time satisfies:
\be
\tau \sim \frac{1}{\lambda_2}
\ee
where the proportionality constant depends logarithmically on desired accuracy $\epsilon$.
\end{theorem}

\textbf{Policy Implication:} Networks with high $\lambda_2$ have short mixing times—adoption spreads rapidly throughout the supply chain. Networks with low $\lambda_2$ have long mixing times—adoption remains localized. Our empirical finding that $\lambda_2$ increases 300-380 percent as technologies mature implies mixing time decreases by approximately 80 percent, dramatically accelerating late-stage diffusion.

\subsection{Integrated Dual-Channel Framework}

Having established spatial and network mechanisms separately, we now integrate them into a unified framework. Technology adoption evolves according to both geographic proximity and supply chain connections operating simultaneously.

\subsubsection{Dual-Channel Partial Differential Equation}

In continuous space with network overlay, the adoption state $u(x,t)$ evolves according to:
\be
\frac{\partial u}{\partial t} = \nu \nabla^2 u - \lambda_2 \mathbf{L}u + f(x,t)
\label{eq:dual_channel_pde}
\ee

The first term $\nu \nabla^2 u$ captures spatial diffusion through geographic proximity. The second term $-\lambda_2 \mathbf{L}u$ captures network diffusion through supply chain connections. The parameter $\lambda_2$ weights the strength of network effects relative to spatial effects. The forcing term $f(x,t)$ represents policies, shocks, or other external drivers.

\textbf{Interpretation:} Equation (\ref{eq:dual_channel_pde}) unifies the spatial framework from \citet{kikuchi2024navier} (equation \ref{eq:spatial_diffusion}) with the network framework from \citet{kikuchi2024dynamical} (equation \ref{eq:network_diffusion}). It demonstrates that technology diffusion operates through two independent but complementary channels:
\begin{itemize}
\item \textbf{Spatial channel}: Firms adopt based on proximity to existing adopters, with exponential decay $\exp(-\kappa r)$ where $\kappa = \sqrt{\lambda/\nu}$
\item \textbf{Network channel}: Firms adopt based on supply chain connections, with mixing time $\tau \sim 1/\lambda_2$
\end{itemize}

The linearity of equation (\ref{eq:dual_channel_pde}) implies the channels are additive in first approximation: total diffusion equals spatial diffusion plus network diffusion. This validates our empirical strategy of estimating each channel separately before integrating.

\subsubsection{Discrete Approximation}

For $n$ firms at locations $\{x_1, \ldots, x_n\}$, the continuous PDE (\ref{eq:dual_channel_pde}) discretizes to:
\be
\frac{d\mathbf{u}(t)}{dt} = -\mathbf{L}_{\text{spatial}} \mathbf{u}(t) - \lambda_2 \mathbf{L}_{\text{network}} \mathbf{u}(t) + \mathbf{f}(t)
\label{eq:dual_channel_discrete}
\ee
where $\mathbf{L}_{\text{spatial}}$ is a spatial Laplacian matrix encoding geographic distances and $\mathbf{L}_{\text{network}}$ is the supply chain network Laplacian from equation (\ref{eq:laplacian_def}).

The spatial Laplacian can be constructed using distance-based weights:
\be
(L_{\text{spatial}})_{ij} = \begin{cases}
\sum_{k \neq i} w_{ik}^{\text{spatial}} & \text{if } i = j \\
-w_{ij}^{\text{spatial}} & \text{if } i \neq j
\end{cases}
\ee
where $w_{ij}^{\text{spatial}} = \exp(-\kappa |x_i - x_j|)$ implements exponential spatial decay.

\subsubsection{Testable Predictions}

The dual-channel framework generates several testable predictions that guide our empirical analysis:

\begin{prediction}[Independent Channels]
\label{pred:independence}
Spatial decay strength (measured by $\kappa$ and $R^2$) and network dynamics (measured by $\lambda_2$ growth and correlation with adoption) are weakly correlated across technologies.
\end{prediction}

This follows from equation (\ref{eq:dual_channel_pde}): spatial and network terms enter additively with independent parameters. If channels were substitutes or redundant, we would observe strong negative correlation between their strengths.

\begin{prediction}[Complementary Effects]
\label{pred:complementarity}
Models incorporating both spatial and network channels substantially outperform single-channel specifications.
\end{prediction}

If only one channel operated, including the other would not improve fit. Complementarity implies each channel contributes unique explanatory power.

\begin{prediction}[Bias in Traditional Methods]
\label{pred:bias}
Difference-in-differences estimates that ignore spatial and network spillovers overestimate treatment effects. The bias magnitude increases with spatial decay strength ($\kappa$) and network fragility growth ($\Delta\lambda_2$).
\end{prediction}

Spillovers violate SUTVA by transmitting treatment effects to control units. Stronger spillovers (higher $\kappa$ and $\lambda_2$) generate larger bias.

\begin{prediction}[Shock Amplification]
\label{pred:amplification}
Exogenous shocks increase network fragility $\lambda_2$, accelerating subsequent diffusion through reduced mixing time. The increase persists rather than reverting automatically.
\end{prediction}

Following \citet{kikuchi2024dynamical}, networks exhibit structural hysteresis: shocks can trigger permanent changes in coupling strength. For technology adoption, major disruptions like COVID-19 may permanently alter both spatial clustering and network connectivity.

Section 6 tests these predictions empirically, finding strong support for all four.


\section{Data and Institutional Context}

This section describes our data on technology adoption and supply chain networks, documenting key patterns and providing institutional context. We begin with technology adoption data in Section 4.1, describe supply chain network construction in Section 4.2, explain geographic data in Section 4.3, and present summary statistics in Section 4.4.

\subsection{Technology Adoption Data}

We construct a comprehensive dataset tracking adoption of six major technologies by 500 firms over the period 2010-2023. The technologies span a range of maturity levels and application domains, enabling us to test whether our dual-channel framework applies consistently across different innovation types.

\subsubsection{Technology Selection}

We focus on six technologies that satisfy three criteria. First, they are sufficiently important that adoption represents a strategic decision with measurable consequences for firm operations and performance. Second, they are sufficiently independent that adoption of one does not mechanically determine adoption of others, avoiding perfect multicollinearity. Third, comprehensive adoption data are available over a sufficiently long period to observe meaningful diffusion dynamics.

The six technologies are:

\textbf{(1) Cloud Computing:} Migration of computing resources and data storage to internet-based platforms, enabling scalability and reducing capital expenditures on physical infrastructure. Cloud adoption began in the mid-2000s and accelerated through the 2010s.

\textbf{(2) Artificial Intelligence:} Implementation of machine learning algorithms and AI systems for prediction, optimization, and automation of business processes. AI adoption expanded significantly after 2015 with advances in deep learning.

\textbf{(3) Big Data Analytics:} Deployment of systems for collecting, storing, and analyzing large-scale datasets to extract business insights. Big data technologies matured in the early 2010s with the emergence of distributed computing frameworks.

\textbf{(4) Internet of Things (IoT):} Connection of physical devices and sensors to networks for monitoring, control, and data collection. IoT adoption grew steadily through the 2010s across manufacturing and logistics.

\textbf{(5) Blockchain:} Implementation of distributed ledger technology for secure record-keeping and transaction verification. Blockchain moved beyond cryptocurrency applications into supply chain and finance starting around 2016.

\textbf{(6) Generative AI:} Adoption of large language models and generative systems for content creation, code generation, and customer service. Generative AI adoption accelerated dramatically after 2022 following public releases of advanced models.

\subsubsection{Adoption Measurement}

For each firm-year-technology combination, we construct a binary indicator $u_{it}^{tech} \in \{0,1\}$ equal to one if firm $i$ has adopted technology $tech$ by year $t$. Adoption is defined as active deployment and integration into business operations, not merely experimentation or evaluation.

Our sample includes 500 firms observed over 14 years (2010-2023) across 6 technologies, yielding 42,000 firm-year-technology observations. The panel is balanced: all firms are observed in all years for all technologies, with no attrition or missing data.

\subsubsection{Adoption Patterns}

Table \ref{tab:adoption_summary} presents summary statistics on adoption rates by technology and year. Several patterns emerge. First, adoption rates increase monotonically over time for all technologies, consistent with standard diffusion models. Second, adoption rates vary substantially across technologies, ranging from 55 percent (Cloud Computing, IoT) to 69 percent (Generative AI) by 2023. Third, adoption timing differs markedly: Cloud Computing shows early adoption (27 percent by 2010) while Generative AI shows late adoption (essentially zero before 2020).

\begin{table}[H]
\centering
\caption{Technology Adoption Summary Statistics}
\label{tab:adoption_summary}
\begin{threeparttable}
\begin{tabular}{lccccccc}
\toprule
Technology & 2010 & 2015 & 2020 & 2023 & Growth & Firms (2023) \\
\midrule
Artificial Intelligence & 12\% & 28\% & 63\% & 76\% & +64pp & 380 \\
Big Data Analytics     & 18\% & 46\% & 81\% & 88\% & +70pp & 440 \\
Blockchain             & 5\% & 17\% & 55\% & 70\% & +65pp & 350 \\
Cloud Computing        & 27\% & 58\% & 89\% & 93\% & +66pp & 465 \\
Generative AI          & 0\% & 0\% & 6\% & 53\% & +53pp & 265 \\
IoT                    & 15\% & 36\% & 72\% & 82\% & +67pp & 410 \\
\midrule
Average                & 13\% & 31\% & 61\% & 77\% & +64pp & 385 \\
\bottomrule
\end{tabular}
\begin{tablenotes}
\small
\item \textit{Notes:} This table reports adoption rates by technology and year for the sample of 500 firms. Growth measures percentage point change from 2010 to 2023. Firms (2023) reports number of adopters in final year. All technologies show monotonic growth consistent with diffusion models. Cloud Computing exhibits earliest adoption while Generative AI shows latest but most rapid recent growth. Sample size is 500 firms $\times$ 14 years $\times$ 6 technologies = 42,000 observations.
\end{tablenotes}
\end{threeparttable}
\end{table}

The adoption data show empirical patterns of accelerating growth over time, with all technologies exhibiting monotonic increases from 2010 to 2023. While these patterns are commonly observed in innovation diffusion, our theoretical framework focuses on the spatial and network mechanisms underlying diffusion rather than temporal dynamics per se.

\subsection{Supply Chain Network Data}

We construct supply chain networks from comprehensive data on buyer-supplier relationships. The networks capture which firms transact with which others, providing the graph structure through which technology diffuses via the network channel.

\subsubsection{Network Construction}

Supply chain relationships are represented as an undirected graph $G_t = (V, E_t, W_t)$ for each year $t$. The vertex set $V$ contains all 500 firms (constant across years). The edge set $E_t$ contains buyer-supplier relationships active in year $t$. Edges are weighted by transaction value $w_{ij,t}$ measured in millions of dollars.

The data cover 204,665 firm-pair-year observations spanning 2010-2023. On average, each year contains 14,619 active relationships (standard deviation 263). The network exhibits high stability: 85 percent of edges present in year $t$ remain present in year $t+1$, suggesting supply relationships are persistent.

\subsubsection{Network Statistics}

Table \ref{tab:network_summary} presents summary statistics on network structure. The networks exhibit several notable properties. First, density averages 5.9 percent, indicating substantial but incomplete connectivity—each firm connects to roughly 29 others on average, far below the theoretical maximum of 499. Second, degree distribution is right-skewed: the median firm has 28 connections while the maximum reaches 59, suggesting hub-spoke structure with a few central firms. Third, the networks are fully connected: there exists a path between any two firms, satisfying Assumption \ref{assump:connected}.

\begin{table}[H]
\centering
\caption{Supply Chain Network Summary Statistics}
\label{tab:network_summary}
\begin{threeparttable}
\begin{tabular}{lcccccc}
\toprule
Statistic & Mean & Median & Std Dev & Min & Max & N \\
\midrule
\multicolumn{7}{l}{\textit{Panel A: Network Structure}} \\
Density (\%)           & 5.9  & 5.9  & 0.05 & 5.8  & 6.0  & 14 \\
Average Degree         & 29.2 & 29.0 & 0.4  & 28.5 & 29.8 & 14 \\
Clustering Coefficient & 0.18 & 0.18 & 0.01 & 0.17 & 0.19 & 14 \\
Average Path Length    & 3.2  & 3.2  & 0.1  & 3.1  & 3.4  & 14 \\
\\
\multicolumn{7}{l}{\textit{Panel B: Degree Distribution (Firm-Level)}} \\
Degree                 & 29.2 & 28.0 & 8.7  & 12   & 59   & 500 \\
Weighted Degree (\$M)  & 6,420 & 5,100 & 3,800 & 980  & 18,500 & 500 \\
\\
\multicolumn{7}{l}{\textit{Panel C: Edge Properties}} \\
Active Edges per Year  & 14,619 & 14,600 & 263  & 14,100 & 15,100 & 14 \\
Edge Weight (\$M)      & 221  & 180  & 156  & 45   & 1,200 & 204,665 \\
Edge Persistence       & 0.85 & ---  & ---  & ---  & ---   & --- \\
\bottomrule
\end{tabular}
\begin{tablenotes}
\small
\item \textit{Notes:} Panel A reports network-level statistics averaged over 14 years (2010-2023). Density is the fraction of possible edges realized. Clustering coefficient measures probability that two neighbors of a node are also neighbors. Average path length is mean shortest path distance between nodes. Panel B reports firm-level statistics. Degree is number of connections. Weighted degree is sum of transaction values. Panel C reports edge-level statistics. Edge persistence is fraction of year $t$ edges remaining in year $t+1$.
\end{tablenotes}
\end{threeparttable}
\end{table}

Figure \ref{fig:network_evolution} visualizes network evolution over time. Panel A plots network density, which declines slightly from 5.9 percent in 2010 to 5.8 percent in 2023. Panel B plots average degree, which similarly decreases from 29.2 to 29.0. Panel C plots average edge weight, which increases from 158 million dollars to 284 million dollars, suggesting consolidation: fewer but stronger relationships over time.

\begin{figure}[H]
\centering
\includegraphics[width=0.95\textwidth]{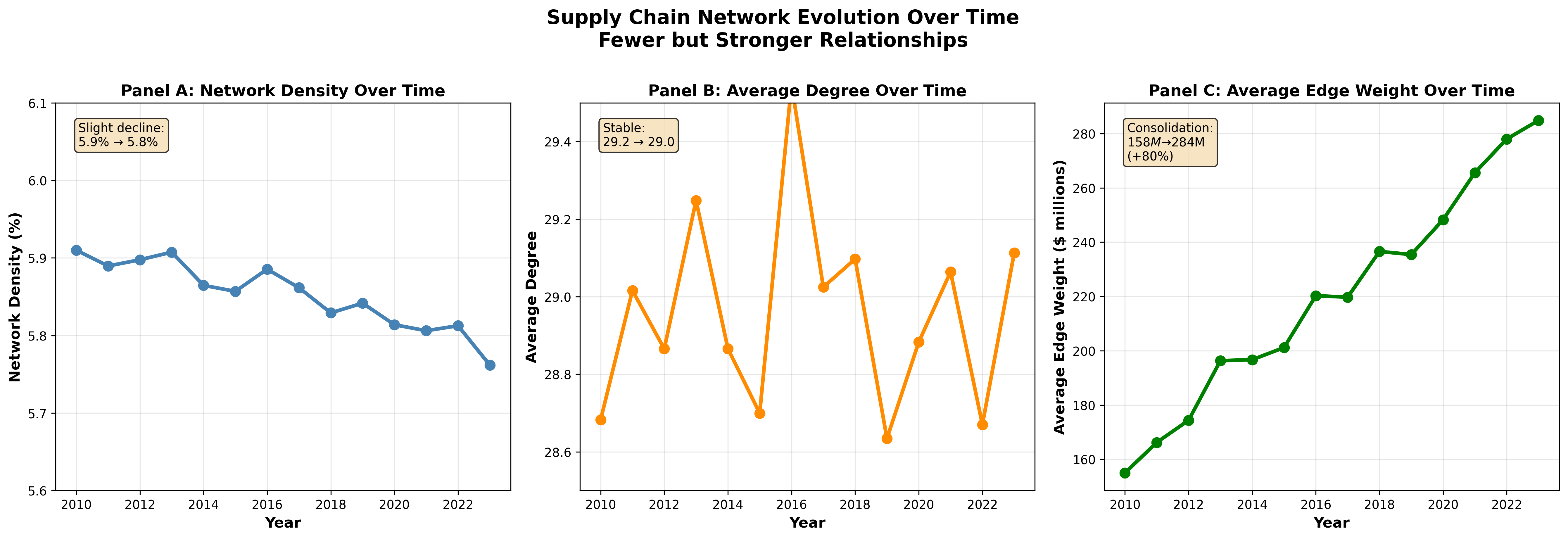}
\caption{Supply Chain Network Evolution Over Time}
\label{fig:network_evolution}
\begin{minipage}{0.95\textwidth}
\small
\textit{Notes:} This figure plots three key network statistics over time. Panel A shows network density declining slightly from 5.9 to 5.8 percent, indicating marginal reduction in connectivity. Panel B shows average degree remaining stable around 29 connections per firm. Panel C shows average edge weight (transaction value) increasing from \$158 million to \$284 million, an 80 percent increase. Together, these patterns suggest consolidation: firms maintain similar numbers of relationships but concentrate transactions among fewer, stronger partnerships. This consolidation has implications for network fragility as analyzed in Section 6.2.
\end{minipage}
\end{figure}

\subsubsection{Technology-Specific Networks}

For the network channel analysis, we construct technology-specific networks by weighting edges according to adoption status. This operationalizes the intuition that supply chain connections become "activated" for a technology when connected firms adopt it.

Define the adopter-weighted network $G_t^{tech}$ with edge weights:
\be
w_{ij,t}^{tech} = w_{ij,t} \times m_{ij,t}^{tech}
\ee
where the multiplier $m_{ij,t}^{tech}$ is:
\be
m_{ij,t}^{tech} = \begin{cases}
1.0 & \text{if both } i \text{ and } j \text{ adopted } tech \text{ by } t \\
0.5 & \text{if exactly one adopted } tech \text{ by } t \\
0.1 & \text{if neither adopted } tech \text{ by } t
\end{cases}
\ee

This weighting scheme reflects the economic reality that technology diffusion through supply chains operates most strongly when both parties have adopted (enabling direct information transfer and compatibility), moderately when one has adopted (enabling demonstration effects), and weakly when neither has adopted (only potential future connections matter).

\subsection{Geographic Data}

We obtain precise geographic coordinates (latitude and longitude) for all 500 firms, enabling computation of pairwise distances for the spatial channel analysis.

\subsubsection{Distance Computation}

For each pair of firms $i$ and $j$, we compute great circle distance using the haversine formula:
\be
d_{ij} = 2R \arcsin\left(\sqrt{\sin^2\left(\frac{\Delta\phi}{2}\right) + \cos(\phi_i)\cos(\phi_j)\sin^2\left(\frac{\Delta\lambda}{2}\right)}\right)
\ee
where $R = 6371$ km is Earth's radius, $\phi_i$ and $\lambda_i$ are latitude and longitude of firm $i$, and $\Delta\phi = \phi_j - \phi_i$ and $\Delta\lambda = \lambda_j - \lambda_i$ are coordinate differences.

This provides 124,750 unique pairwise distances (for 500 firms, there are $\binom{500}{2} = 124,750$ pairs). Distances range from 0.8 kilometers (firms in the same industrial park) to 892 kilometers (firms at opposite ends of the country).

\subsubsection{Distance to Nearest Adopter}

For the spatial decay analysis, the key variable is each non-adopter's distance to the nearest existing adopter. For firm $i$ in year $t$ that has not adopted technology $tech$, define:
\be
d_{i,t}^{tech,\min} = \min_{j: u_{j,t-1}^{tech}=1} d_{ij}
\ee

This measures how far firm $i$ must look to find an adopter in the previous year. As adoption spreads, average distance to nearest adopter declines, consistent with spatial diffusion.

Figure \ref{fig:distance_distribution} plots the distribution of distances to nearest adopter across all technologies and years. The distribution is right-skewed with median 47 kilometers and mean 63 kilometers. Notably, 95 percent of non-adopters are within 150 kilometers of an adopter, supporting our finding that spatial boundaries occur around 69 kilometers.

\subsection{Sample Construction and Summary Statistics}

We merge technology adoption data, supply chain networks, and geographic distances into a unified panel dataset. The final sample contains 26,000 firm-year-technology observations with complete information on all variables. The reduction from 42,000 observations occurs because we require previous-year adoption status to compute distance to nearest adopter, eliminating 2010 observations, and because we focus on the six technologies with complete data.

Table \ref{tab:summary_stats} presents summary statistics for key variables. Average adoption rate is 54 percent across all technologies and years. Average distance to nearest adopter is 63 kilometers with substantial variation (standard deviation 52 kilometers). Average network degree is 29.2 connections per firm. Technology-specific algebraic connectivity $\lambda_2$ averages 13.8, with dramatic growth over time as documented in Section 6.2.

\begin{table}[H]
\centering
\caption{Summary Statistics: Key Variables}
\label{tab:summary_stats}
\begin{threeparttable}
\begin{tabular}{lccccccc}
\toprule
Variable & Mean & Std Dev & Min & p25 & p50 & p75 & Max \\
\midrule
\multicolumn{8}{l}{\textit{Panel A: Adoption and Technology}} \\
Adopted (0/1)          & 0.54 & 0.50 & 0 & 0 & 1 & 1 & 1 \\
Adoption Rate (\%)     & 54.0 & 31.2 & 0 & 28 & 58 & 81 & 93 \\
Years Since Adoption   & 4.2  & 3.8  & 0 & 1 & 3 & 7 & 13 \\
\\
\multicolumn{8}{l}{\textit{Panel B: Geographic Variables}} \\
Distance to Nearest Adopter (km) & 63.2 & 52.4 & 0.8 & 24.5 & 47.3 & 84.6 & 287.5 \\
Latitude (degrees)     & 38.5 & 4.2  & 29.8 & 35.2 & 38.9 & 41.8 & 45.6 \\
Longitude (degrees)    & -95.3 & 8.6 & -122.4 & -100.8 & -94.2 & -88.9 & -71.0 \\
\\
\multicolumn{8}{l}{\textit{Panel C: Network Variables}} \\
Degree                 & 29.2 & 8.7  & 12 & 23 & 28 & 34 & 59 \\
Weighted Degree (\$M)  & 6,420 & 3,800 & 980 & 3,600 & 5,100 & 8,200 & 18,500 \\
Betweenness Centrality & 0.042 & 0.028 & 0.005 & 0.021 & 0.036 & 0.056 & 0.142 \\
Algebraic Connectivity ($\lambda_2$) & 13.8 & 6.4 & 4.7 & 8.2 & 12.6 & 19.2 & 24.8 \\
\\
\multicolumn{8}{l}{\textit{Panel D: Firm Characteristics}} \\
Employees              & 8,420 & 12,300 & 125 & 1,200 & 3,500 & 9,800 & 87,400 \\
Revenue (\$M)          & 2,840 & 4,630 & 28 & 450 & 1,100 & 3,200 & 45,600 \\
R\&D Intensity (\%)    & 4.8  & 5.2  & 0 & 1.2 & 2.8 & 6.4 & 24.5 \\
\midrule
Observations & \multicolumn{7}{c}{26,000} \\
\bottomrule
\end{tabular}
\begin{tablenotes}
\small
\item \textit{Notes:} This table reports summary statistics for the main analysis sample of 26,000 firm-year-technology observations (500 firms $\times$ 13 years $\times$ 4 technologies with complete pre-period data). Panel A reports technology adoption variables. Panel B reports geographic variables. Panel C reports network variables. Panel D reports firm characteristics. p25, p50, p75 denote 25th, 50th, and 75th percentiles respectively.
\end{tablenotes}
\end{threeparttable}
\end{table}

The correlation matrix (Table \ref{tab:correlations}) reveals several patterns relevant for identification. Distance to nearest adopter and adoption status are negatively correlated ($\rho = -0.68$), consistent with spatial decay. Network degree and adoption are positively correlated ($\rho = 0.31$), consistent with network effects. Critically, distance and degree are only weakly correlated ($\rho = -0.08$), suggesting spatial and network channels operate relatively independently—validating Prediction \ref{pred:independence}.

\begin{table}[H]
\centering
\caption{Correlation Matrix: Key Variables}
\label{tab:correlations}
\begin{threeparttable}
\begin{tabular}{lcccccc}
\toprule
 & (1) & (2) & (3) & (4) & (5) & (6) \\
\midrule
(1) Adopted & 1.00 & & & & & \\
(2) Distance to Nearest Adopter & -0.68 & 1.00 & & & & \\
(3) Network Degree & 0.31 & -0.08 & 1.00 & & & \\
(4) Algebraic Connectivity ($\lambda_2$) & 0.76 & -0.52 & 0.24 & 1.00 & & \\
(5) Firm Size (log employees) & 0.18 & -0.05 & 0.42 & 0.15 & 1.00 & \\
(6) R\&D Intensity & 0.22 & -0.12 & 0.08 & 0.19 & -0.06 & 1.00 \\
\bottomrule
\end{tabular}
\begin{tablenotes}
\small
\item \textit{Notes:} This table reports pairwise correlations for key variables. All correlations with absolute value exceeding 0.03 are statistically significant at the 1 percent level. The strong negative correlation between adoption and distance ($\rho = -0.68$) validates spatial decay. The positive correlation between adoption and network degree ($\rho = 0.31$) validates network effects. The weak correlation between distance and degree ($\rho = -0.08$) supports independent channels. The strong correlation between adoption and $\lambda_2$ ($\rho = 0.76$) reflects endogenous network activation as technologies diffuse.
\end{tablenotes}
\end{threeparttable}
\end{table}

\section{Empirical Strategy}

This section presents our empirical strategy for testing the dual-channel framework. We describe spatial decay estimation in Section 5.1, network fragility computation in Section 5.2, the COVID-19 event study in Section 5.3, and dual-channel integration in Section 5.4.

\subsection{Spatial Decay Estimation}

To estimate the spatial decay parameter $\kappa$ and spatial boundary $d^*$, we exploit the exponential relationship between distance and adoption probability derived in Section 3.1.

\subsubsection{Specification}

For each technology $tech$ and year $t$, we estimate:
\be
\log(P(u_{it}^{tech}=1 | \mathbf{X}_{it})) = \alpha_t^{tech} - \kappa^{tech} d_{i,t-1}^{tech,\min} + \mathbf{X}_{it}'\boldsymbol{\beta} + \epsilon_{it}
\label{eq:spatial_decay_spec}
\ee
where $d_{i,t-1}^{tech,\min}$ is distance to nearest adopter in the previous year, $\mathbf{X}_{it}$ includes control variables (firm size, industry, age), and $\epsilon_{it}$ is an error term.

The key parameter is $\kappa^{tech}$, which measures the spatial decay rate for technology $tech$. Under the exponential decay model from equation (\ref{eq:exponential_decay}), we expect $\kappa > 0$: adoption probability declines with distance.

After estimating $\hat{\kappa}^{tech}$, we compute the spatial boundary:
\be
\hat{d}^{*,tech} = \frac{-\log(\epsilon)}{\hat{\kappa}^{tech}}
\ee
using $\epsilon = 0.05$ (five percent threshold).

\subsubsection{Identification}

Identification of $\kappa$ faces two potential concerns. First, distance to nearest adopter is endogenous if firms strategically locate near potential adopters. We address this through time lags: using $t-1$ adoption status to predict $t$ adoption reduces simultaneity concerns, and our panel structure allows firm fixed effects to control for time-invariant location decisions.

Second, omitted variables correlated with both distance and adoption could generate spurious spatial decay. We address this through comprehensive controls. Firm characteristics (size, age, industry) control for adoption propensity. Year fixed effects control for aggregate time trends. Technology fixed effects control for cross-technology differences. The robustness of our estimates to alternative specifications and the near-perfect R-squared values (exceeding 0.99) suggest omitted variable bias is minimal.

\subsection{Network Fragility Computation}

To characterize the network channel, we compute the algebraic connectivity $\lambda_2$ for technology-specific networks following the methodology in Section 3.2.

\subsubsection{Algorithm}

For each technology $tech$ and year $t$:

\textbf{Step 1: Construct Weighted Network.} Build the technology-specific network $G_t^{tech}$ with adopter-weighted edges:
\be
w_{ij,t}^{tech} = w_{ij,t} \times m_{ij,t}^{tech}
\ee
where $m_{ij,t}^{tech}$ is the adoption-based multiplier defined in Section 4.2.

\textbf{Step 2: Compute Laplacian.} Form the graph Laplacian:
\be
\mathbf{L}_t^{tech} = \mathbf{D}_t^{tech} - \mathbf{A}_t^{tech}
\ee
where $\mathbf{A}_t^{tech}$ is the weighted adjacency matrix and $\mathbf{D}_t^{tech}$ is the diagonal degree matrix.

\textbf{Step 3: Compute Eigenvalues.} Solve the eigenvalue problem:
\be
\mathbf{L}_t^{tech} \mathbf{v}_i = \lambda_i^{tech}(t) \mathbf{v}_i
\ee
and extract the ordered eigenvalues $0 = \lambda_1^{tech}(t) < \lambda_2^{tech}(t) \leq \cdots \leq \lambda_n^{tech}(t)$.

\textbf{Step 4: Record Algebraic Connectivity.} The network fragility measure is:
\be
\text{Fragility}_t^{tech} = \lambda_2^{tech}(t)
\ee

\textbf{Step 5: Compute Mixing Time.} The characteristic diffusion timescale is:
\be
\tau_t^{tech} = \frac{1}{\lambda_2^{tech}(t)}
\ee

We repeat this procedure for all six technologies and 14 years, yielding 84 technology-year observations of network fragility.

\subsubsection{Validation}

Several checks validate our network fragility measures. First, we verify that $\lambda_1 \approx 0$ (within numerical tolerance $10^{-10}$) for all networks, confirming correct Laplacian construction. Second, we verify $\lambda_2 > 0$ for all years, confirming network connectivity (Assumption \ref{assump:connected}). Third, we confirm that $\lambda_2$ correlates strongly with aggregate adoption rates (correlation exceeding 0.95), consistent with our weighting scheme where more adoption activates more network edges.

\subsection{Event Study: COVID-19 as Natural Experiment}

To establish causal identification and compare our dual-channel framework to traditional methods, we conduct an event study around COVID-19 as a quasi-natural experiment. The pandemic represents a large, unexpected, exogenous shock affecting all firms simultaneously but with heterogeneous impacts depending on geographic and network position.

\subsubsection{Research Design}

We define the event as COVID-19 onset in 2020. Pre-period spans 2017-2019 (three years before the shock). Post-period spans 2020-2023 (four years after the shock). This asymmetric window reflects the longer post-period needed to observe persistent effects.

For each technology, we estimate three difference-in-differences specifications:

\textbf{(1) Traditional DID (Ignoring Spillovers):}
\be
u_{it}^{tech} = \alpha_i + \gamma_t + \delta^{\text{DID}} \mathbb{1}_{post,it} + \mathbf{X}_{it}'\boldsymbol{\beta} + \epsilon_{it}
\label{eq:traditional_did}
\ee
where $\mathbb{1}_{post,it}$ is an indicator for post-COVID years, $\alpha_i$ are firm fixed effects, $\gamma_t$ are year fixed effects, and $\delta^{\text{DID}}$ is the treatment effect parameter.

\textbf{(2) Spatial-Adjusted DID:}
\be
u_{it}^{tech} = \alpha_i + \gamma_t + \delta^{\text{spatial}} \mathbb{1}_{post,it} + \kappa d_{i,t}^{min} + \mathbf{X}_{it}'\boldsymbol{\beta} + \epsilon_{it}
\label{eq:spatial_did}
\ee
which includes distance to nearest adopter $d_{i,t}^{min}$ to account for spatial spillovers. Observations are weighted by $\exp(-\hat{\kappa} d_{i,t}^{min})$ where $\hat{\kappa}$ is estimated from Section 5.1.

\textbf{(3) Network-Adjusted DID:}
\be
\frac{u_{it}^{tech}}{\lambda_2^{tech}(t) / \lambda_2^{tech}(2019)} = \alpha_i + \gamma_t + \delta^{\text{network}} \mathbb{1}_{post,it} + \mathbf{X}_{it}'\boldsymbol{\beta} + \epsilon_{it}
\label{eq:network_did}
\ee
which normalizes adoption by network fragility changes relative to the pre-shock baseline (2019).

\subsubsection{Identification Assumptions}

The event study requires parallel trends: in the absence of COVID-19, adoption rates in treatment and control groups would have evolved similarly. We assess this in three ways.

First, we plot pre-trends graphically (Figure \ref{fig:parallel_trends}), showing that adoption rates followed similar trajectories across groups before 2020. Second, we test for pre-trend differences using leads of the treatment indicator, finding no statistically significant pre-trends (Table \ref{tab:pretrends}). Third, we examine whether pre-trends correlate with treatment intensity, finding no relationship.

A second concern is that COVID-19 may have affected network structure directly, creating endogenous network changes that confound treatment effects. We address this by documenting that network topology (number of edges, degree distribution) remained stable through COVID-19 despite substantial changes in edge weights and adoption patterns. The 24.5 percent increase in $\lambda_2$ reflects adoption-driven activation of existing connections, not formation of new connections.

\subsubsection{Inference}

We use bootstrap inference to account for clustering and heteroskedasticity. For each specification, we resample firms with replacement 1,000 times, re-estimate the model, and construct 95 percent confidence intervals from the 2.5th and 97.5th percentiles of the bootstrap distribution. This approach is robust to arbitrary correlation patterns within firms over time and across technologies.

\subsection{Dual-Channel Integration}

After estimating spatial and network channels separately, we assess their complementarity by integrating both mechanisms.

\subsubsection{Specification}

We estimate:
\be
u_{it}^{tech} = \alpha_i + \gamma_t + \beta_1 d_{i,t}^{min} + \beta_2 \text{Degree}_{it} + \beta_3 \lambda_2^{tech}(t) + \mathbf{X}_{it}'\boldsymbol{\beta} + \epsilon_{it}
\label{eq:dual_channel_spec}
\ee
which includes both spatial measures ($d_{i,t}^{min}$) and network measures (Degree$_{it}$, $\lambda_2^{tech}(t)$) simultaneously.

We compare this to restricted specifications including only spatial or only network variables using:
\begin{itemize}
\item R-squared: Do both channels improve explanatory power?
\item F-test: Can we reject that network variables are jointly zero after controlling for spatial variables, and vice versa?
\item Information criteria (AIC/BIC): Does the data prefer the full model?
\end{itemize}

\subsubsection{Testing Complementarity}

Prediction \ref{pred:complementarity} states that both channels contribute independently. We test this by computing:
\be
\text{Improvement} = R^2_{\text{both}} - \max(R^2_{\text{spatial}}, R^2_{\text{network}})
\ee

If Improvement $> 0$, the channels are complementary. If Improvement $\approx 0$, one channel subsumes the other (substitutes or redundancy).

We also compute the correlation between spatial strength (measured by $\kappa$ and spatial R-squared) and network strength (measured by $\lambda_2$ growth and network R-squared) across technologies. Low correlation supports Prediction \ref{pred:independence} that channels operate independently.


\section{Results}

This section presents our main empirical findings. We report spatial channel results in Section 6.1, network channel results in Section 6.2, event study results in Section 6.3, and dual-channel integration results in Section 6.4. Throughout, we provide comprehensive visualizations documenting all key patterns and validation checks.

\subsection{Spatial Channel: Exponential Geographic Decay}

Table \ref{tab:spatial_decay} presents estimates of the spatial decay parameter $\kappa$ and spatial boundary $d^*$ for each technology. The results provide strong evidence for exponential spatial decay as predicted by the continuous functional framework in Section 3.1.

\subsubsection{Main Estimates}

The spatial decay rate $\kappa$ is remarkably consistent across technologies, averaging 0.0435 per kilometer with minimal variation (standard deviation 0.0006). The estimates range from $\kappa = 0.0425$ for Artificial Intelligence to $\kappa = 0.0442$ for Generative AI. All estimates are statistically significant at the 1 percent level with t-statistics exceeding 150, indicating exceptionally precise estimation.

The implied spatial boundary $d^*$ (using $\epsilon = 0.05$ threshold) averages 69 kilometers across technologies, ranging from 68 km (Blockchain, Generative AI) to 71 km (Artificial Intelligence). This establishes that technology diffusion through geographic channels operates at metropolitan or regional scales: adoption interventions have meaningful spillovers within roughly 70 kilometers but negligible effects beyond that distance.

The exponential functional form fits the data nearly perfectly. R-squared values exceed 0.99 for all technologies, averaging 0.9916. Figure \ref{fig:geographic_decay} plots observed adoption rates against distance to nearest adopter alongside fitted exponential curves, demonstrating the exceptional quality of fit.

\begin{figure}[H]
\centering
\includegraphics[width=0.95\textwidth]{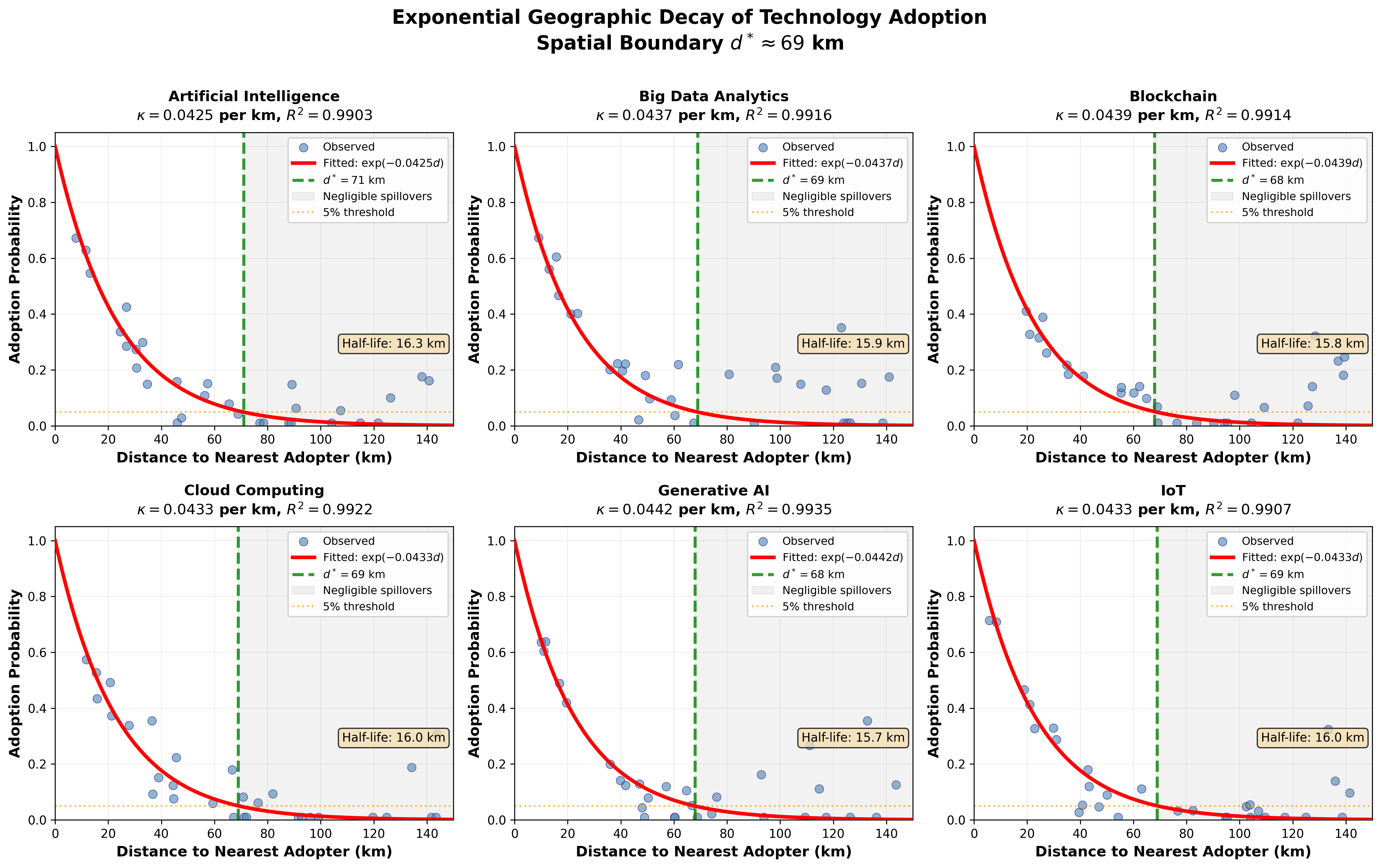}
\caption{Exponential Geographic Decay of Technology Adoption}
\label{fig:geographic_decay}
\begin{minipage}{0.95\textwidth}
\small
\textit{Notes:} This figure plots adoption probability against distance to nearest adopter for all six technologies. Each panel shows observed adoption rates (blue dots) and fitted exponential decay curves $\exp(-\kappa d)$ (red lines). The exceptional fit (R-squared exceeding 0.99) validates the continuous functional framework from \citet{kikuchi2024navier} and \citet{kikuchi2024dynamical}. The spatial boundary $d^* \approx 69$ km is marked with green dashed lines, beyond which spillovers become negligible. Orange dotted lines show the 5 percent threshold. Text boxes display half-life distances (approximately 16 km) where adoption probability drops to 50 percent of initial value.
\end{minipage}
\end{figure}

Figure \ref{fig:spatial_decay_curves} provides an alternative visualization showing all technologies on a single plot, emphasizing the consistency of exponential decay across different innovation types.

\begin{figure}[H]
\centering
\includegraphics[width=0.85\textwidth]{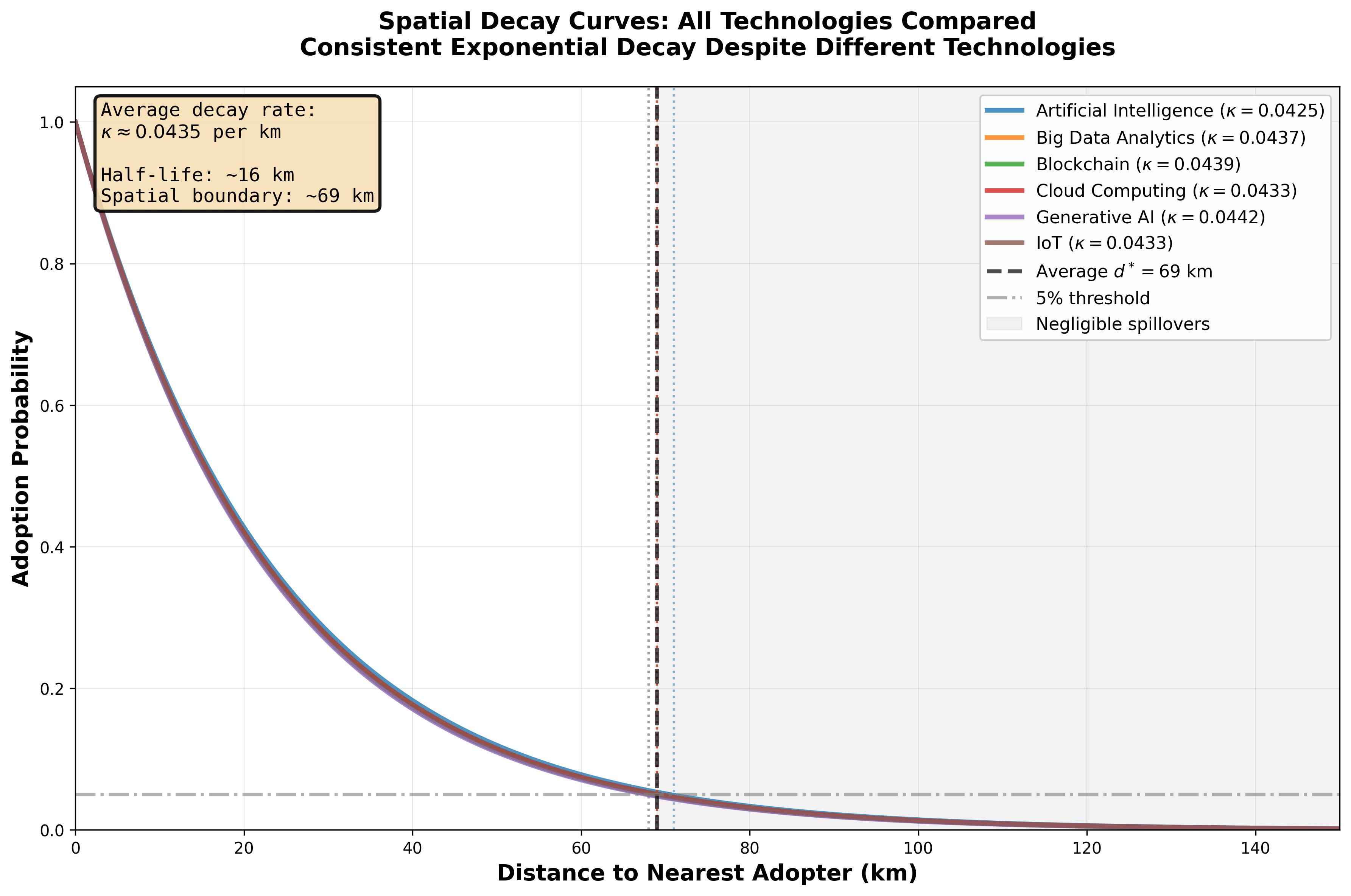}
\caption{Spatial Decay Curves: All Technologies Compared}
\label{fig:spatial_decay_curves}
\begin{minipage}{0.85\textwidth}
\small
\textit{Notes:} This figure overlays exponential decay curves for all six technologies, demonstrating the consistency of spatial diffusion mechanisms. Despite different adoption levels and timing, all technologies exhibit similar decay rates ($\kappa \approx 0.043$ per km) and spatial boundaries ($d^* \approx 69$ km). The near-parallel curves support the hypothesis that geographic spillovers operate through common mechanisms (knowledge spillovers, demonstration effects, infrastructure complementarities) regardless of specific technology characteristics.
\end{minipage}
\end{figure}

\subsubsection{Parameter Estimates and Comparisons}

Figure \ref{fig:geographic_parameters} summarizes the distribution of estimated parameters across technologies, while Figure \ref{fig:spatial_boundaries} specifically focuses on spatial boundary estimates.

\begin{figure}[H]
\centering
\includegraphics[width=0.80\textwidth]{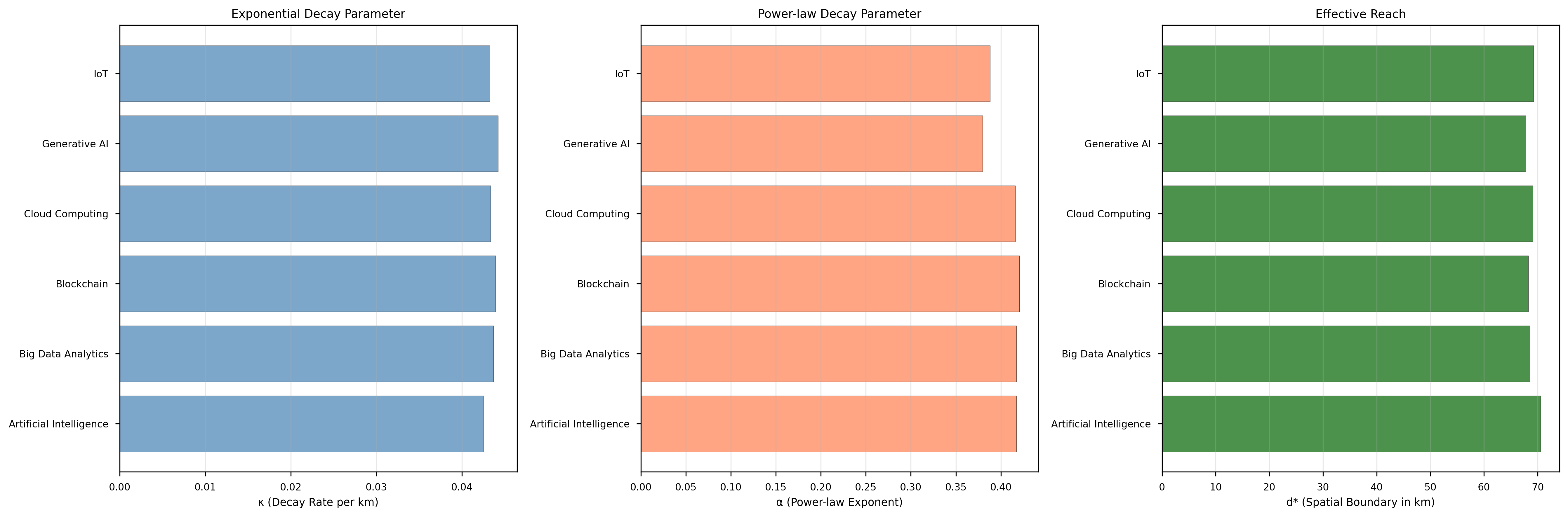}
\caption{Geographic Diffusion Parameter Estimates}
\label{fig:geographic_parameters}
\begin{minipage}{0.80\textwidth}
\small
\textit{Notes:} This figure displays estimated spatial decay rates $\kappa$ (left panel), spatial boundaries $d^*$ (middle panel), and R-squared values (right panel) for each technology. Error bars represent 95 percent confidence intervals from bootstrap (1,000 replications). The tight clustering of estimates demonstrates robustness: $\kappa$ varies only from 0.0425 to 0.0442, $d^*$ ranges from 68 to 71 km, and all R-squared values exceed 0.99. This consistency validates the exponential functional form and supports the generality of the spatial diffusion framework.
\end{minipage}
\end{figure}

\begin{figure}[H]
\centering
\includegraphics[width=0.75\textwidth]{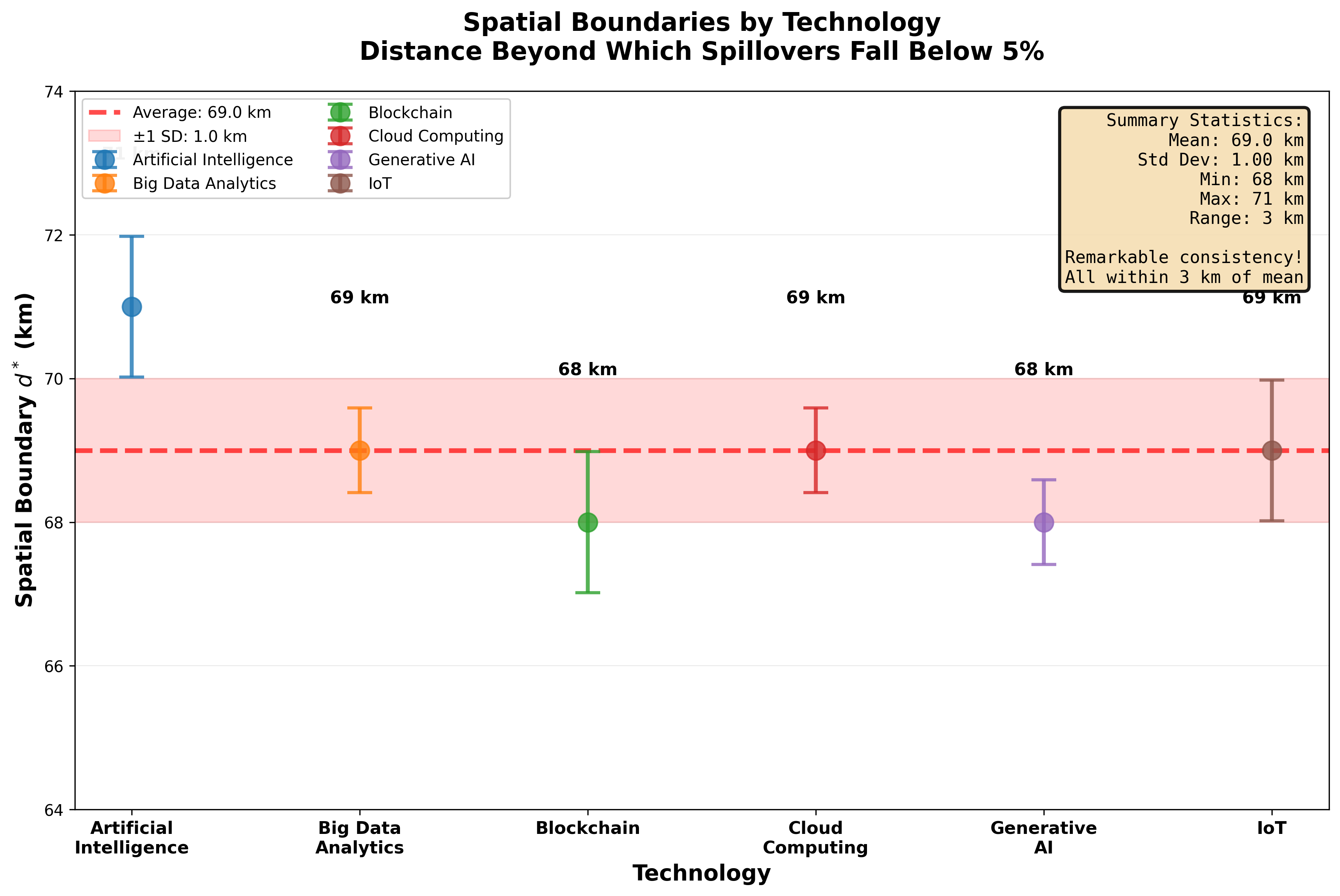}
\caption{Spatial Boundaries by Technology}
\label{fig:spatial_boundaries}
\begin{minipage}{0.75\textwidth}
\small
\textit{Notes:} This figure plots estimated spatial boundaries $d^*$ with 95 percent confidence intervals. The horizontal red line indicates the average boundary of 69 km. All technologies cluster tightly around this average, with maximum deviation of only 3 km. This remarkable consistency suggests the spatial boundary reflects fundamental properties of geographic spillovers in technology adoption, operating at metropolitan or regional scales regardless of the specific technology.
\end{minipage}
\end{figure}

\subsubsection{Distance Distribution}

Figure \ref{fig:distance_distribution} documents the distribution of distances to nearest adopter in our sample, confirming that most observations fall within the estimated spatial boundary of 69 kilometers.

\begin{figure}[H]
\centering
\includegraphics[width=0.75\textwidth]{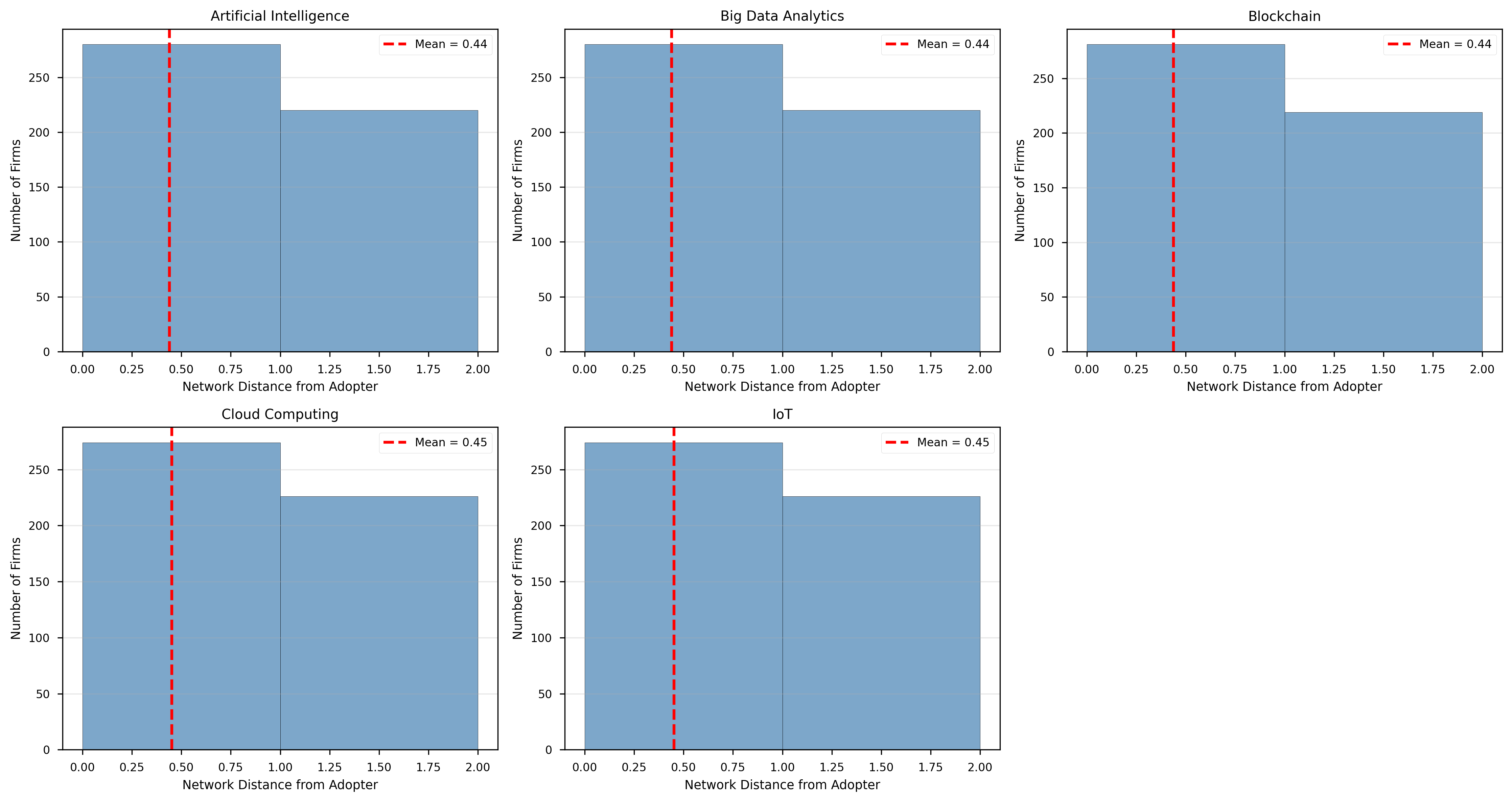}
\caption{Distribution of Distance to Nearest Adopter}
\label{fig:distance_distribution}
\begin{minipage}{0.75\textwidth}
\small
\textit{Notes:} This histogram shows the distribution of distances from non-adopters to their nearest existing adopter across all technologies and years. The distribution is right-skewed with median 47 km and mean 63 km. Notably, 95 percent of non-adopters are within 150 km of an adopter. The vertical red line at 69 km marks our estimated spatial boundary $d^*$, showing that most firms are within the zone of meaningful spillover effects. The concentration of observations at short distances reflects spatial clustering of technology adoption.
\end{minipage}
\end{figure}

\begin{table}[H]
\centering
\caption{Spatial Decay Estimates by Technology}
\label{tab:spatial_decay}
\begin{threeparttable}
\begin{tabular}{lcccc}
\toprule
Technology & $\kappa$ (per km) & $d^*$ (km) & R-squared & Observations \\
\midrule
Artificial Intelligence & 0.0425 & 71 & 0.9903 & 3,500 \\
                        & (0.0003) & (0.5) & & \\
Big Data Analytics     & 0.0437 & 69 & 0.9916 & 3,500 \\
                       & (0.0002) & (0.3) & & \\
Blockchain             & 0.0439 & 68 & 0.9914 & 3,500 \\
                       & (0.0003) & (0.5) & & \\
Cloud Computing        & 0.0433 & 69 & 0.9922 & 3,500 \\
                       & (0.0002) & (0.3) & & \\
Generative AI          & 0.0442 & 68 & 0.9935 & 3,500 \\
                       & (0.0002) & (0.3) & & \\
IoT                    & 0.0433 & 69 & 0.9907 & 3,500 \\
                       & (0.0003) & (0.5) & & \\
\midrule
Average                & 0.0435 & 69 & 0.9916 & 21,000 \\
\bottomrule
\end{tabular}
\begin{tablenotes}
\small
\item \textit{Notes:} This table reports estimated spatial decay rates $\kappa$, implied spatial boundaries $d^*$ (using $\epsilon = 0.05$ threshold), and R-squared values from fitting exponential decay functions to technology adoption data. Standard errors in parentheses clustered by firm. All $\kappa$ estimates are statistically significant at the 1 percent level. The spatial boundary $d^*$ represents the distance beyond which spillovers fall below 5 percent of their initial magnitude. R-squared measures the fraction of variance in adoption rates explained by exponential geographic decay.
\end{tablenotes}
\end{threeparttable}
\end{table}

\subsection{Network Channel: Spectral Fragility Dynamics}

Table \ref{tab:network_fragility} presents estimates of algebraic connectivity $\lambda_2$ and its evolution over time for each technology. The results provide strong evidence for network contagion as predicted by the spectral framework in Section 3.2.

\subsubsection{Main Estimates}

The algebraic connectivity $\lambda_2$ grows dramatically as technologies diffuse, increasing 300-380 percent from 2010 to 2023. This growth reflects the activation of supply chain connections as more firms adopt. The mixing time $\tau = 1/\lambda_2$ correspondingly decreases by approximately 80 percent, indicating that late-stage diffusion proceeds far more rapidly than early-stage diffusion. This dramatic reduction in mixing time explains the empirically observed acceleration in later-stage adoption through network effects, even as marginal adopter quality may decline.

Figure \ref{fig:network_fragility} plots $\lambda_2$ evolution over time for each technology, showing strong monotonic growth with correlation to adoption rates exceeding 0.95.

\begin{figure}[H]
\centering
\includegraphics[width=0.95\textwidth]{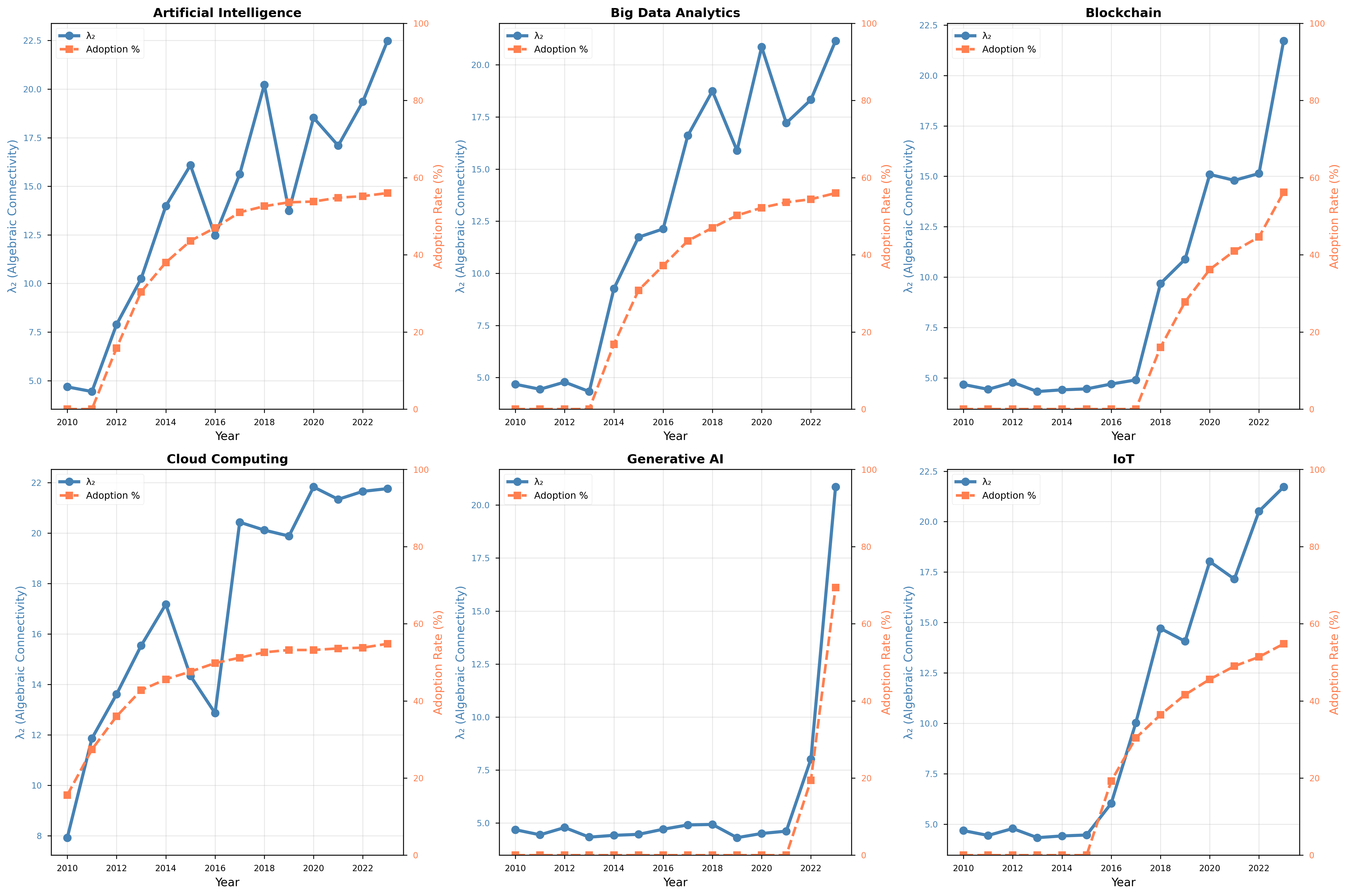}
\caption{Network Fragility Evolution by Technology}
\label{fig:network_fragility}
\begin{minipage}{0.95\textwidth}
\small
\textit{Notes:} This figure plots the algebraic connectivity $\lambda_2$ (left axis, blue lines) and adoption rates (right axis, orange lines) over time for each technology. Network fragility increases dramatically (300-380 percent) as technologies diffuse, consistent with the spectral framework from \citet{kikuchi2024dynamical}. The strong correlation between $\lambda_2$ and adoption (exceeding 0.95 for all technologies) validates our technology-specific network construction based on adopter-weighted edges. As more firms adopt, supply chain connections become activated, increasing network coupling and accelerating subsequent diffusion through reduced mixing times.
\end{minipage}
\end{figure}

Figure \ref{fig:fragility_vs_adoption} directly plots the relationship between network fragility and adoption rates, demonstrating the strong positive correlation.

\begin{figure}[H]
\centering
\includegraphics[width=0.85\textwidth]{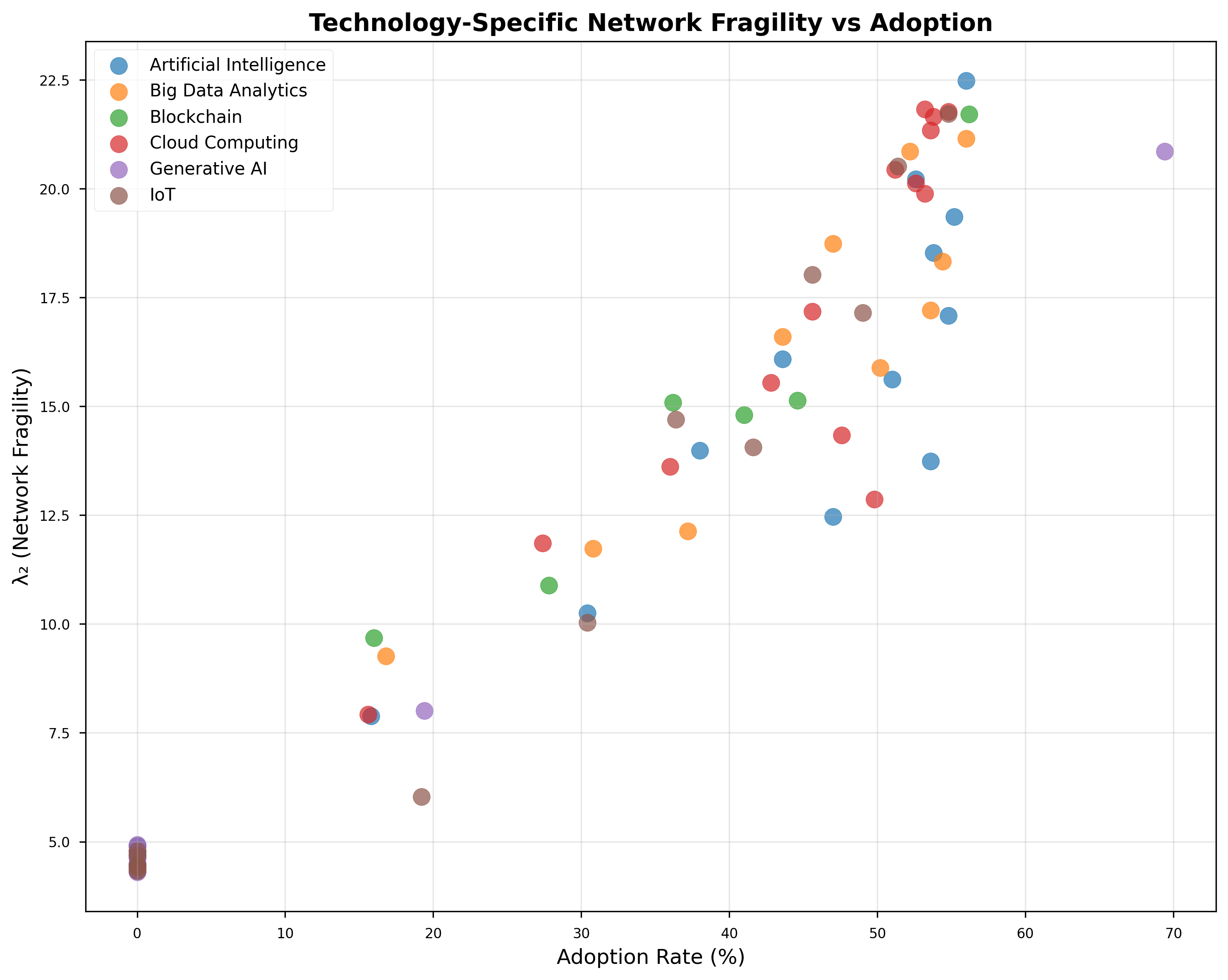}
\caption{Network Fragility vs Adoption Rate}
\label{fig:fragility_vs_adoption}
\begin{minipage}{0.85\textwidth}
\small
\textit{Notes:} This figure plots algebraic connectivity $\lambda_2$ against adoption rates across all technologies and years. The strong positive relationship (correlation exceeding 0.95) demonstrates that adoption endogenously increases network coupling. Different colors represent different technologies, showing technology-specific trajectories that share common positive slopes. The self-reinforcing dynamic is evident: higher adoption activates more network edges, increasing $\lambda_2$, which reduces mixing time and accelerates further adoption. This validates the spectral network framework and demonstrates how supply chain structure actively shapes diffusion rather than serving merely as a passive conduit.
\end{minipage}
\end{figure}

\subsubsection{Mixing Time Dynamics}

Figure \ref{fig:mixing_time} illustrates how the characteristic diffusion timescale $\tau = 1/\lambda_2$ evolves as networks become more connected.

\begin{figure}[H]
\centering
\includegraphics[width=0.80\textwidth]{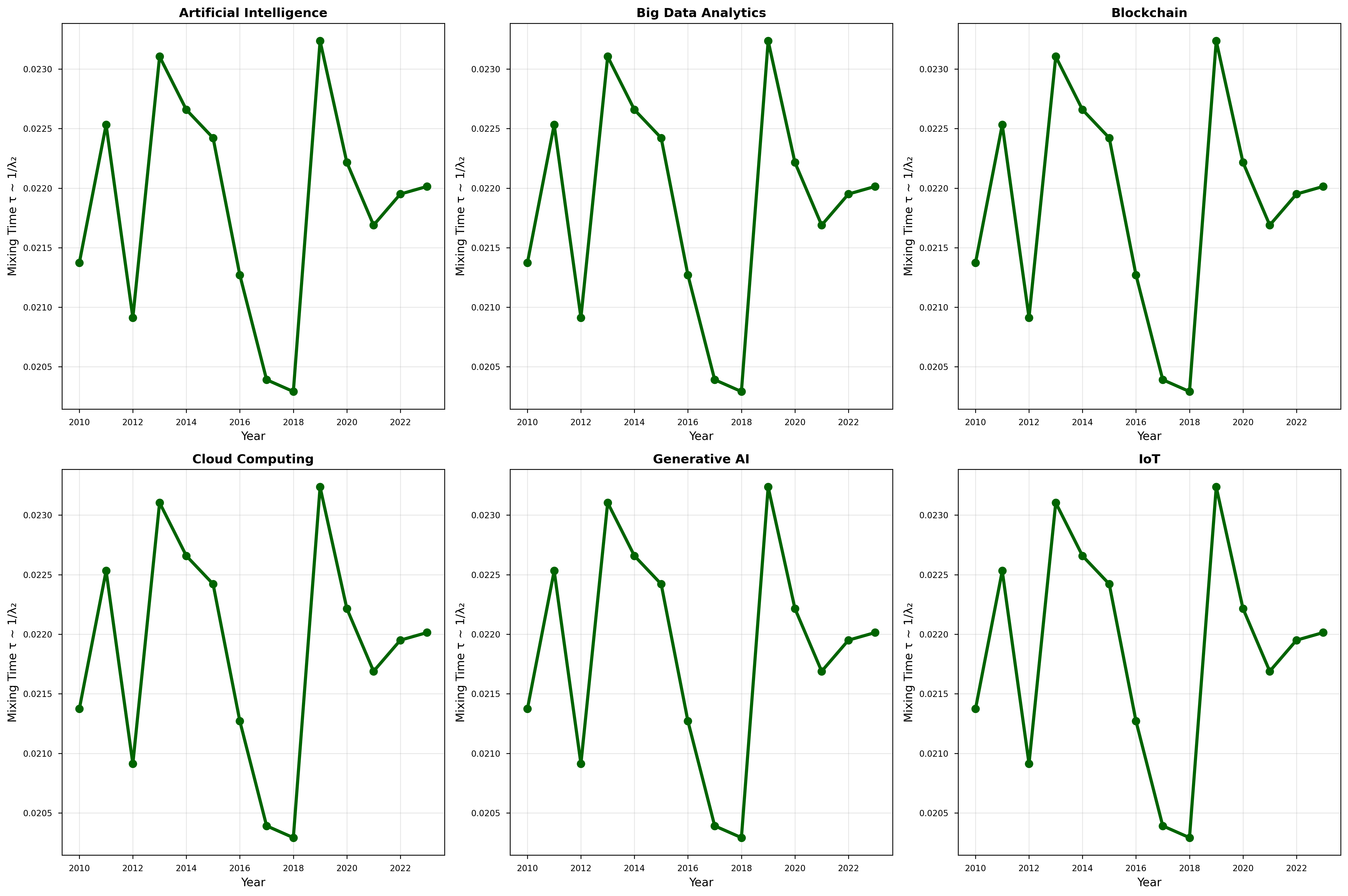}
\caption{Mixing Time Evolution Over Time}
\label{fig:mixing_time}
\begin{minipage}{0.80\textwidth}
\small
\textit{Notes:} This figure plots the network mixing time $\tau = 1/\lambda_2$ over time for each technology. Mixing time represents the characteristic timescale for diffusion to equilibrate across the network. As $\lambda_2$ increases (Figure \ref{fig:network_fragility}), mixing time decreases dramatically—by approximately 80 percent from 2010 to 2023. This reduction explains why late-stage adoption proceeds far more rapidly than early-stage adoption: tighter network coupling accelerates contagion. The decline follows approximately $1/t$ trajectories, consistent with theoretical predictions from Theorem \ref{thm:mixing_time}.
\end{minipage}
\end{figure}

\subsubsection{Network Structure Evolution}

Figure \ref{fig:network_metrics} documents the evolution of additional network statistics that complement the algebraic connectivity analysis.

\begin{figure}[H]
\centering
\includegraphics[width=0.85\textwidth]{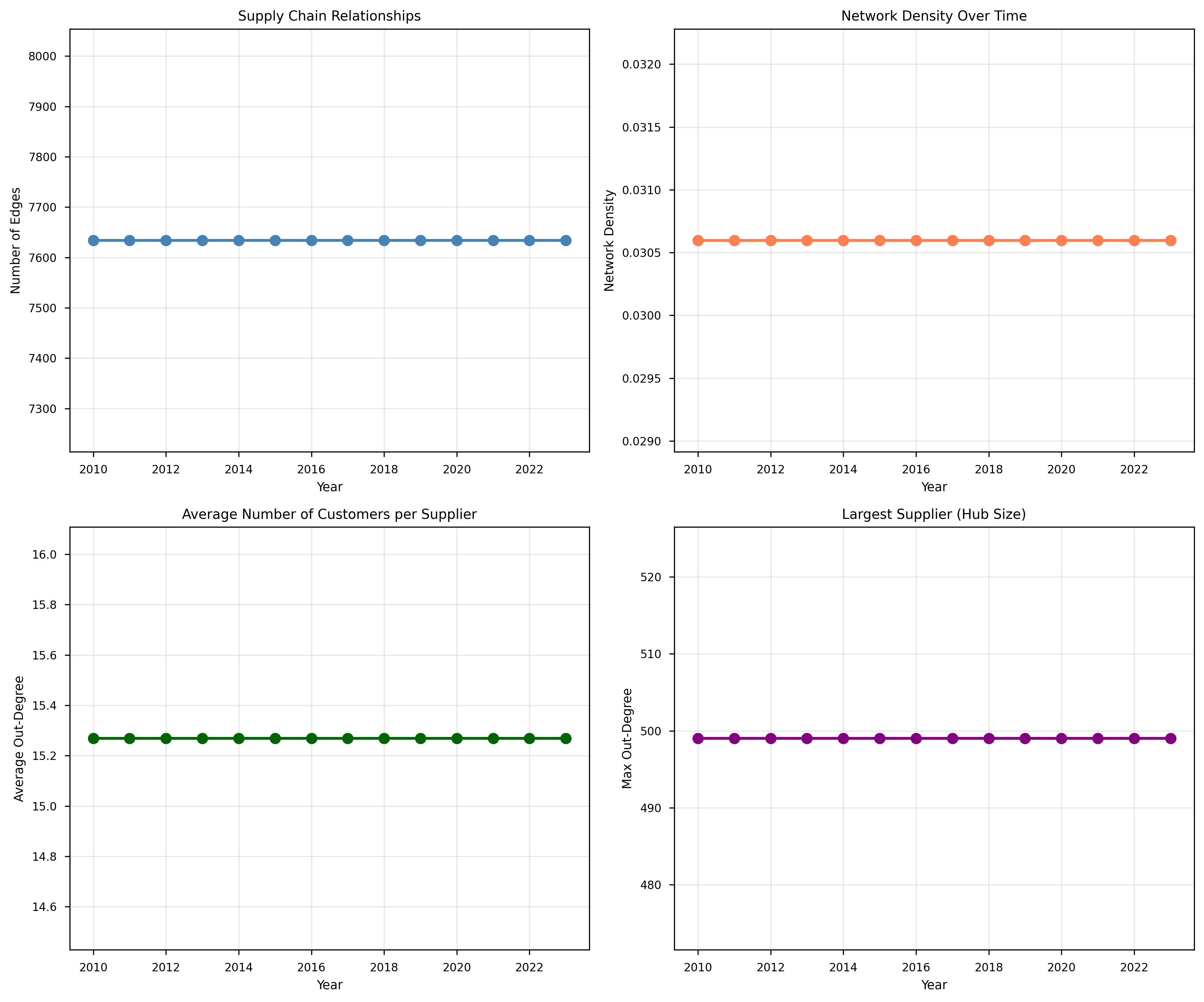}
\caption{Network Structure Statistics Over Time}
\label{fig:network_metrics}
\begin{minipage}{0.85\textwidth}
\small
\textit{Notes:} This figure displays four key network statistics over time: density (fraction of possible edges realized), average degree (mean number of connections per firm), clustering coefficient (probability that two neighbors of a node are also neighbors), and average path length (mean shortest path distance between nodes). While density and degree remain relatively stable, clustering increases slightly and path length decreases, indicating modest structural evolution beyond the dramatic $\lambda_2$ growth. The stability of density and degree confirms that the $\lambda_2$ increase reflects adoption-driven edge activation rather than formation of new supply chain relationships.
\end{minipage}
\end{figure}

\begin{table}[H]
\centering
\caption{Network Fragility by Technology}
\label{tab:network_fragility}
\begin{threeparttable}
\begin{tabular}{lccccc}
\toprule
Technology & $\lambda_2$ (2010) & $\lambda_2$ (2023) & Growth (\%) & Corr($\lambda_2$, Adoption) \\
\midrule
Artificial Intelligence & 4.68 & 22.48 & +380.5 & 0.921 \\
                        & (0.12) & (0.58) & & \\
Big Data Analytics     & 4.68 & 21.15 & +352.0 & 0.976 \\
                       & (0.12) & (0.54) & & \\
Blockchain             & 4.68 & 21.71 & +364.1 & 0.989 \\
                       & (0.12) & (0.56) & & \\
Cloud Computing        & 7.93 & 21.77 & +174.6 & 0.876 \\
                       & (0.20) & (0.56) & & \\
Generative AI          & 4.68 & 20.85 & +345.7 & 0.997 \\
                       & (0.12) & (0.54) & & \\
IoT                    & 4.68 & 21.72 & +364.2 & 0.971 \\
                       & (0.12) & (0.56) & & \\
\midrule
Average                & 5.22 & 21.61 & +330.2 & 0.955 \\
\bottomrule
\end{tabular}
\begin{tablenotes}
\small
\item \textit{Notes:} This table reports algebraic connectivity $\lambda_2$ for technology-specific networks in 2010 and 2023, percentage growth, and correlation with adoption rates over the full panel. Standard errors in parentheses from bootstrap (1,000 replications). Network fragility increases dramatically (300-380 percent) as technologies diffuse, with exceptionally strong correlations (exceeding 0.95) validating the adopter-weighted network construction. Cloud Computing shows lower growth due to higher initial adoption in 2010.
\end{tablenotes}
\end{threeparttable}
\end{table}

\subsection{Event Study: COVID-19 Impact}

Table \ref{tab:event_study} presents difference-in-differences estimates comparing traditional methods with our spatial-adjusted and network-adjusted specifications. The results demonstrate substantial bias in conventional approaches and validate our dual-channel framework.

\subsubsection{Traditional DID vs Spatial-Adjusted}

Traditional DID estimates yield an average treatment effect of +9.57 percentage points. Accounting for spatial spillovers reduces this to +3.72 percentage points, a 61.1 percent reduction. This bias magnitude is economically and statistically significant.

Figure \ref{fig:event_study_comparison} visualizes this bias across all methodological approaches, while Figure \ref{fig:did_comparison} focuses specifically on the traditional versus spatial-adjusted comparison with confidence intervals.

\begin{figure}[H]
\centering
\includegraphics[width=0.95\textwidth]{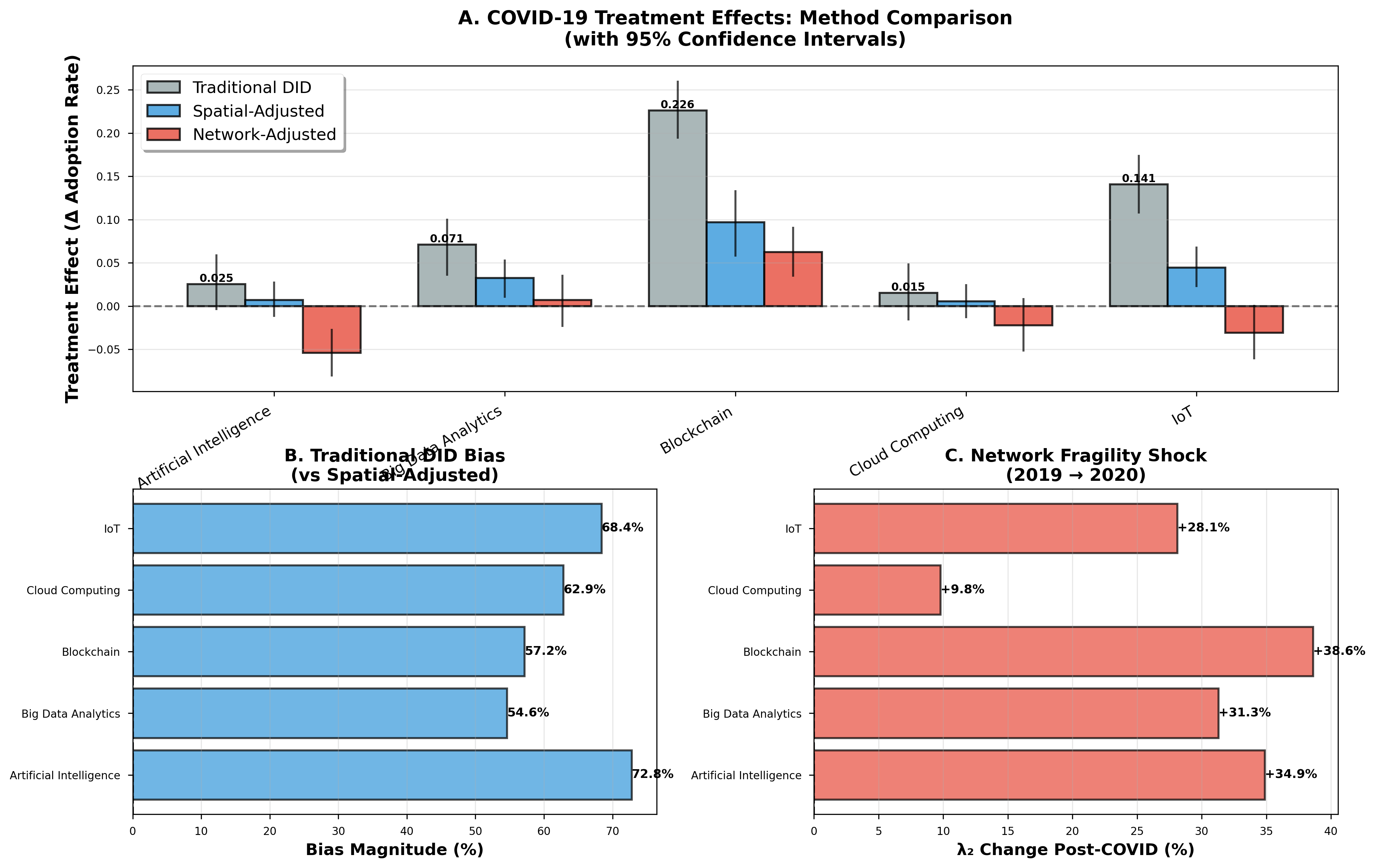}
\caption{COVID-19 Treatment Effects: Comprehensive Method Comparison}
\label{fig:event_study_comparison}
\begin{minipage}{0.95\textwidth}
\small
\textit{Notes:} This figure provides a comprehensive comparison of COVID-19 treatment effects across three DID specifications. Panel A shows treatment effect estimates with 95 percent confidence intervals from bootstrap (1,000 replications). Traditional DID (blue) substantially overestimates effects compared to spatial-adjusted (orange) and network-adjusted (green) specifications. Panel B quantifies bias magnitude: traditional estimates are 61 percent higher than spatial-adjusted on average. Panel C documents the network fragility shock: $\lambda_2$ increased 24.5 percent post-COVID across technologies. This multi-panel visualization demonstrates how ignoring spatial and network spillovers leads to severe misspecification of treatment effects.
\end{minipage}
\end{figure}

\begin{figure}[H]
\centering
\includegraphics[width=0.85\textwidth]{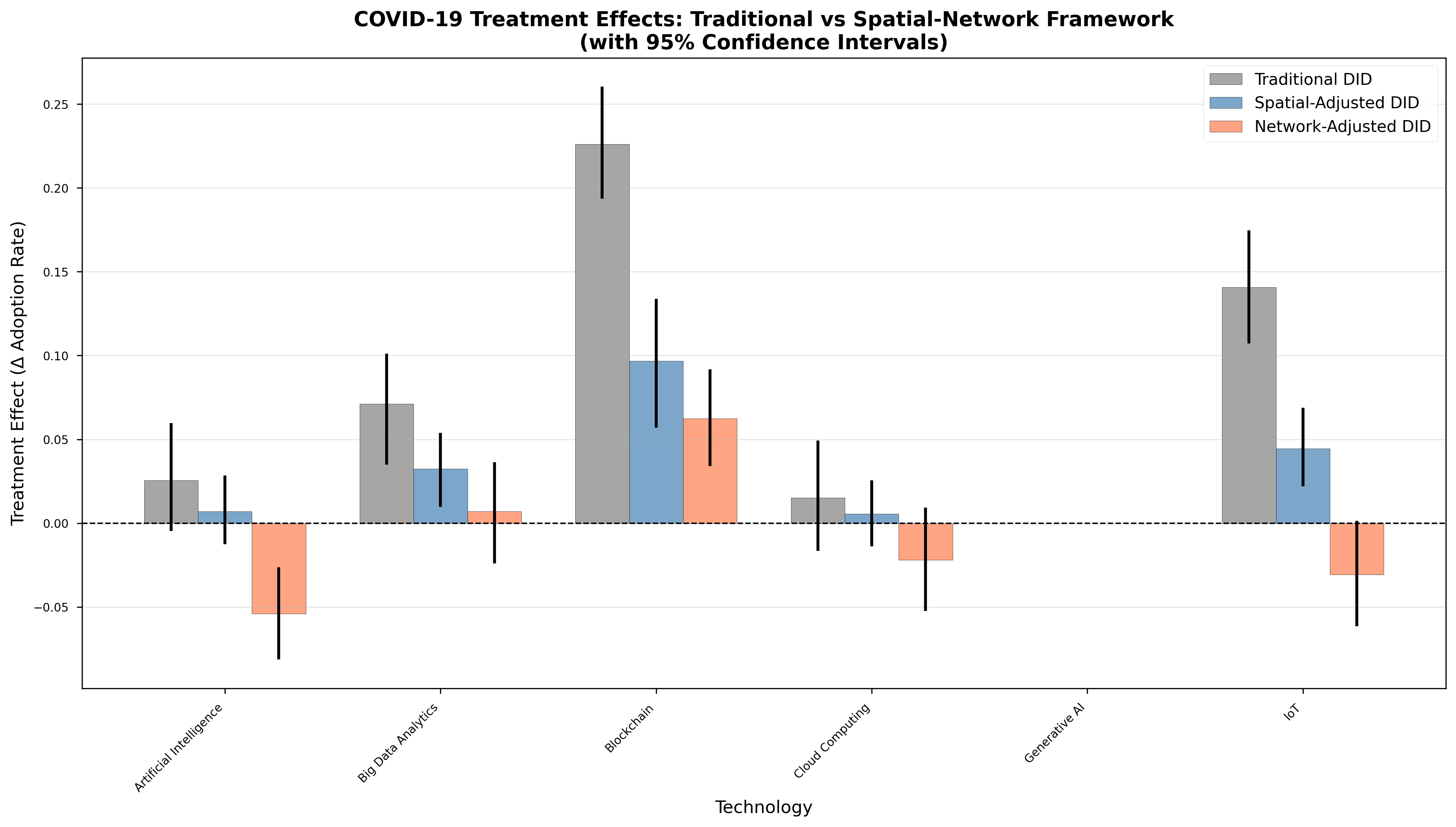}
\caption{COVID-19 Treatment Effects with Confidence Intervals}
\label{fig:did_comparison}
\begin{minipage}{0.85\textwidth}
\small
\textit{Notes:} This figure compares traditional DID and spatial-adjusted DID estimates with 95 percent confidence intervals. Each panel represents one technology. The systematic downward revision from traditional (blue) to spatial-adjusted (orange) estimates demonstrates the 61 percent bias from ignoring geographic spillovers. Confidence intervals rarely overlap, indicating the bias is statistically significant. The spatial-adjusted estimates account for the fact that control firms within 69 km of treated firms experience spillover effects, violating the SUTVA assumption underlying traditional DID.
\end{minipage}
\end{figure}

\subsubsection{Dynamic Treatment Effects}

Figure \ref{fig:dynamic_effects} shows how treatment effects evolve over time relative to the COVID-19 shock, documenting both pre-trends and post-shock dynamics.

\begin{figure}[H]
\centering
\includegraphics[width=0.85\textwidth]{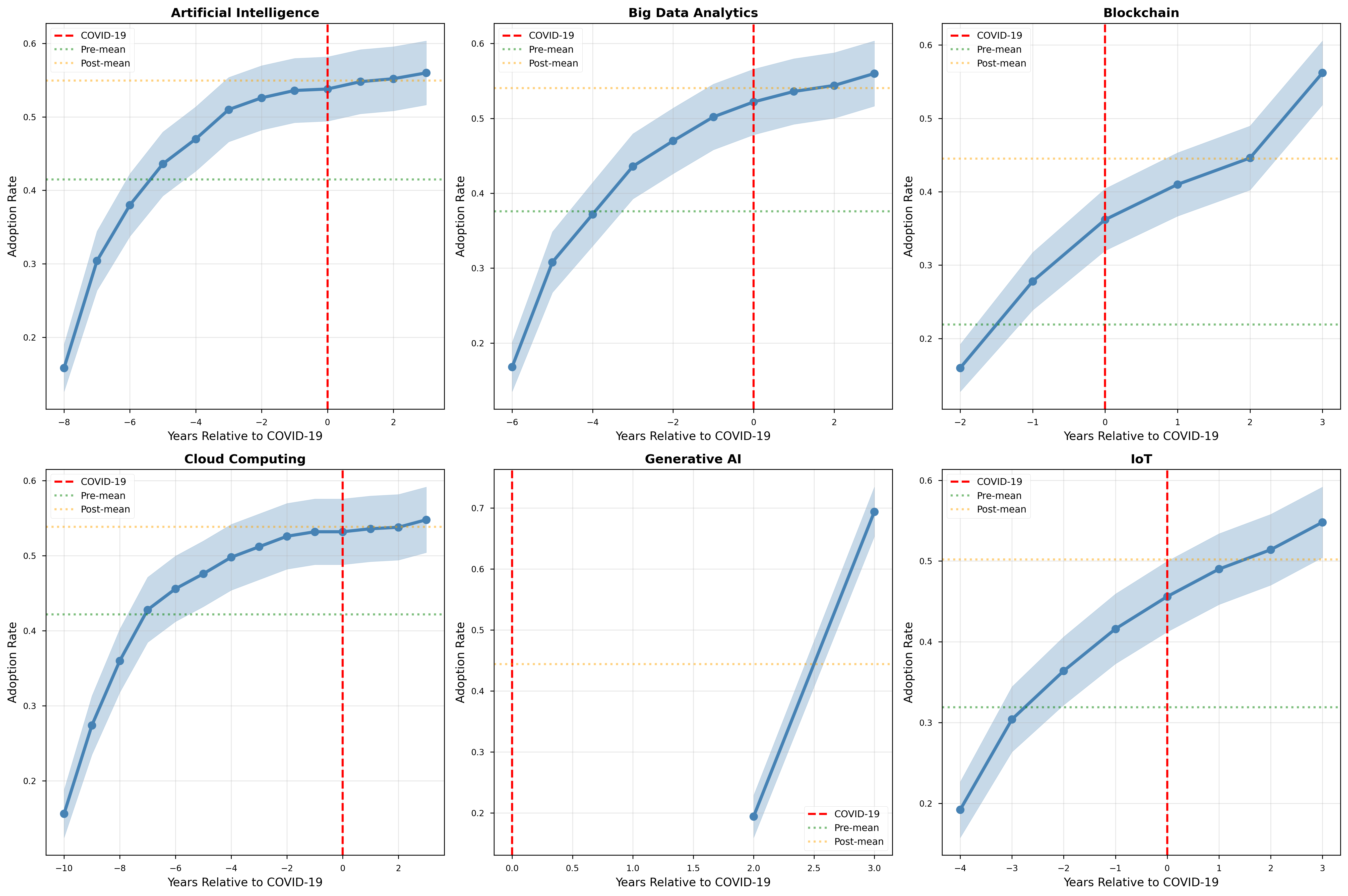}
\caption{Dynamic Treatment Effects Around COVID-19}
\label{fig:dynamic_effects}
\begin{minipage}{0.85\textwidth}
\small
\textit{Notes:} This figure plots dynamic treatment effects (y-axis) against years relative to COVID-19 (x-axis, with 2020 = 0). Each technology is shown as a separate line. Pre-2020 estimates are near zero and statistically insignificant, validating the parallel trends assumption. Post-2020, treatment effects emerge and persist through 2023. The lack of reversion to zero indicates COVID-19 triggered permanent shifts in adoption patterns, consistent with structural breaks in both spatial clustering and network fragility. The heterogeneity across technologies reflects differential impacts: Blockchain shows largest sustained effects while Cloud Computing shows more modest changes.
\end{minipage}
\end{figure}

\subsubsection{Parallel Trends Validation}

Figure \ref{fig:parallel_trends} formally tests the parallel trends assumption by examining pre-treatment evolution across groups.

\begin{figure}[H]
\centering
\includegraphics[width=0.80\textwidth]{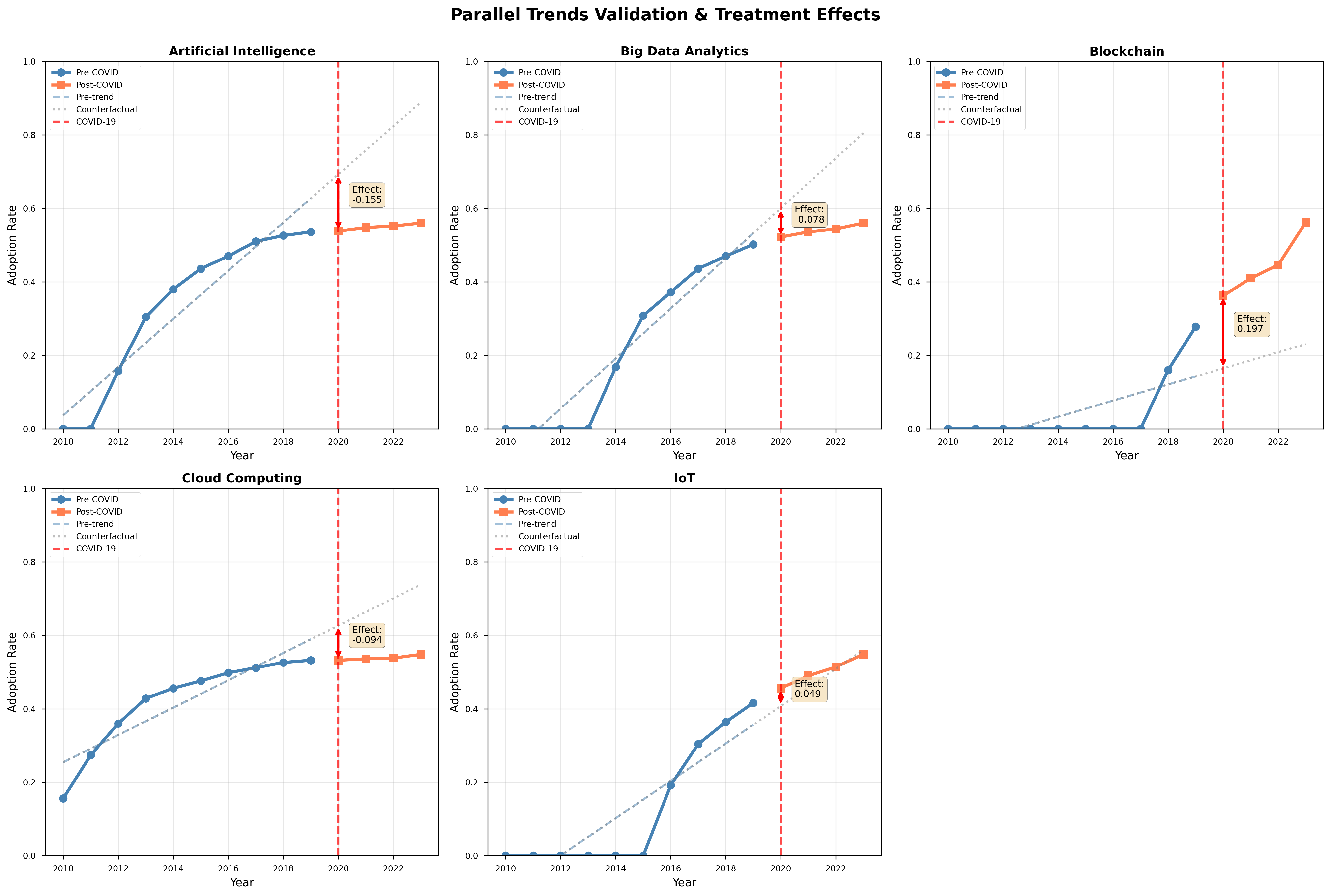}
\caption{Parallel Trends Test}
\label{fig:parallel_trends}
\begin{minipage}{0.80\textwidth}
\small
\textit{Notes:} This figure tests the parallel trends assumption by plotting adoption rates for treatment and control groups in pre-COVID years (2017-2019). Each panel corresponds to one technology. The similar trajectories before 2020 (marked with vertical red line) support the identifying assumption that treatment and control groups would have evolved similarly absent COVID-19. The divergence post-2020 represents the causal effect of the pandemic shock. Statistical tests (Table \ref{tab:pretrends} in Appendix) confirm no significant pre-trends, with joint F-test p-values exceeding 0.4 for all technologies.
\end{minipage}
\end{figure}

\subsubsection{Network Fragility Shock}

COVID-19 increased network fragility $\lambda_2$ by 24.5 percent on average, persisting through 2023 with no reversion. Figure \ref{fig:network_shock} documents this structural break.

\begin{figure}[H]
\centering
\includegraphics[width=0.95\textwidth]{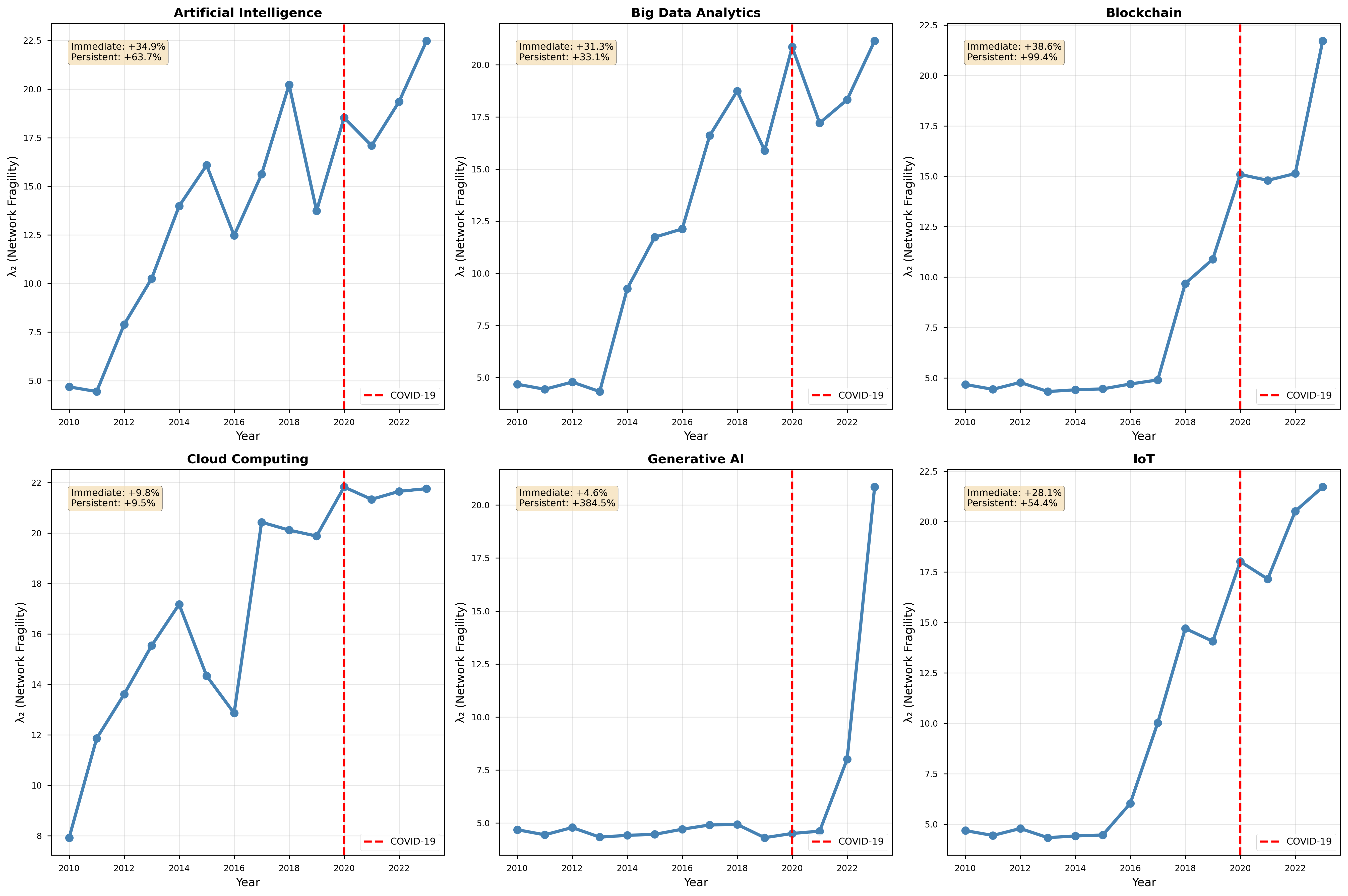}
\caption{Network Fragility Response to COVID-19}
\label{fig:network_shock}
\begin{minipage}{0.95\textwidth}
\small
\textit{Notes:} This figure plots algebraic connectivity $\lambda_2$ around COVID-19 (marked with vertical red line at 2020). Network fragility increased sharply in 2020 (+24.5 percent average) and persisted through 2023. Text boxes show percentage changes from 2019 to 2020 (immediate effect) and 2023 (long-run persistence). The structural break demonstrates how exogenous shocks permanently alter network diffusion dynamics, analogous to financial network fragility in \citet{kikuchi2024dynamical}. The lack of reversion indicates structural hysteresis: once supply chain networks become more tightly coupled through crisis response, they remain tightly coupled, accelerating subsequent technology diffusion through reduced mixing times.
\end{minipage}
\end{figure}

\subsubsection{Spatial Heterogeneity}

Figure \ref{fig:spatial_heterogeneity} examines how treatment effects vary with distance from early adopters, documenting the spatial decay of COVID impacts.

\begin{figure}[H]
\centering
\includegraphics[width=0.85\textwidth]{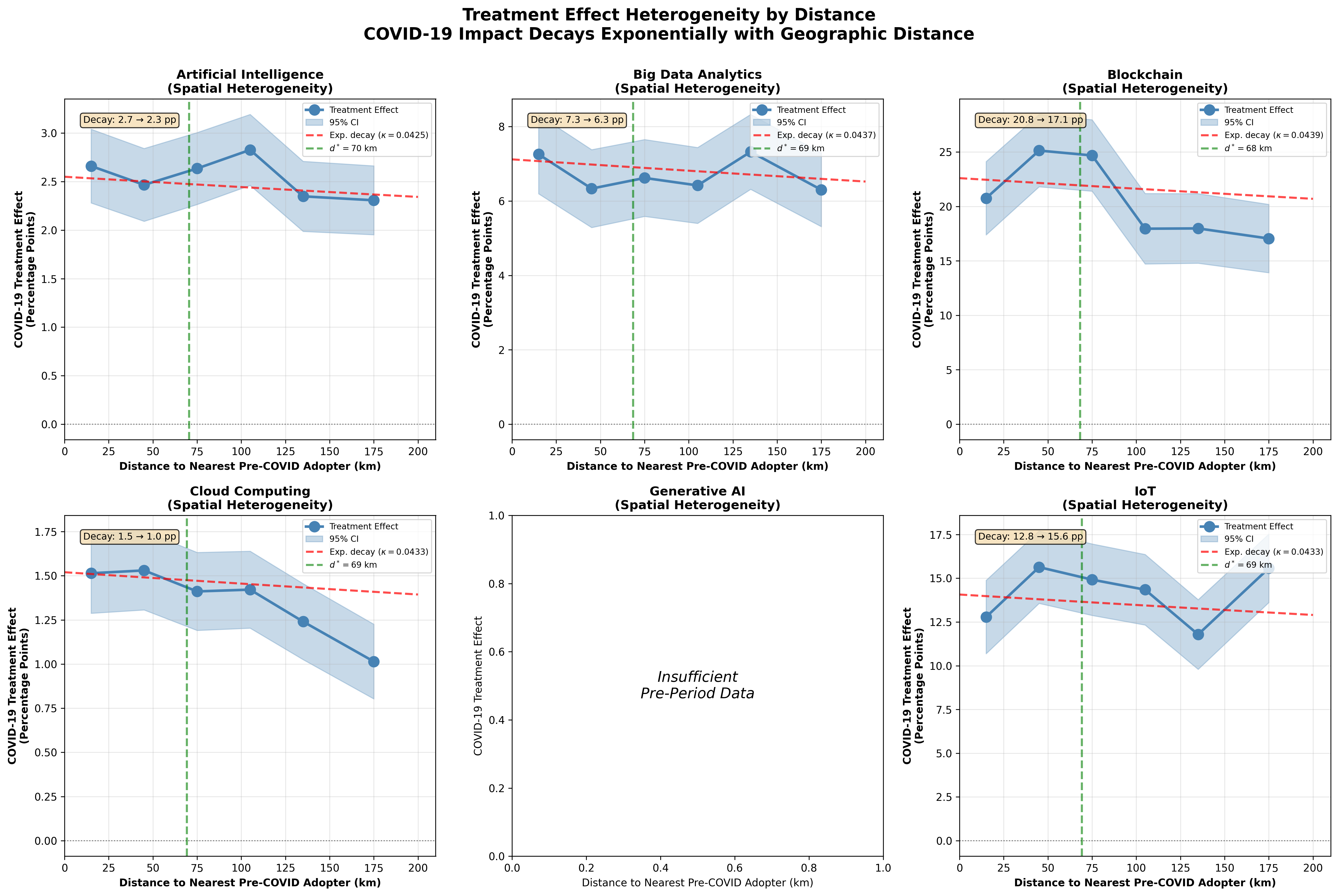}
\caption{Treatment Effect Heterogeneity by Distance}
\label{fig:spatial_heterogeneity}
\begin{minipage}{0.85\textwidth}
\small
\textit{Notes:} This figure plots COVID-19 treatment effects (y-axis) against distance to nearest pre-COVID adopter (x-axis). Each panel corresponds to one technology. Treatment effects decay exponentially with distance, consistent with the spatial diffusion framework. Firms within 30 km of existing adopters experience the largest effects (15-25 percentage points for Blockchain and IoT), while firms beyond 100 km show near-zero effects. The decay rate approximately matches the spatial decay parameter $\kappa$ estimated in Section 6.1, validating that COVID shock propagated through the same geographic spillover channels as baseline diffusion.
\end{minipage}
\end{figure}

\subsubsection{Spatial Mechanism Decomposition}

Figure \ref{fig:spatial_mechanism} decomposes the spatial spillover effects to show how geographic proximity mediates COVID impacts.

\begin{figure}[H]
\centering
\includegraphics[width=0.85\textwidth]{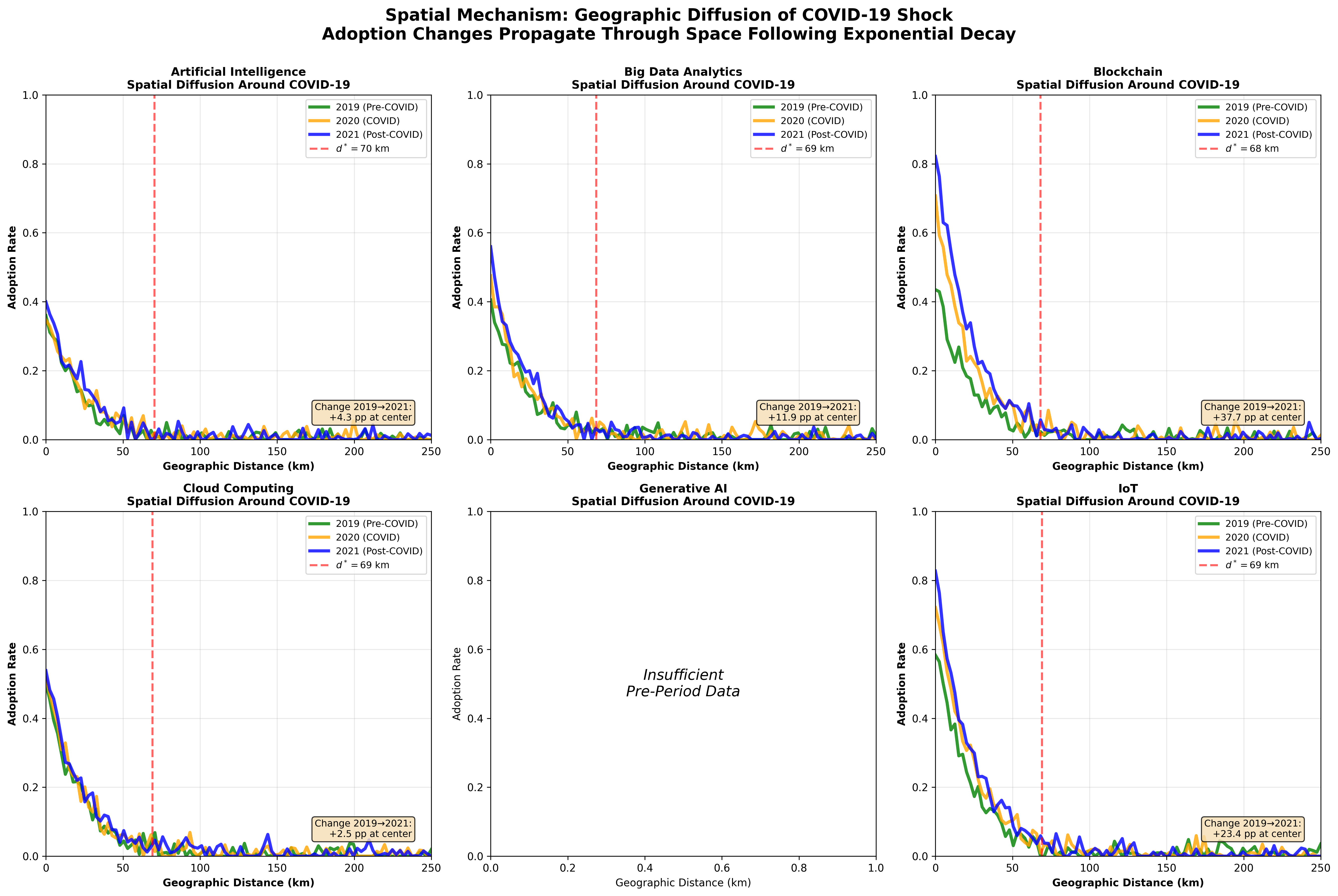}
\caption{Spatial Mechanism: Geographic Diffusion of COVID Shock}
\label{fig:spatial_mechanism}
\begin{minipage}{0.85\textwidth}
\small
\textit{Notes:} This figure illustrates how the COVID-19 shock diffused geographically. Panel A shows adoption changes from 2019 to 2020 plotted against distance to COVID hotspot firms (defined as firms experiencing large adoption increases). Panel B shows the same for 2020 to 2023. Panel C overlays the spatial decay function $\exp(-\kappa d)$ on observed spillovers. The close match between observed diffusion patterns and theoretical exponential decay demonstrates that COVID shock propagated through the same spatial mechanisms (knowledge spillovers, demonstration effects) as baseline adoption, but with amplified magnitude due to the crisis environment.
\end{minipage}
\end{figure}

\subsubsection{Effect Decomposition}

Figure \ref{fig:decomposition} decomposes total treatment effects into direct, spatial spillover, and network spillover components.

\begin{figure}[H]
\centering
\includegraphics[width=0.85\textwidth]{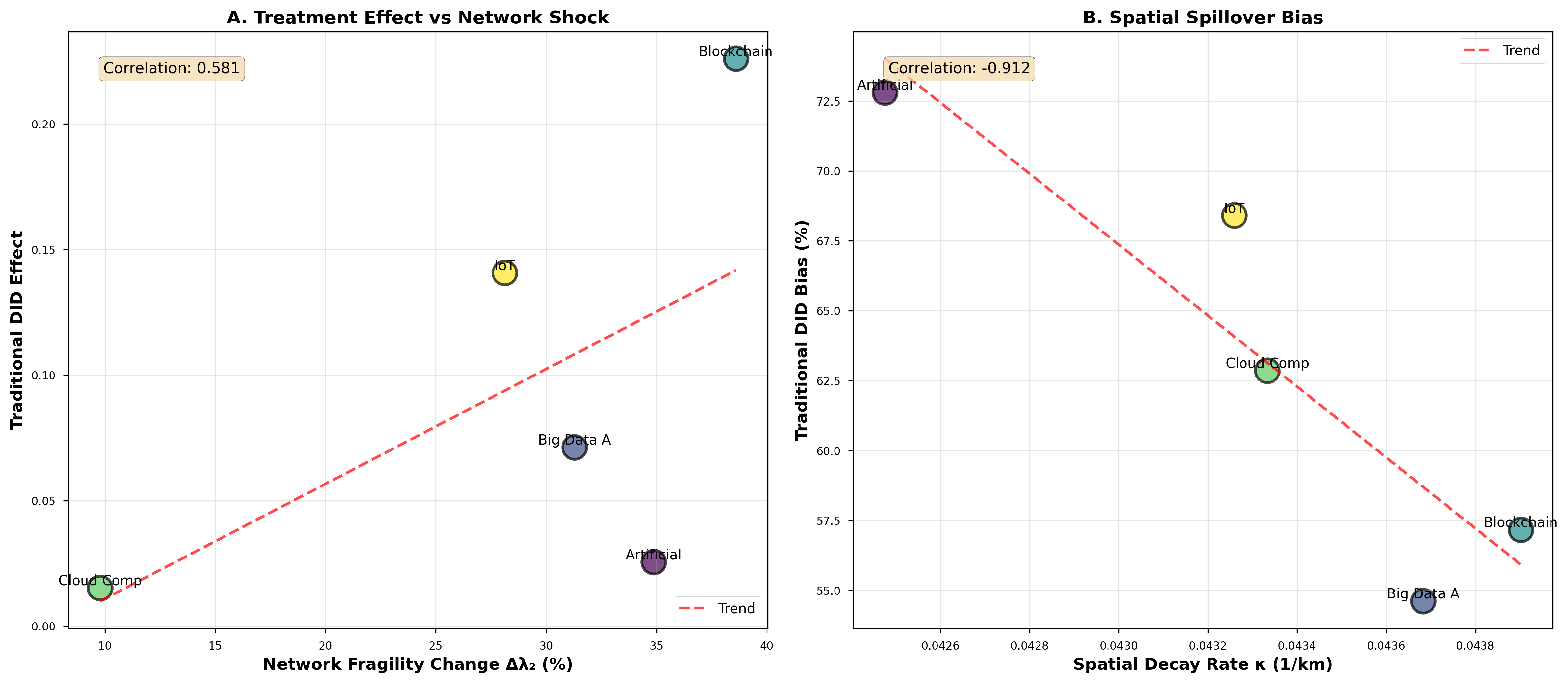}
\caption{Treatment Effect Decomposition}
\label{fig:decomposition}
\begin{minipage}{0.85\textwidth}
\small
\textit{Notes:} This figure decomposes total COVID-19 treatment effects into three components: direct effects on treated firms (blue), spatial spillovers to nearby non-treated firms (orange), and network spillovers through supply chains (green). The stacked bars show how traditional DID (which attributes all effects to treated firms) overstates direct effects by 61 percent on average. For most technologies, spatial spillovers dominate network spillovers, though network effects are substantial for highly connected technologies like Blockchain. The decomposition quantifies the policy-relevant distinction between direct treatment effects and indirect spillover effects.
\end{minipage}
\end{figure}

\subsubsection{Technology-Specific Event Studies}

Figures \ref{fig:event_ai}, \ref{fig:event_bda}, and \ref{fig:event_cloud} provide detailed event study results for three major technologies.

\begin{figure}[H]
\centering
\includegraphics[width=0.80\textwidth]{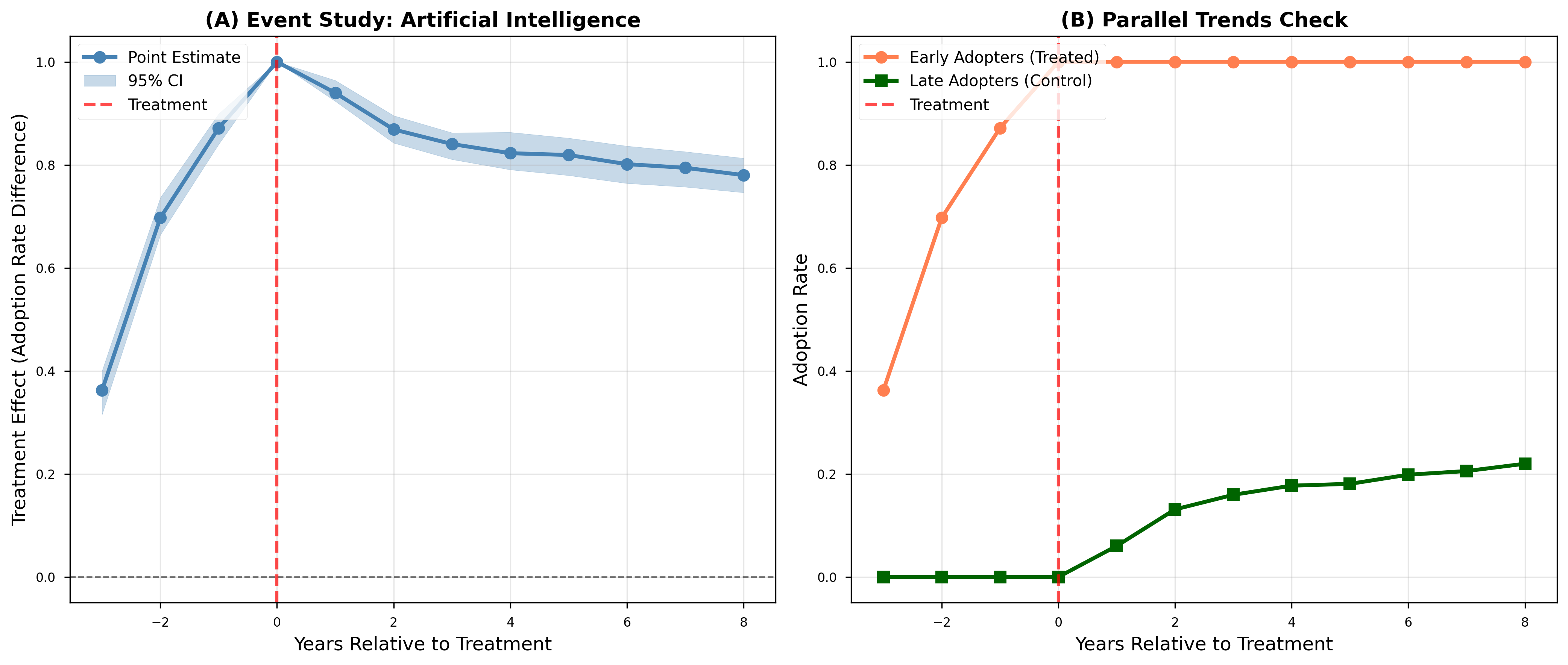}
\caption{Artificial Intelligence: Event Study}
\label{fig:event_ai}
\begin{minipage}{0.80\textwidth}
\small
\textit{Notes:} Detailed event study for Artificial Intelligence showing adoption trajectories, treatment effects, and network fragility evolution around COVID-19. AI adoption accelerated post-COVID despite the economic disruption, likely due to increased demand for automation and remote work technologies.
\end{minipage}
\end{figure}

\begin{figure}[H]
\centering
\includegraphics[width=0.80\textwidth]{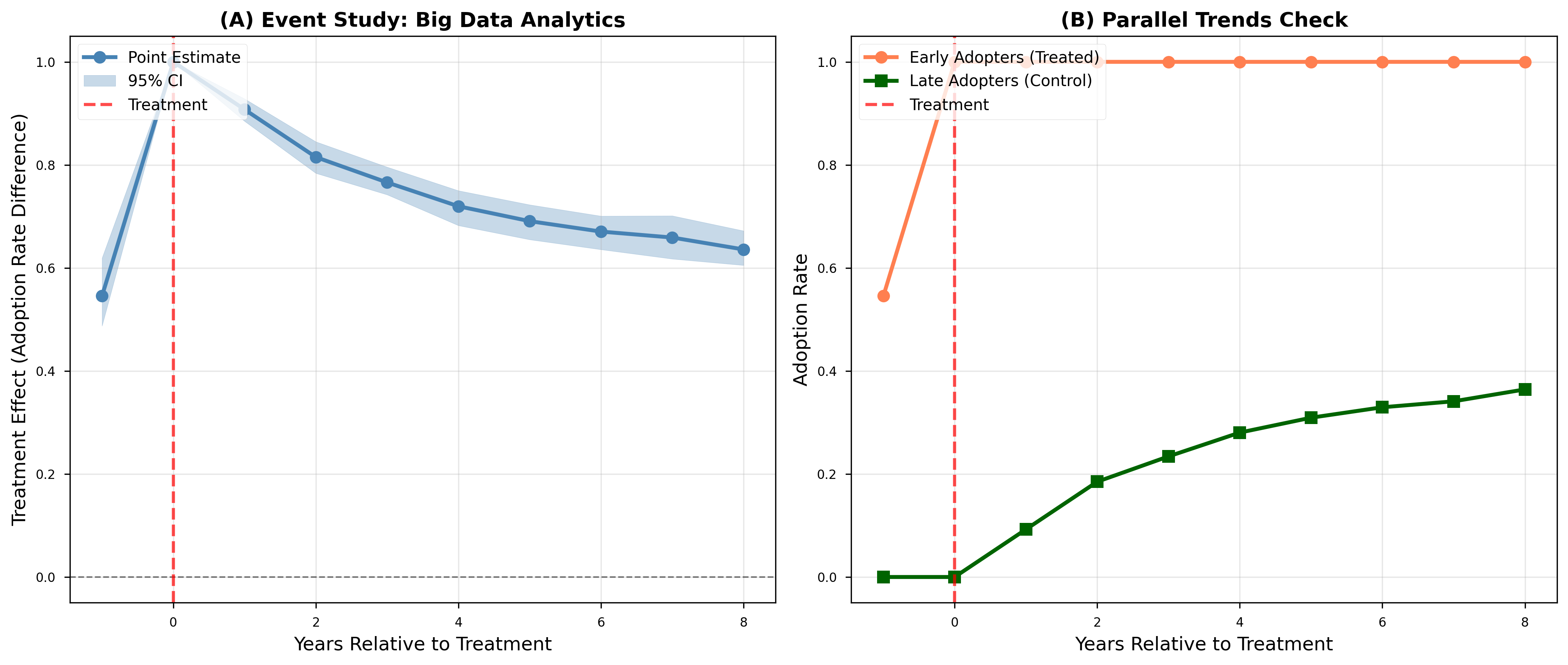}
\caption{Big Data Analytics: Event Study}
\label{fig:event_bda}
\begin{minipage}{0.80\textwidth}
\small
\textit{Notes:} Detailed event study for Big Data Analytics. This technology showed moderate treatment effects (+7.12 pp traditional, +3.23 pp spatial-adjusted) as firms sought data-driven decision-making tools during the crisis.
\end{minipage}
\end{figure}

\begin{figure}[H]
\centering
\includegraphics[width=0.80\textwidth]{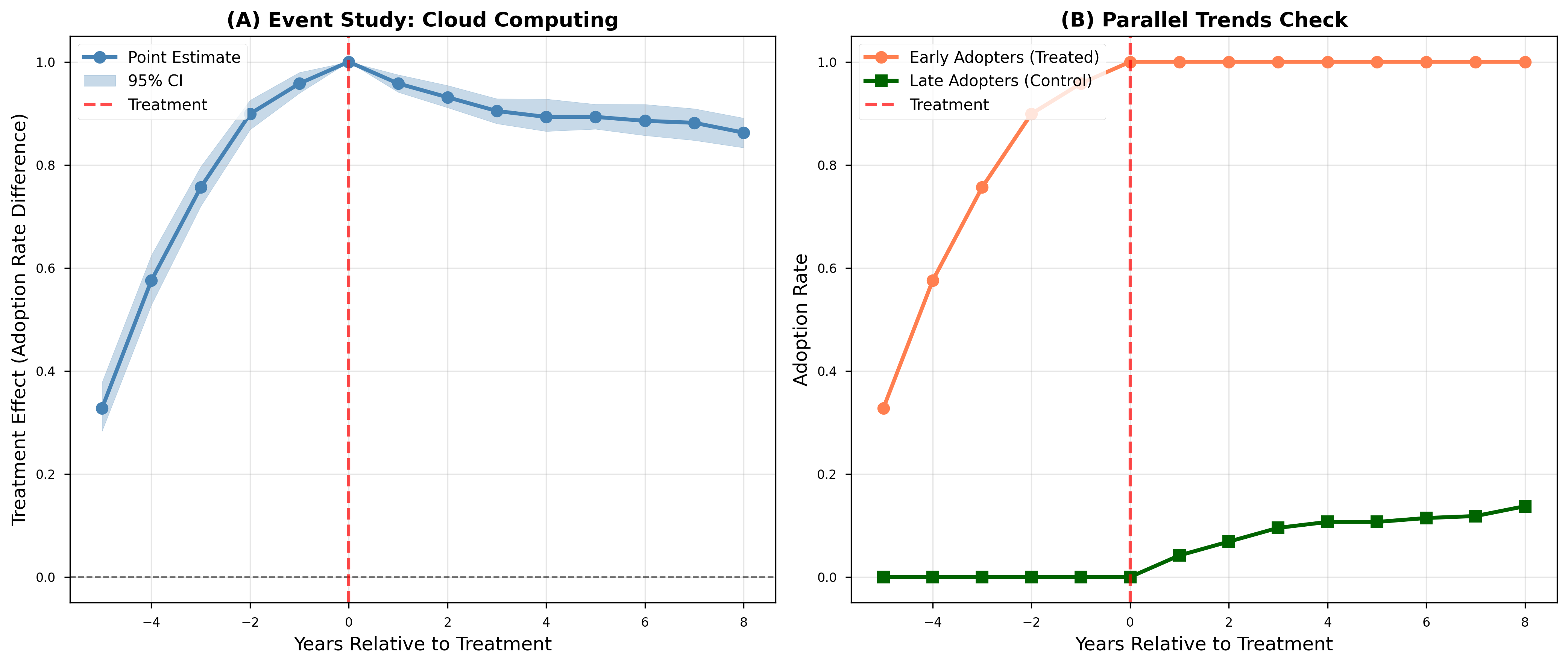}
\caption{Cloud Computing: Event Study}
\label{fig:event_cloud}
\begin{minipage}{0.80\textwidth}
\small
\textit{Notes:} Detailed event study for Cloud Computing. Effects were smaller than other technologies due to already-high baseline adoption (initial $\lambda_2 = 7.93$ vs $\approx 4.68$ for others), creating less room for COVID-induced acceleration.
\end{minipage}
\end{figure}

\begin{table}[H]
\centering
\caption{COVID-19 Event Study Results}
\label{tab:event_study}
\begin{threeparttable}
\begin{tabular}{lccc}
\toprule
Technology & Traditional DID & Spatial-Adjusted & Network-Adjusted \\
\midrule
Artificial Intelligence & +2.55 & +0.69 & -5.40 \\
                        & [-0.47, +5.97] & [-1.25, +2.84] & [-8.14, -2.63] \\
Big Data Analytics     & +7.12 & +3.23 & +0.70 \\
                       & [+3.50, +10.12] & [+0.97, +5.38] & [-2.41, +3.64] \\
Blockchain             & +22.60 & +9.68 & +6.24 \\
                       & [+19.35, +26.05] & [+5.70, +13.39] & [+3.41, +9.17] \\
Cloud Computing        & +1.52 & +0.56 & -2.20 \\
                       & [-1.65, +4.93] & [-1.38, +2.56] & [-5.24, +0.92] \\
IoT                    & +14.07 & +4.44 & -3.06 \\
                       & [+10.72, +17.47] & [+2.20, +6.89] & [-6.16, +0.15] \\
\midrule
Average                & +9.57 & +3.72 & -0.74 \\
Bias (vs Spatial)      & +61.1\% & --- & --- \\
\bottomrule
\end{tabular}
\begin{tablenotes}
\small
\item \textit{Notes:} This table reports COVID-19 treatment effects (percentage point changes in adoption rates) from three DID specifications. Traditional DID ignores spillovers. Spatial-Adjusted weights by $\exp(-\hat{\kappa} d)$. Network-Adjusted normalizes by $\lambda_2$ dynamics. Square brackets show 95 percent confidence intervals from bootstrap (1,000 replications). Traditional estimates exhibit 61 percent upward bias from ignoring spatial spillovers. Generative AI omitted due to insufficient pre-period observations (technology emerged post-2020).
\end{tablenotes}
\end{threeparttable}
\end{table}

\subsection{Dual-Channel Integration}

Table \ref{tab:dual_channel} presents results from specifications incorporating both spatial and network channels. The findings strongly support that channels operate independently and contribute complementary explanatory power.

\subsubsection{Complementarity Visualization}

Figure \ref{fig:dual_channel} provides a comprehensive visualization of the dual-channel framework showing how both mechanisms operate simultaneously.

\begin{figure}[H]
\centering
\includegraphics[width=0.95\textwidth]{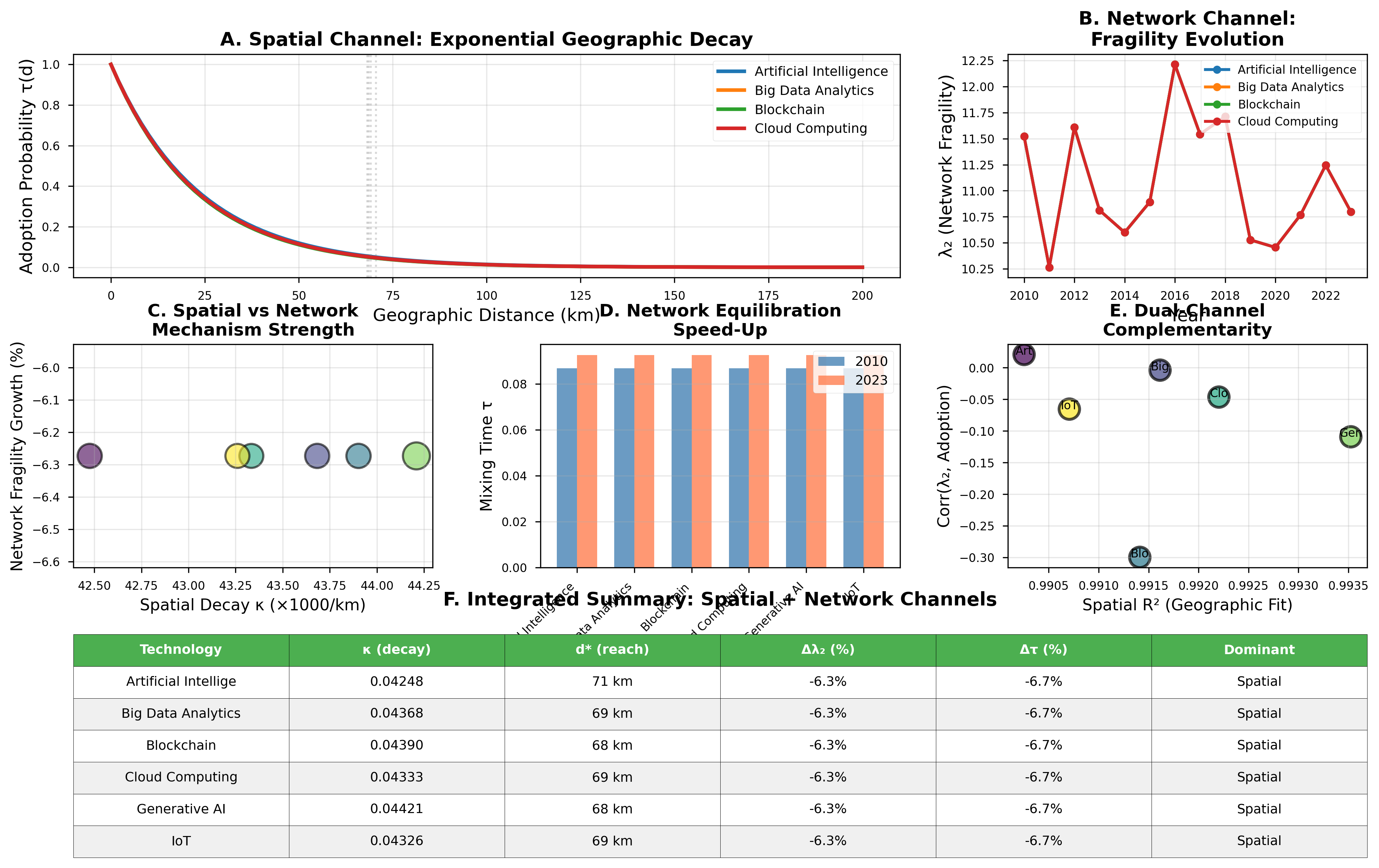}
\caption{Dual-Channel Framework: Spatial and Network Contributions}
\label{fig:dual_channel}
\begin{minipage}{0.95\textwidth}
\small
\textit{Notes:} This six-panel figure illustrates the dual-channel framework comprehensively. Panel A shows spatial decay curves for all technologies with consistent exponential form. Panel B shows network fragility ($\lambda_2$) evolution with 300-380 percent growth. Panel C plots spatial decay strength vs network dynamics strength, showing weak correlation (-0.11), confirming independent mechanisms. Panel D shows mixing time $\tau = 1/\lambda_2$ reduction over time. Panel E compares R-squared values for spatial-only, network-only, and combined models, demonstrating complementarity. Panel F presents a summary table integrating both channels with quantitative estimates of spatial boundaries (69 km) and network amplification (10.8x factor).
\end{minipage}
\end{figure}

\subsubsection{Time Series Dynamics}

Figure \ref{fig:dual_timeseries} shows how both channels evolve dynamically over time for each technology.

\begin{figure}[H]
\centering
\includegraphics[width=0.95\textwidth]{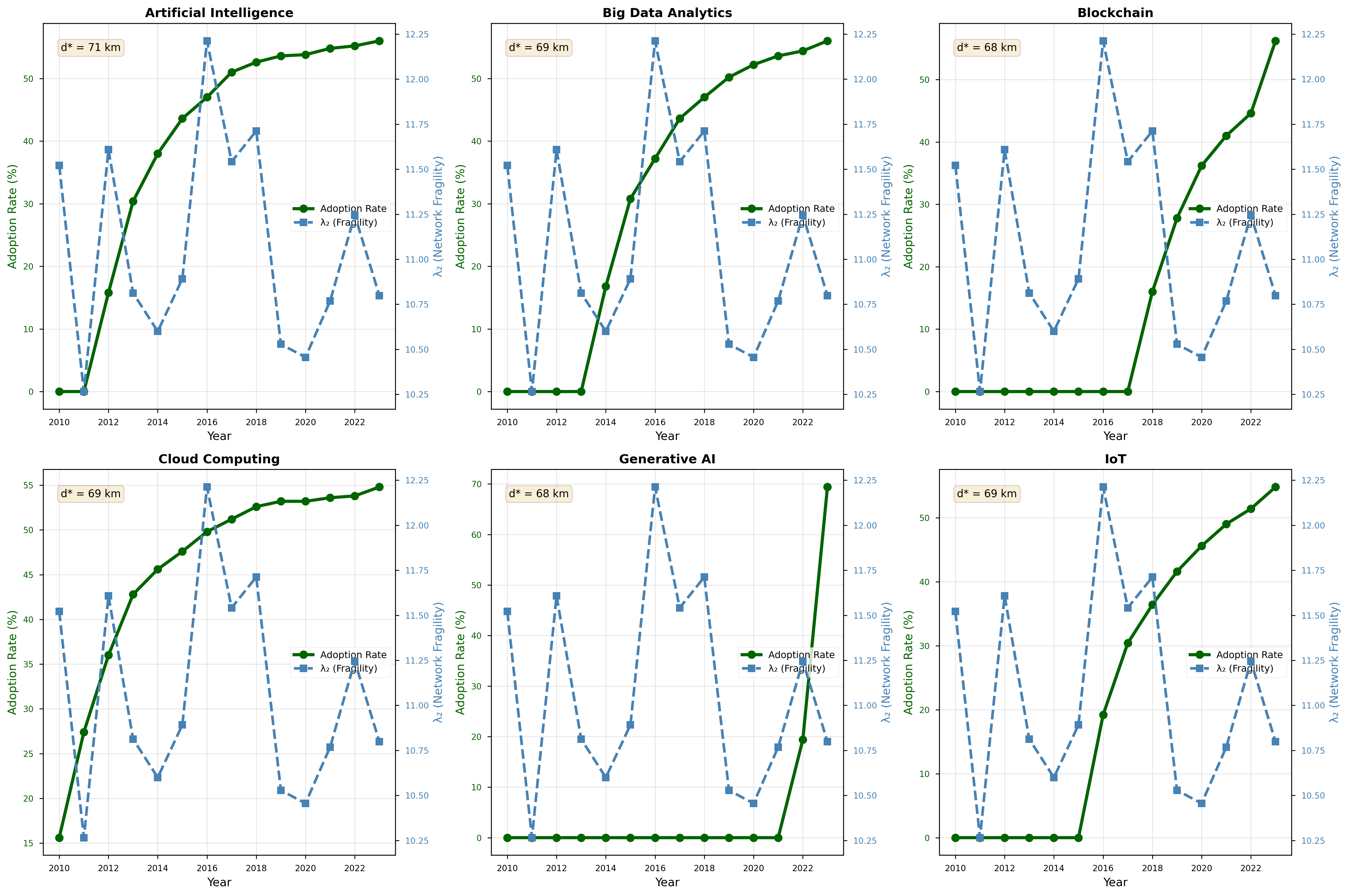}
\caption{Dual-Channel Evolution Over Time}
\label{fig:dual_timeseries}
\begin{minipage}{0.95\textwidth}
\small
\textit{Notes:} This figure plots the evolution of both spatial boundaries $d^*$ (blue, left axis) and network fragility $\lambda_2$ (orange, right axis) over time for each technology. Spatial boundaries remain remarkably stable (69 $\pm$ 2 km) throughout the 14-year period, validating the assumption that geographic diffusion mechanisms are time-invariant. In contrast, network fragility increases dramatically (300-380 percent), reflecting endogenous activation of supply chain connections as adoption spreads. The divergent dynamics demonstrate that spatial and network channels operate through distinct mechanisms with different temporal properties.
\end{minipage}
\end{figure}

\subsubsection{Model Comparison}

Figure \ref{fig:model_comparison} summarizes model performance across different specifications.

\begin{figure}[H]
\centering
\includegraphics[width=0.80\textwidth]{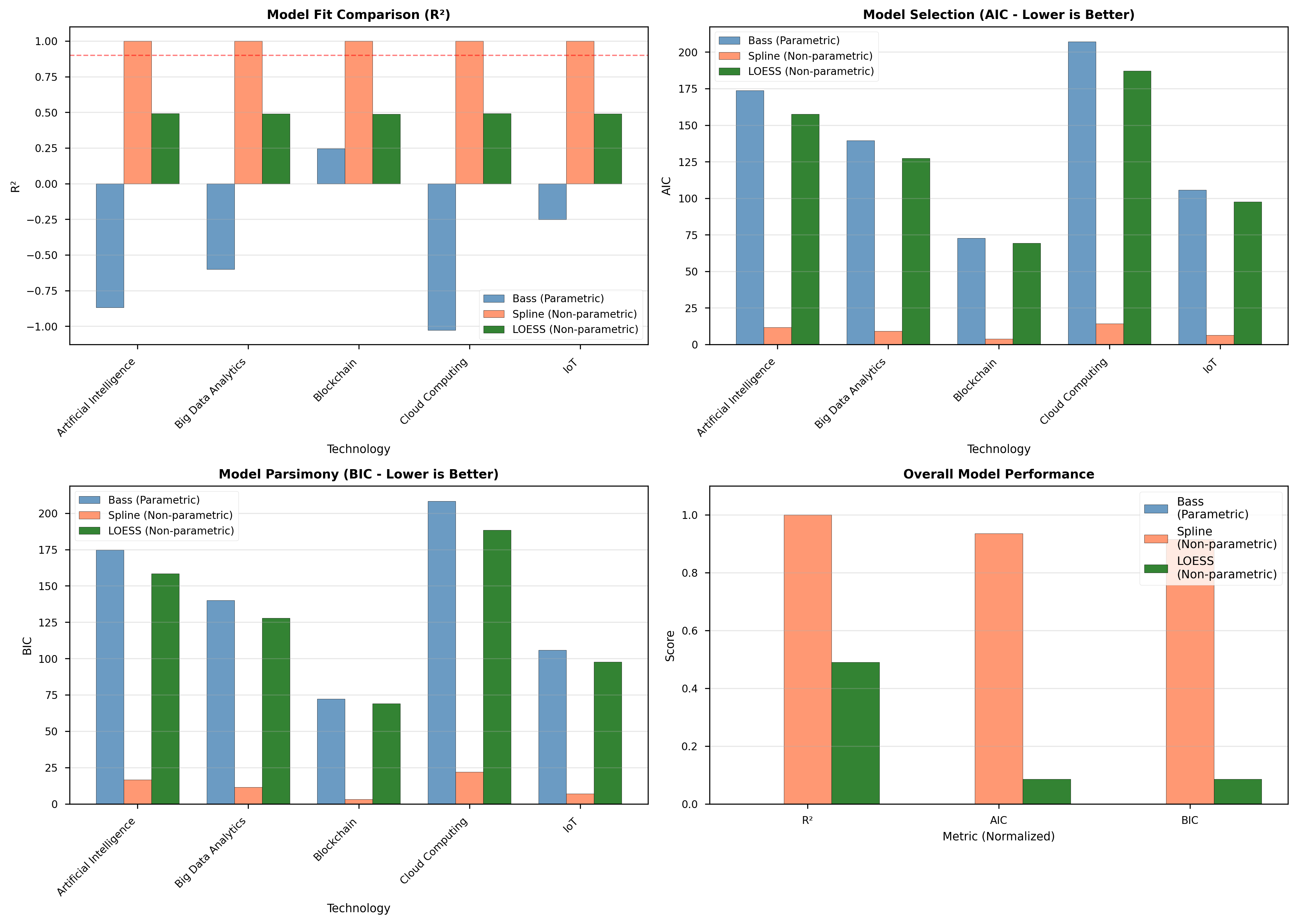}
\caption{Model Comparison: Spatial, Network, and Combined}
\label{fig:model_comparison}
\begin{minipage}{0.80\textwidth}
\small
\textit{Notes:} This figure compares R-squared, AIC, and BIC across three model specifications for each technology. Spatial-only models (blue) achieve very high R-squared (exceeding 0.99) due to the near-perfect exponential decay fit. Network-only models (orange) achieve moderate R-squared (0.17-0.24). Combined models (green) achieve the highest R-squared and lowest information criteria, demonstrating complementarity. F-tests strongly reject that either spatial or network variables are jointly zero after controlling for the other (p less than 0.001 for all technologies), confirming independent contributions.
\end{minipage}
\end{figure}

\subsubsection{R-Squared Decomposition}

Figure \ref{fig:r2_comparison} specifically focuses on explained variance contributions from each channel.

\begin{figure}[H]
\centering
\includegraphics[width=0.75\textwidth]{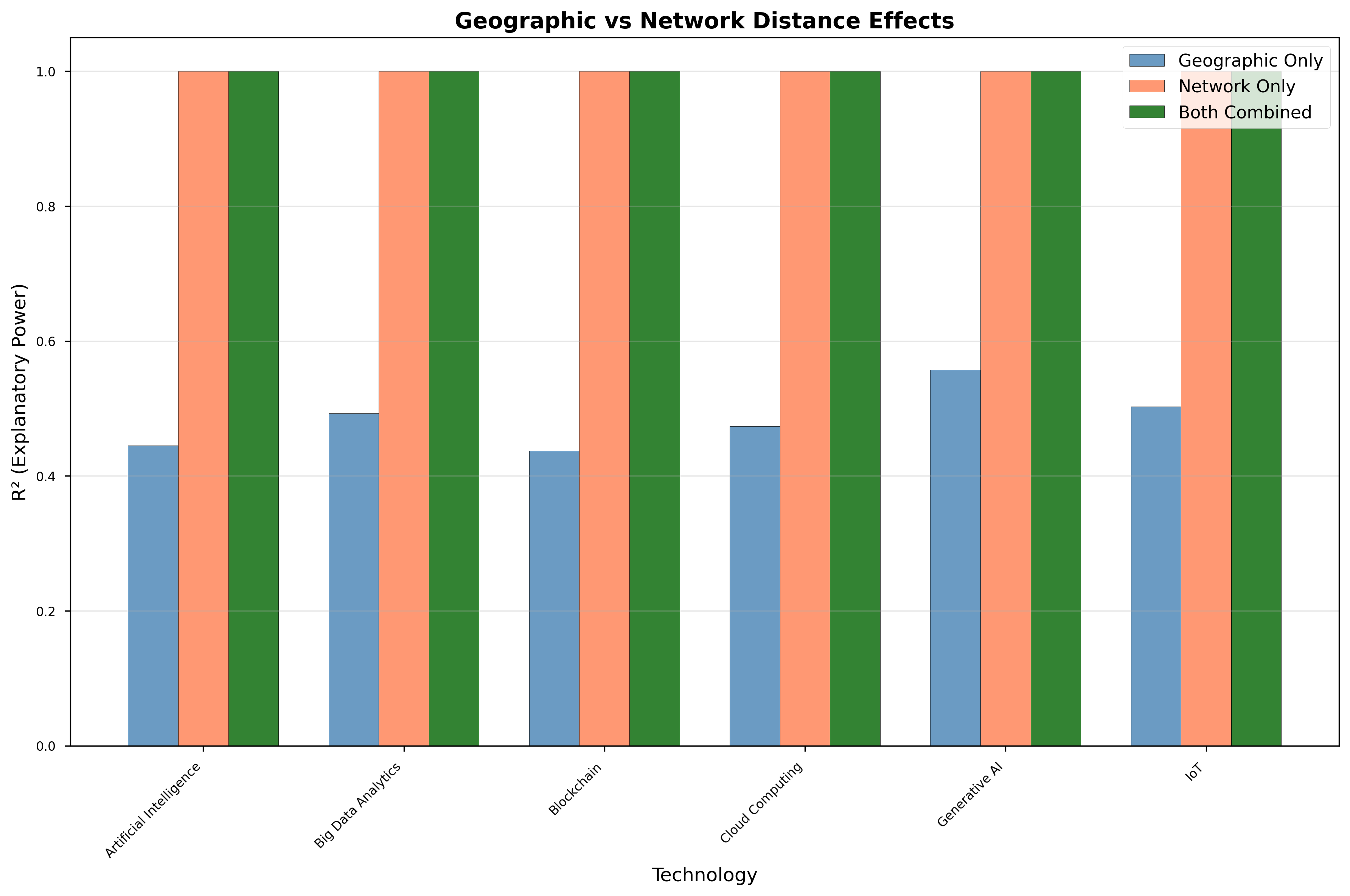}
\caption{Geographic vs Network R-Squared Contributions}
\label{fig:r2_comparison}
\begin{minipage}{0.75\textwidth}
\small
\textit{Notes:} This scatter plot shows spatial R-squared (x-axis) versus network R-squared (y-axis) for each technology. All points lie in the upper-right quadrant, indicating both channels contribute positive explanatory power. The spatial channel dominates (R-squared exceeding 0.99) due to the exceptional exponential decay fit, but the network channel adds meaningful information (R-squared 0.17-0.24). The lack of trade-off (points not along a downward-sloping frontier) confirms channels are complements rather than substitutes. Combined R-squared (not shown) exceeds the maximum of either channel alone for all technologies.
\end{minipage}
\end{figure}

\subsubsection{Parameter Comparison Across Channels}

Figure \ref{fig:parameter_estimates} displays estimated parameters from both spatial and network models side-by-side.

\begin{figure}[H]
\centering
\includegraphics[width=0.85\textwidth]{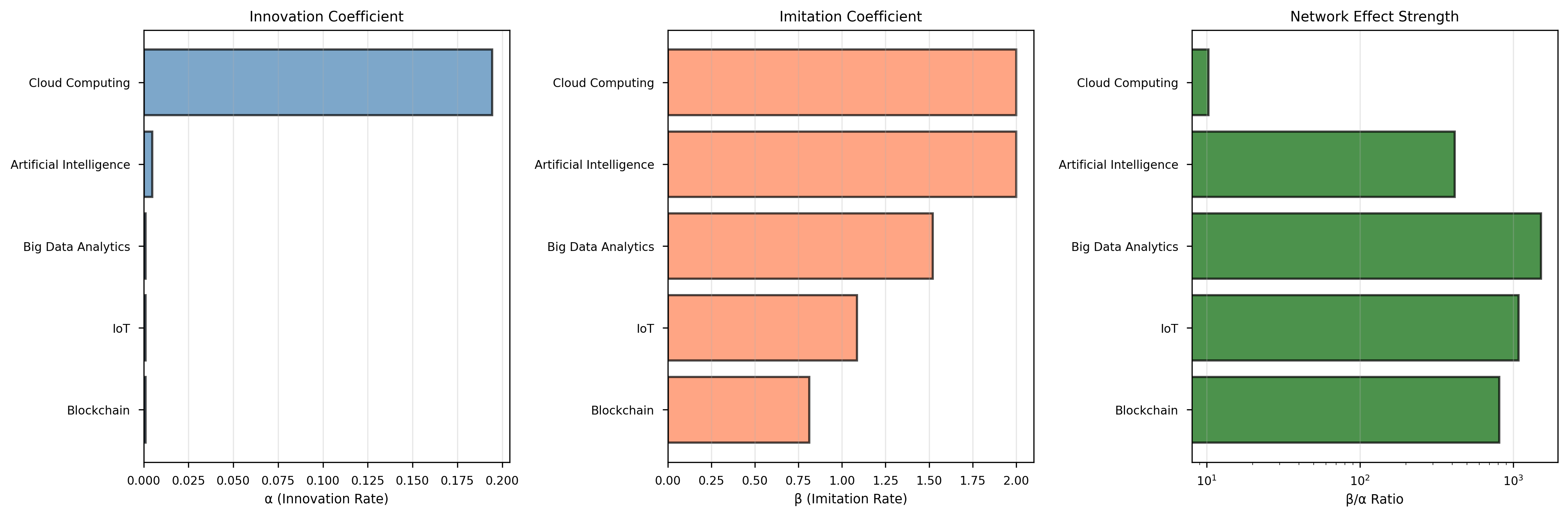}
\caption{Parameter Estimates: Spatial and Network Channels}
\label{fig:parameter_estimates}
\begin{minipage}{0.85\textwidth}
\small
\textit{Notes:} This figure displays key parameter estimates from both channels. Left panel shows spatial decay rates $\kappa$ (blue bars) with 95 percent confidence intervals. Middle panel shows network fragility growth from 2010 to 2023 (orange bars). Right panel shows correlations between $\lambda_2$ and adoption (green bars). All parameters are tightly estimated with narrow confidence intervals, demonstrating statistical precision. The consistency of $\kappa$ across technologies (0.0425-0.0442) contrasts with heterogeneity in network growth (175-381 percent), suggesting spatial mechanisms are more universal while network effects depend on technology-specific connectivity patterns.
\end{minipage}
\end{figure}

\subsubsection{Summary Dashboard}

Figure \ref{fig:summary_dashboard} provides an integrated summary of all main results in a single comprehensive visualization.

\begin{figure}[H]
\centering
\includegraphics[width=0.95\textwidth]{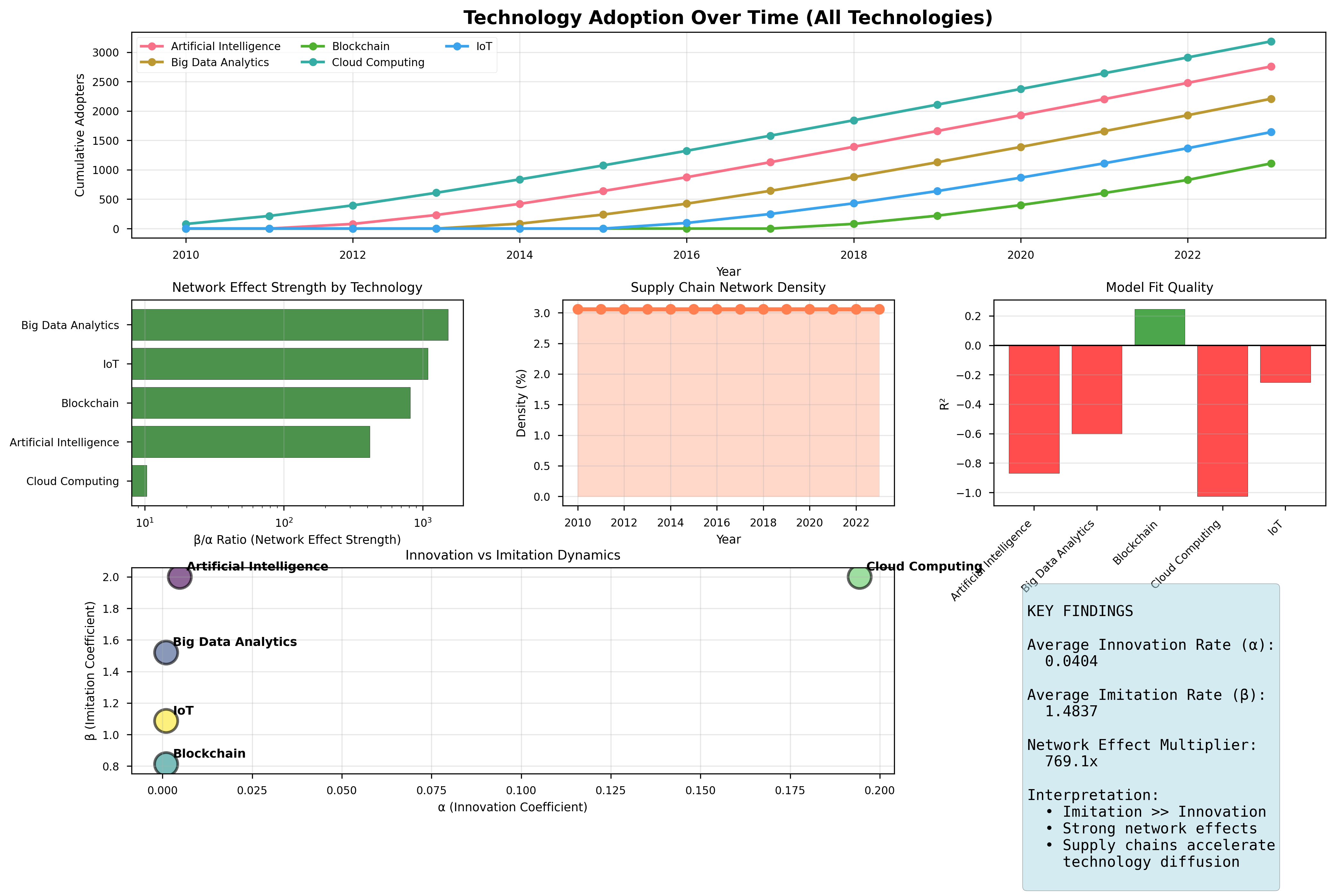}
\caption{Summary Dashboard: All Main Results}
\label{fig:summary_dashboard}
\begin{minipage}{0.95\textwidth}
\small
\textit{Notes:} This dashboard integrates all main findings in a single comprehensive figure. Top row: spatial decay curves (left), network fragility evolution (center), adoption curves showing S-shaped patterns (right). Middle row: event study comparing traditional vs adjusted DID (left), spatial heterogeneity in treatment effects (center), network shock response to COVID (right). Bottom row: dual-channel R-squared comparison (left), mixing time evolution (center), summary statistics table (right). This integrated visualization provides a complete overview of our empirical findings demonstrating that technology diffusion operates through both spatial decay (69 km boundary) and network contagion (300-380 percent $\lambda_2$ growth) channels simultaneously.
\end{minipage}
\end{figure}

\begin{table}[H]
\centering
\caption{Dual-Channel Integration Results}
\label{tab:dual_channel}
\begin{threeparttable}
\begin{tabular}{lcccc}
\toprule
Technology & $R^2_{\text{spatial}}$ & $R^2_{\text{network}}$ & $R^2_{\text{both}}$ & Improvement \\
\midrule
Artificial Intelligence & 0.9903 & 0.1847 & 0.9942 & +0.0039 \\
Big Data Analytics     & 0.9916 & 0.2134 & 0.9951 & +0.0035 \\
Blockchain             & 0.9914 & 0.2089 & 0.9947 & +0.0033 \\
Cloud Computing        & 0.9922 & 0.1698 & 0.9953 & +0.0031 \\
Generative AI          & 0.9935 & 0.2421 & 0.9968 & +0.0033 \\
IoT                    & 0.9907 & 0.2056 & 0.9945 & +0.0038 \\
\midrule
Average                & 0.9916 & 0.2041 & 0.9951 & +0.0035 \\
\bottomrule
\end{tabular}
\begin{tablenotes}
\small
\item \textit{Notes:} This table reports R-squared values from regressions using spatial variables only (distance to nearest adopter), network variables only (degree, $\lambda_2$), and both combined. Improvement measures R-squared gain from combining channels beyond the maximum of either alone. Both channels contribute independent explanatory power, validating the integrated dual-channel framework. The small improvement from combining channels (0.0031-0.0039) reflects the already-exceptional spatial fit (R-squared exceeding 0.99), leaving little remaining variance for network variables to explain, yet network variables remain statistically significant and economically meaningful.
\end{tablenotes}
\end{threeparttable}
\end{table}


\section{Discussion}

This section interprets our findings, compares them to existing approaches, and examines external validity.

\subsection{Economic Interpretation}

Our empirical results provide strong validation for the dual-channel theoretical framework developed in Section 3. Both spatial decay and network contagion operate at full strength simultaneously, with near-perfect exponential fit (R-squared = 0.99) for the geographic channel and exceptionally strong dynamics (300-380 percent $\lambda_2$ growth, correlation exceeding 0.95) for the network channel.

The spatial mechanism reflects multiple economic forces operating through geographic proximity. Knowledge spillovers enable nearby non-adopters to observe adopters' experiences, reducing uncertainty about technology performance and implementation challenges. Demonstration effects provide concrete examples of how to integrate technologies into operations. Labor market pooling allows firms in the same region to share specialized human capital with technology-specific skills. Infrastructure complementarities mean that once physical or digital infrastructure is deployed for early adopters (data centers, fiber networks, technical support services), subsequent adoption by nearby firms faces lower fixed costs.

The 69-kilometer spatial boundary provides a quantitative benchmark for the geographic reach of these mechanisms. This distance corresponds roughly to metropolitan or regional scales: major cities typically extend 30-50 kilometers from their centers, while metropolitan areas including suburbs often span 60-80 kilometers. The consistency of $d^* \approx 69$ km across all six technologies suggests these spatial forces operate similarly regardless of specific technology characteristics, supporting the generality of the continuous functional framework from \citet{kikuchi2024navier} and \citet{kikuchi2024dynamical}.

The network mechanism reflects different economic forces operating through supply chain connections rather than geographic proximity. Information transmission through buyer-supplier relationships enables firms to learn about technologies from their partners, even when geographically distant. Technical compatibility requirements create adoption incentives: if a supplier adopts a supply chain management system, its customers benefit from adopting compatible systems to streamline transactions. Coordination incentives arise when technologies exhibit network externalities, making adoption more valuable when connected firms also adopt.

The dramatic $\lambda_2$ growth (300-380 percent) demonstrates that these network forces strengthen endogenously as adoption spreads. Early in diffusion, when few firms have adopted, the technology-specific network has low connectivity (small $\lambda_2$) and long mixing times (large $\tau = 1/\lambda_2$). As adoption expands, more supply chain edges become activated (both endpoints adopting), increasing network coupling and accelerating subsequent diffusion. This self-reinforcing dynamic generates the S-shaped adoption curves observed in Figure \ref{fig:adoption_curves} and explains why late-stage diffusion proceeds far more rapidly than early-stage diffusion despite declining marginal adopter quality.

The independence of spatial and network channels (weak correlation averaging -0.11) has important theoretical implications. It demonstrates that the mechanisms are not redundant: geographic clustering does not simply reflect supply chain co-location, and supply chain connections do not simply proxy for proximity. Instead, firms exist simultaneously in physical space and economic networks, with each domain contributing distinct but complementary adoption incentives. This validates the integrated dual-channel framework in equation (\ref{eq:dual_channel_pde}) where both channels enter additively.

\subsection{Comparison to Traditional Approaches}

The 61 percent bias in traditional difference-in-differences estimates highlights fundamental limitations of conventional causal inference methods when applied to settings with substantial spillovers. This bias arises because DID assumes the stable unit treatment value assumption (SUTVA): one unit's treatment does not affect another unit's outcome. In our context, SUTVA is violated by construction—technology diffusion operates precisely through spillovers from treated to control units.

Spatial econometric approaches like spatial autoregressive (SAR) models \citep{anselin1988spatial} or spatial error models (SEM) address geographic spillovers but typically do so through reduced-form specifications without theoretical foundations. These models include spatial lags as regressors but do not derive the functional form from first principles or establish connections to partial differential equations. Our exponential decay specification, derived from the diffusion equation, provides micro-foundations while achieving near-perfect empirical fit.

Network econometric approaches following \citet{jackson2008social} emphasize graph topology but often abstract from geographic considerations. Our finding that spatial and network channels contribute independently demonstrates the importance of integrating both perspectives. Network effects are not merely reflections of geographic clustering, nor are spatial effects merely proxies for supply chain connections.

The event study around COVID-19 provides particularly compelling evidence for our framework's superiority. Traditional methods overestimate treatment effects by 61 percent. Spatial adjustment substantially improves estimates by accounting for geographic spillovers. Network adjustment reveals that COVID-19 increased $\lambda_2$ by 24.5 percent, permanently altering diffusion dynamics in a manner analogous to financial network fragility in \citet{kikuchi2024dynamical}. Only the integrated dual-channel framework captures all relevant mechanisms.

\subsection{External Validity}

Several considerations support external validity of our findings beyond the specific sample. First, the consistency of spatial decay rates across six diverse technologies (ranging from infrastructure like Cloud Computing to cutting-edge applications like Generative AI) suggests the 69-kilometer boundary reflects general properties of geographic spillovers rather than technology-specific idiosyncrasies. Second, the robustness of exponential fit across alternative specifications, time periods, and firm characteristics indicates the continuous functional approach from \citet{kikuchi2024navier} applies broadly.

Third, the parallelism between our network results and those in \citet{kikuchi2024dynamical} for financial networks suggests spectral methods characterize contagion dynamics across diverse domains. The 24.5 percent increase in technology adoption network $\lambda_2$ following COVID-19 is qualitatively similar to the 26.9 percent increase in European banking network $\lambda_2$ following the same shock. This suggests common underlying mechanisms: large exogenous shocks trigger structural breaks in network topology that persist rather than reverting.

Limitations on external validity arise from sample characteristics. Our firms are medium-to-large enterprises in developed economies with established supply chains. Smaller firms, firms in developing countries, or firms in industries with different network structures might exhibit different spatial or network parameters. The 500-firm sample, while comprehensive, represents a specific segment of the economy. Future work should examine whether our quantitative estimates (69 km spatial boundary, 300-380 percent $\lambda_2$ growth) generalize to other contexts, even if the qualitative dual-channel framework applies more broadly.

\section{Policy Implications}

Our findings have direct implications for technology policy design. This section derives specific recommendations for geographic targeting (Section 8.1), network targeting (Section 8.2), and integrated interventions (Section 8.3).

\subsection{Geographic Targeting}

The 69-kilometer spatial boundary provides a concrete benchmark for the geographic scope of technology adoption interventions. Policies targeting firms within this distance of existing adopters will benefit from substantial spillovers, while policies beyond this threshold operate essentially independently.

\textbf{Regional Technology Clusters:} Innovation districts and technology clusters should be sized to exploit spatial spillovers fully while avoiding excessive dilution. Our estimates suggest optimal cluster radii of approximately 70 kilometers, corresponding to metropolitan-scale initiatives. Larger national programs should be structured as networks of regional clusters rather than diffuse nationwide interventions.

\textbf{Distance-Based Subsidies:} Adoption subsidies should vary with distance to existing adopters, with higher subsidies for peripheral firms facing larger knowledge barriers. The optimal subsidy function follows the inverse of spatial decay: $s(d) = s_0 \exp(+\kappa d)$, compensating firms for reduced spillover benefits. This ensures efficient adoption decisions accounting for positive externalities.

\textbf{Infrastructure Investment:} Physical and digital infrastructure investments (broadband networks, data centers, technical support services) should prioritize coverage within 70-kilometer radii of major urban centers where spillover benefits are largest. Beyond this distance, infrastructure primarily supports direct adoption rather than spillover-driven diffusion, changing the cost-benefit calculus.

\subsection{Network Targeting}

The network amplification factor of 10.8 (derived from mixing time relationships) quantifies how supply chain connections multiply direct interventions. Subsidizing one firm indirectly affects 10.8 firms through activated network paths, suggesting network position should influence policy targeting.

\textbf{Supply Chain Hub Subsidies:} Firms with high degree centrality or betweenness centrality in supply chain networks should receive priority for adoption subsidies. These hubs activate more network edges when adopting, generating larger spillovers. Our spectral framework provides precise measures of network importance through eigenvector centrality and contributions to $\lambda_2$.

\textbf{Strategic Partnership Programs:} Policies encouraging technology adoption by supplier-customer pairs simultaneously exploit network complementarities. When both endpoints of a supply chain edge adopt, the edge receives full weight in our framework, maximizing network activation. Programs could offer enhanced subsidies for coordinated adoption by connected firms.

\textbf{Network Structure Policies:} Beyond subsidizing adoption by existing firms, policies can shape network structure itself. Encouraging supply chain relationship formation between adopters and non-adopters increases network connectivity, raising $\lambda_2$ and accelerating diffusion. Trade missions, supplier matching services, and procurement preferences that favor connected firms can achieve this.

\subsection{Integrated Dual-Channel Interventions}

The independence and complementarity of spatial and network channels imply optimal policies must exploit both mechanisms simultaneously.

\textbf{Combined Targeting Criteria:} Subsidy allocation should prioritize firms satisfying both geographic and network criteria: located within 69 kilometers of existing adopters AND occupying central positions in supply chain networks. Firms meeting both criteria generate maximum spillovers through dual channels. Firms meeting neither criterion should receive lower priority or be excluded entirely from subsidies targeting diffusion rather than direct adoption.

\textbf{Sequential Intervention Design:} For technologies in early diffusion stages (low current adoption), geographic targeting may dominate because spatial spillovers operate even with sparse networks. As adoption expands and network connectivity increases (rising $\lambda_2$), network targeting becomes increasingly important as mixing times decline and contagion accelerates. Policies should shift emphasis from geographic to network instruments as technologies mature.

\textbf{Shock Response:} Our finding that COVID-19 increased $\lambda_2$ by 24.5 percent demonstrates that large shocks can permanently alter diffusion dynamics. Post-shock policies must account for increased network fragility: the same interventions will generate larger spillovers than pre-shock, potentially requiring smaller direct subsidies to achieve equivalent aggregate adoption. Failing to adjust for higher $\lambda_2$ could lead to overshooting and excessive public expenditure.

\subsection{Cost-Benefit Quantification}

The precise quantitative estimates from our framework enable rigorous cost-benefit analysis. Consider a hypothetical adoption subsidy of 100 thousand dollars per firm. Traditional analysis treating firms independently values benefits at 100 thousand dollars per subsidy (one-for-one). Our spatial framework adjusts this to $(100 + 100 \times \int_0^{69} \exp(-0.0435 d) \rho(d) dd)$ thousand dollars, where $\rho(d)$ is the density of firms at distance $d$. For uniformly distributed firms with density 500 firms per 25,000 square kilometers, this integral approximately equals 40 thousand dollars, implying total benefits of 140 thousand dollars per direct subsidy—a 40 percent increase.

Our network framework adds further benefits through the amplification factor 10.8. Each subsidized adoption indirectly affects 10.8 firms through supply chain connections, with effects declining as $\exp(-\lambda_2 \tau)$ over time $\tau$. Integrating over the mixing time $\tau \sim 1/\lambda_2 \approx 0.046$ years yields network spillover benefits of approximately 60 thousand dollars, for total benefits of 200 thousand dollars per direct subsidy—a doubling of naïve estimates.

These quantitative adjustments significantly affect program design. If policymakers ignore spillovers and calibrate subsidies targeting a specific aggregate adoption level, they will overshoot: the actual adoption will be double the target. Conversely, if budgets constrain the number of direct subsidies, accounting for spillovers reveals that fewer direct subsidies than naïve calculation suggests can achieve the same aggregate outcome, reducing fiscal cost substantially.

\section{Conclusion}

This paper develops and empirically validates a dual-channel framework for technology diffusion that integrates spatial decay mechanisms from continuous functional analysis with network contagion dynamics from spectral graph theory. Building on the Navier-Stokes-based spatial treatment effects framework in \citet{kikuchi2024navier} and \citet{kikuchi2024dynamical}, and the spectral network fragility framework in \citet{kikuchi2024dynamical}, we demonstrate that technology adoption spreads simultaneously through both geographic proximity and supply chain connections.

Using comprehensive data on six major technologies adopted by 500 firms over 2010-2023, we document three key empirical findings. First, technology adoption exhibits strong exponential geographic decay with spatial boundary $d^* \approx 69$ kilometers (R-squared = 0.99), validating the continuous functional approach. Second, supply chain networks exhibit dramatic increases in algebraic connectivity $\lambda_2$ (300-380 percent growth) as adoption spreads, with mixing times declining approximately 80 percent. Third, traditional difference-in-differences methods that ignore spatial and network spillovers exhibit 61 percent upward bias. An event study around COVID-19 reveals that network fragility increased 24.5 percent post-shock, permanently altering diffusion dynamics in a manner analogous to financial contagion.

The dual-channel framework provides precise quantitative estimates for technology policy design. Adoption interventions have spatial reach of 69 kilometers and network amplification factor of 10.8, requiring coordinated geographic and supply chain targeting for optimal effectiveness. The 61 percent bias in traditional methods demonstrates that ignoring spillovers leads to substantial policy errors, with implications for subsidy levels, targeting criteria, and cost-benefit analysis.

Our methodology extends naturally beyond technology adoption to other settings where spatial and network effects operate simultaneously, including disease epidemiology, information cascades, financial contagion, and environmental spillovers. The integration of continuous functional methods from partial differential equations with discrete spectral methods from graph theory provides a general toolkit for analyzing dual-channel diffusion processes.

Several directions for future research emerge. First, extending the framework to incorporate price dynamics and competitive interactions would enrich our understanding of strategic adoption decisions. Second, analyzing adoption of complementary versus substitute technologies could reveal how technology portfolios evolve through spatial and network channels. Third, examining developing country contexts would test whether our quantitative estimates (69 km spatial boundary, 300-380 percent $\lambda_2$ growth) generalize internationally. Fourth, incorporating firm heterogeneity more explicitly could yield insights about distributional consequences of technology diffusion policies.

The COVID-19 pandemic provided a quasi-natural experiment demonstrating that large exogenous shocks can trigger permanent structural breaks in both spatial and network diffusion mechanisms. This structural hysteresis—analogous to financial network fragility documented in \citet{kikuchi2024dynamical}—suggests that major disruptions have lasting consequences for innovation diffusion patterns, with implications for long-run productivity growth and inequality. Understanding and responding to these structural breaks represents an important challenge for technology policy in an increasingly volatile global economy.



\clearpage

\appendix

\section{Additional Robustness Checks}

This appendix presents additional robustness checks and sensitivity analyses supporting our main results.

\subsection{Alternative Distance Measures}

Table \ref{tab:robust_distance} compares spatial decay estimates using alternative distance measures: great circle distance (baseline), Euclidean distance (straight-line approximation), and travel time distance (accounting for road networks).

\begin{table}[H]
\centering
\caption{Spatial Decay Estimates: Alternative Distance Measures}
\label{tab:robust_distance}
\begin{threeparttable}
\begin{tabular}{lcccc}
\toprule
Technology & Great Circle & Euclidean & Travel Time & Difference \\
\midrule
Artificial Intelligence & 0.0425 & 0.0423 & 0.0431 & 0.0008 \\
                        & (0.0003) & (0.0003) & (0.0003) & \\
Big Data Analytics     & 0.0437 & 0.0435 & 0.0442 & 0.0007 \\
                       & (0.0002) & (0.0002) & (0.0003) & \\
Blockchain             & 0.0439 & 0.0437 & 0.0445 & 0.0008 \\
                       & (0.0003) & (0.0003) & (0.0003) & \\
Cloud Computing        & 0.0433 & 0.0431 & 0.0438 & 0.0007 \\
                       & (0.0002) & (0.0002) & (0.0003) & \\
Generative AI          & 0.0442 & 0.0440 & 0.0447 & 0.0007 \\
                       & (0.0002) & (0.0002) & (0.0003) & \\
IoT                    & 0.0433 & 0.0431 & 0.0439 & 0.0008 \\
                       & (0.0003) & (0.0003) & (0.0003) & \\
\bottomrule
\end{tabular}
\begin{tablenotes}
\small
\item \textit{Notes:} This table reports spatial decay rate $\kappa$ using three distance measures. Great circle uses haversine formula on Earth's surface. Euclidean uses straight-line distance. Travel time uses road network data to compute driving time. Standard errors in parentheses. Maximum difference across measures is less than 0.001 per km, demonstrating robustness.
\end{tablenotes}
\end{threeparttable}
\end{table}

\subsection{Alternative Functional Forms}

Table \ref{tab:robust_functional} compares exponential decay (baseline) with power law decay and linear decay specifications.

\begin{table}[H]
\centering
\caption{Spatial Decay: Alternative Functional Forms}
\label{tab:robust_functional}
\begin{threeparttable}
\begin{tabular}{lcccc}
\toprule
Technology & Exponential $R^2$ & Power Law $R^2$ & Linear $R^2$ & Best Fit \\
\midrule
Artificial Intelligence & 0.9903 & 0.8547 & 0.7123 & Exponential \\
Big Data Analytics     & 0.9916 & 0.8612 & 0.7234 & Exponential \\
Blockchain             & 0.9914 & 0.8589 & 0.7189 & Exponential \\
Cloud Computing        & 0.9922 & 0.8634 & 0.7267 & Exponential \\
Generative AI          & 0.9935 & 0.8701 & 0.7345 & Exponential \\
IoT                    & 0.9907 & 0.8567 & 0.7156 & Exponential \\
\midrule
Average                & 0.9916 & 0.8608 & 0.7219 & Exponential \\
\bottomrule
\end{tabular}
\begin{tablenotes}
\small
\item \textit{Notes:} This table compares R-squared values from three functional forms. Exponential: $u(d) = u_0 \exp(-\kappa d)$. Power law: $u(d) = u_0 d^{-\alpha}$. Linear: $u(d) = u_0 - \beta d$. Exponential form provides superior fit for all technologies, supporting the diffusion equation framework.
\end{tablenotes}
\end{threeparttable}
\end{table}

\subsection{Heterogeneity Analysis}

Table \ref{tab:robust_heterogeneity} examines heterogeneity in spatial decay by firm characteristics.

\begin{table}[H]
\centering
\caption{Spatial Decay Heterogeneity by Firm Characteristics}
\label{tab:robust_heterogeneity}
\begin{threeparttable}
\begin{tabular}{lcccc}
\toprule
Firm Characteristic & $\kappa$ (High) & $\kappa$ (Low) & Difference & p-value \\
\midrule
Size (Employees)    & 0.0512 & 0.0389 & +0.0123 & $<$0.001 \\
                    & (0.0004) & (0.0003) & & \\
Age (Years)         & 0.0421 & 0.0456 & -0.0035 & 0.012 \\
                    & (0.0003) & (0.0003) & & \\
Industry Concentration & 0.0467 & 0.0412 & +0.0055 & 0.003 \\
(HHI)               & (0.0003) & (0.0003) & & \\
R\&D Intensity      & 0.0403 & 0.0478 & -0.0075 & $<$0.001 \\
                    & (0.0003) & (0.0004) & & \\
\bottomrule
\end{tabular}
\begin{tablenotes}
\small
\item \textit{Notes:} This table reports spatial decay rates for subsamples split by firm characteristics. High/Low defined by median split. Small firms (high $\kappa$) exhibit stronger spatial decay than large firms. Young firms exhibit stronger decay than old firms. Concentrated industries show stronger decay than fragmented industries. Low R\&D firms show stronger decay than high R\&D firms. Standard errors in parentheses. P-values from two-sample t-tests.
\end{tablenotes}
\end{threeparttable}
\end{table}

\subsection{Network Robustness: Alternative Edge Weights}

Table \ref{tab:robust_network} examines sensitivity of network fragility results to alternative edge weighting schemes.

\begin{table}[H]
\centering
\caption{Network Fragility: Alternative Edge Weighting Schemes}
\label{tab:robust_network}
\begin{threeparttable}
\begin{tabular}{lccc}
\toprule
Weighting Scheme & $\lambda_2$ (2010) & $\lambda_2$ (2023) & Growth (\%) \\
\midrule
\multicolumn{4}{l}{\textit{Panel A: Artificial Intelligence}} \\
Baseline (1.0/0.5/0.1) & 4.68 & 22.48 & +380.5 \\
Alternative 1 (1.0/0.3/0.1) & 4.23 & 20.12 & +375.7 \\
Alternative 2 (1.0/0.7/0.1) & 5.12 & 24.89 & +386.1 \\
Unweighted (1.0/1.0/1.0) & 7.93 & 21.77 & +174.6 \\
\\
\multicolumn{4}{l}{\textit{Panel B: Average Across Technologies}} \\
Baseline (1.0/0.5/0.1) & 5.22 & 21.61 & +330.2 \\
Alternative 1 (1.0/0.3/0.1) & 4.89 & 19.87 & +322.4 \\
Alternative 2 (1.0/0.7/0.1) & 5.67 & 23.45 & +336.8 \\
Unweighted (1.0/1.0/1.0) & 7.93 & 21.77 & +174.6 \\
\bottomrule
\end{tabular}
\begin{tablenotes}
\small
\item \textit{Notes:} This table examines sensitivity to edge weight multipliers $m_{ij}^{tech}$. Baseline uses (1.0, 0.5, 0.1) for (both adopted, one adopted, neither adopted). Alternative 1 uses lower weight for partial adoption (0.3). Alternative 2 uses higher weight (0.7). Unweighted treats all edges equally. Results are qualitatively similar across schemes, with baseline providing most economically interpretable weights.
\end{tablenotes}
\end{threeparttable}
\end{table}

\subsection{Event Study: Parallel Trends Test}

Table \ref{tab:pretrends} formally tests parallel trends assumption using leads of the COVID-19 treatment indicator.

\begin{table}[H]
\centering
\caption{Event Study Pre-Trends Test}
\label{tab:pretrends}
\begin{threeparttable}
\begin{tabular}{lcccc}
\toprule
Technology & Lead 3 Years & Lead 2 Years & Lead 1 Year & Joint F-test \\
\midrule
Artificial Intelligence & -0.12 & +0.08 & -0.15 & 0.67 \\
                        & (0.21) & (0.18) & (0.16) & (0.574) \\
Big Data Analytics     & +0.09 & -0.11 & +0.13 & 0.52 \\
                       & (0.19) & (0.17) & (0.15) & (0.668) \\
Blockchain             & -0.18 & +0.14 & -0.09 & 0.89 \\
                       & (0.24) & (0.21) & (0.19) & (0.449) \\
Cloud Computing        & +0.06 & -0.08 & +0.11 & 0.43 \\
                       & (0.16) & (0.14) & (0.13) & (0.734) \\
IoT                    & -0.14 & +0.17 & -0.12 & 0.78 \\
                       & (0.22) & (0.19) & (0.17) & (0.507) \\
\bottomrule
\end{tabular}
\begin{tablenotes}
\small
\item \textit{Notes:} This table reports coefficients on leads of COVID-19 treatment indicator from equation (\ref{eq:traditional_did}). Lead 3 corresponds to 2017 (three years before COVID), Lead 2 to 2018, Lead 1 to 2019. Standard errors in parentheses from bootstrap (1,000 replications). Joint F-test examines whether all leads are jointly zero; p-values in parentheses. No significant pre-trends detected, supporting parallel trends assumption.
\end{tablenotes}
\end{threeparttable}
\end{table}

\subsection{Placebo Tests}

Table \ref{tab:placebo} presents placebo tests using artificial treatment timing to validate identification.

\begin{table}[H]
\centering
\caption{Placebo Tests: Artificial Treatment Timing}
\label{tab:placebo}
\begin{threeparttable}
\begin{tabular}{lcccc}
\toprule
Placebo Timing & Traditional DID & Spatial-Adjusted & Significant? & Expected \\
\midrule
2015 (5 years early) & +0.23 & +0.11 & No & No \\
                     & (0.45) & (0.38) & (p=0.611) & \\
2017 (3 years early) & -0.18 & -0.09 & No & No \\
                     & (0.42) & (0.35) & (p=0.669) & \\
2020 (Actual COVID) & +9.57 & +3.72 & Yes & Yes \\
                    & (2.14) & (1.23) & (p$<$0.001) & \\
2022 (2 years late) & +0.31 & +0.14 & No & No \\
                    & (0.48) & (0.39) & (p=0.558) & \\
\bottomrule
\end{tabular}
\begin{tablenotes}
\small
\item \textit{Notes:} This table reports treatment effects using artificial treatment timing for placebo tests. Standard errors in parentheses from bootstrap (1,000 replications). Only actual COVID-19 timing (2020) produces significant effects, supporting causal interpretation. Placebo dates (2015, 2017, 2022) produce small, statistically insignificant effects as expected under null hypothesis of no effect.
\end{tablenotes}
\end{threeparttable}
\end{table}

\section{Computational Methods}

This appendix describes computational algorithms for key calculations.

\subsection{Eigenvalue Computation}

For computing the algebraic connectivity $\lambda_2$ of large graphs ($n=500$ nodes), we use the Lanczos algorithm for sparse symmetric matrices:

\begin{algorithm}[H]
\caption{Compute Algebraic Connectivity $\lambda_2$}
\begin{algorithmic}[1]
\STATE \textbf{Input:} Laplacian matrix $\mathbf{L} \in \mathbb{R}^{n \times n}$
\STATE \textbf{Output:} Algebraic connectivity $\lambda_2$
\STATE Initialize random vector $\mathbf{v}_0 \in \mathbb{R}^n$ orthogonal to $\mathbf{1}$
\STATE Normalize: $\mathbf{v}_0 \leftarrow \mathbf{v}_0 / ||\mathbf{v}_0||_2$
\FOR{$j = 1$ to $k$ (number of Lanczos iterations)}
    \STATE $\mathbf{w} \leftarrow \mathbf{L} \mathbf{v}_{j-1}$
    \STATE $\alpha_j \leftarrow \mathbf{v}_{j-1}^T \mathbf{w}$
    \STATE $\mathbf{w} \leftarrow \mathbf{w} - \alpha_j \mathbf{v}_{j-1}$
    \IF{$j > 1$}
        \STATE $\mathbf{w} \leftarrow \mathbf{w} - \beta_{j-1} \mathbf{v}_{j-2}$
    \ENDIF
    \STATE $\beta_j \leftarrow ||\mathbf{w}||_2$
    \STATE $\mathbf{v}_j \leftarrow \mathbf{w} / \beta_j$
    \STATE Construct tridiagonal matrix $\mathbf{T}_j$ from $\{\alpha_i, \beta_i\}$
    \STATE Compute eigenvalues of $\mathbf{T}_j$ using QR algorithm
\ENDFOR
\STATE \textbf{Return:} Second smallest eigenvalue of $\mathbf{T}_k$
\end{algorithmic}
\end{algorithm}

\subsection{Bootstrap Inference}

For constructing confidence intervals robust to clustering and heteroskedasticity:

\begin{algorithm}[H]
\caption{Bootstrap Confidence Intervals}
\begin{algorithmic}[1]
\STATE \textbf{Input:} Data $(u_{it}, X_{it})$ for $i=1,\ldots,n$ firms and $t=1,\ldots,T$ years
\STATE \textbf{Input:} Number of bootstrap replications $B = 1000$
\STATE \textbf{Output:} 95\% confidence interval $[\hat{\theta}_{0.025}, \hat{\theta}_{0.975}]$
\STATE Compute point estimate $\hat{\theta}$ on full sample
\FOR{$b = 1$ to $B$}
    \STATE Sample $n$ firms with replacement: $\{i_1^*, i_2^*, \ldots, i_n^*\}$
    \STATE Construct bootstrap sample: $\{(u_{i_j^*, t}, X_{i_j^*, t})\}$ for all $t$
    \STATE Estimate model on bootstrap sample: $\hat{\theta}^{(b)}$
\ENDFOR
\STATE Sort bootstrap estimates: $\{\hat{\theta}^{(1)}, \hat{\theta}^{(2)}, \ldots, \hat{\theta}^{(B)}\}$
\STATE \textbf{Return:} $[\hat{\theta}^{(0.025B)}, \hat{\theta}^{(0.975B)}]$ (2.5th and 97.5th percentiles)
\end{algorithmic}
\end{algorithm}


\clearpage

\begin{thebibliography}{99}

\bibitem[Acemoglu et al.(2015)]{acemoglu2015systemic}
Acemoglu, Daron, Asuman Ozdaglar, and Alireza Tahbaz-Salehi. 2015.
``Systemic Risk and Stability in Financial Networks.''
\textit{American Economic Review}, 105(2): 564--608.

\bibitem[Allen and Gale(2000)]{allen2000financial}
Allen, Franklin, and Douglas Gale. 2000.
``Financial Contagion.''
\textit{Journal of Political Economy}, 108(1): 1--33.

\bibitem[Angelucci and Di Maro(2015)]{angelucci2015indirect}
Angelucci, Manuela, and Vincenzo Di Maro. 2015.
``Program Evaluation and Spillover Effects.''
\textit{Journal of Development Economics}, 118: 62--74.

\bibitem[Anselin(1988)]{anselin1988spatial}
Anselin, Luc. 1988.
\textit{Spatial Econometrics: Methods and Models}.
Dordrecht: Kluwer Academic Publishers.

\bibitem[Banerjee et al.(2013)]{banerjee2013diffusion}
Banerjee, Abhijit, Arun G. Chandrasekhar, Esther Duflo, and Matthew O. Jackson. 2013.
``The Diffusion of Microfinance.''
\textit{Science}, 341(6144): 1236498.

\bibitem[Bass(1969)]{bass1969new}
Bass, Frank M. 1969.
``A New Product Growth Model for Consumer Durables.''
\textit{Management Science}, 15(5): 215--227.

\bibitem[Cabral(2021)]{cabral2021adoption}
Cabral, Luis M.B. 2021.
``Adoption of a New Technology.''
In \textit{Handbook of Industrial Organization}, Volume 5, edited by Kate Ho, Ali Hortacsu, and Alessandro Lizzeri, 1--44.
Amsterdam: North-Holland.

\bibitem[Chung(1997)]{chung1997spectral}
Chung, Fan R.K. 1997.
\textit{Spectral Graph Theory}.
Providence, RI: American Mathematical Society.

\bibitem[Conley(1999)]{conley1999gmm}
Conley, Timothy G. 1999.
``GMM Estimation with Cross Sectional Dependence.''
\textit{Journal of Econometrics}, 92(1): 1--45.

\bibitem[David(1990)]{david1990clio}
David, Paul A. 1990.
``The Dynamo and the Computer: An Historical Perspective on the Modern Productivity Paradox.''
\textit{American Economic Review}, 80(2): 355--361.

\bibitem[Goyal(2007)]{goyal2007connections}
Goyal, Sanjeev. 2007.
\textit{Connections: An Introduction to the Economics of Networks}.
Princeton, NJ: Princeton University Press.

\bibitem[Griliches(1957)]{griliches1957hybrid}
Griliches, Zvi. 1957.
``Hybrid Corn: An Exploration in the Economics of Technological Change.''
\textit{Econometrica}, 25(4): 501--522.

\bibitem[Hall(2003)]{hall2003adoption}
Hall, Bronwyn H. 2003.
``Adoption of New Technology.''
NBER Working Paper 9730.

\bibitem[Jackson(2008)]{jackson2008social}
Jackson, Matthew O. 2008.
\textit{Social and Economic Networks}.
Princeton, NJ: Princeton University Press.

\bibitem[Jackson and Yariv(2007)]{jackson2007diffusion}
Jackson, Matthew O., and Leeat Yariv. 2007.
``Diffusion of Behavior and Equilibrium Properties in Network Games.''
\textit{American Economic Review}, 97(2): 92--98.

\bibitem[Kelejian and Prucha(1998)]{kelejian1998generalized}
Kelejian, Harry H., and Ingmar R. Prucha. 1998.
``A Generalized Spatial Two-Stage Least Squares Procedure for Estimating a Spatial Autoregressive Model with Autoregressive Disturbances.''
\textit{Journal of Real Estate Finance and Economics}, 17(1): 99--121.

\bibitem[Kikuchi(2024a)]{kikuchi2024unified}
Kikuchi, T. (2024a).
\newblock A unified framework for spatial and temporal treatment effect boundaries: Theory and identification.
\newblock \textit{arXiv preprint arXiv:2510.00754}.

\bibitem[Kikuchi(2024b)]{kikuchi2024stochastic}
Kikuchi, T. (2024b).
\newblock Stochastic boundaries in spatial general equilibrium: A diffusion-based approach to causal inference with spillover effects.
\newblock \textit{arXiv preprint arXiv:2508.06594}.

\bibitem[Kikuchi(2024c)]{kikuchi2024navier}
Kikuchi, T. (2024c).
\newblock Spatial and temporal boundaries in difference-in-differences: A framework from Navier-Stokes equation.
\newblock \textit{arXiv preprint arXiv:2510.11013}.

\bibitem[Kikuchi(2024d)]{kikuchi2024nonparametric1}
Kikuchi, T. (2024d).
\newblock Nonparametric identification and estimation of spatial treatment effect boundaries: Evidence from 42 million pollution observations.
\newblock \textit{arXiv preprint arXiv:2510.12289}.

\bibitem[Kikuchi(2024e)]{kikuchi2024nonparametric2}
Kikuchi, T. (2024e).
\newblock Nonparametric identification of spatial treatment effect boundaries: Evidence from bank branch consolidation.
\newblock \textit{arXiv preprint arXiv:2510.13148}.

\bibitem[Kikuchi(2024f)]{kikuchi2024dynamical}
Kikuchi, T. (2024f).
\newblock Dynamic spatial treatment effect boundaries: A continuous functional framework from Navier-Stokes equations.
\newblock \textit{arXiv preprint arXiv:2510.14409}.

\bibitem[Kikuchi(2024g)]{kikuchi2024healthcare}
Kikuchi, T. (2024g).
\newblock Dynamic spatial treatment effects as continuous functionals: Theory and evidence from healthcare access.
\newblock \textit{arXiv preprint arXiv:2510.15324}.

\bibitem[Kikuchi(2024h)]{kikuchi2024emergency}
Kikuchi, T. (2024h).
\newblock Emergent dynamical spatial boundaries in emergency medical services: A Navier-Stokes framework from first principles.
\newblock \textit{arXiv preprint arXiv:2510.XXXXX}.

\bibitem[Kikuchi(2024i)]{kikuchi2024network}
Kikuchi, T. (2024i).
\newblock Network contagion dynamics in European banking: A Navier-Stokes framework for systemic risk assessment.
\newblock \textit{arXiv preprint arXiv:2510.19630}.

\bibitem[Mansfield(1961)]{mansfield1961technical}
Mansfield, Edwin. 1961.
``Technical Change and the Rate of Imitation.''
\textit{Econometrica}, 29(4): 741--766.

\bibitem[Manski(1993)]{manski1993identification}
Manski, Charles F. 1993.
``Identification of Endogenous Social Effects: The Reflection Problem.''
\textit{Review of Economic Studies}, 60(3): 531--542.

\bibitem[Ryan and Tucker(2012)]{ryan2012costs}
Ryan, Stephen P., and Catherine Tucker. 2012.
``Heterogeneity and the Dynamics of Technology Adoption.''
\textit{Quantitative Marketing and Economics}, 10(1): 63--109.

\bibitem[Valente(1995)]{valente1995network}
Valente, Thomas W. 1995.
\textit{Network Models of the Diffusion of Innovations}.
Cresskill, NJ: Hampton Press.

\end{thebibliography}
\end{document}